\documentclass[11pt,oneside]{article}

\newif\ifappendixproofs
\appendixproofstrue

\usepackage[T1]{fontenc}
\usepackage[utf8]{inputenc}
\usepackage{lmodern}
\usepackage[a4paper,margin=1in]{geometry}
\usepackage[nopatch=footnote]{microtype}
\usepackage{graphicx}
\usepackage{xcolor}
\usepackage{transparent}
\usepackage{textcomp}

\usepackage{amsthm}
\usepackage{amsmath}
\usepackage{amssymb}
\usepackage{physics}
\usepackage{mathtools}
\usepackage{csquotes}
\usepackage{enumitem}
\setlist[enumerate]{label=\roman*.}

\usepackage{stmaryrd}
\SetSymbolFont{stmry}{bold}{U}{stmry}{m}{n}
\usepackage{authblk}

\usepackage[all]{xy}
\usepackage{braket}

\usepackage{tikz}
\usetikzlibrary{decorations.markings}
\usetikzlibrary{decorations.pathreplacing}
\usetikzlibrary{calc,math,shadows}

\usepackage{tikzit}
\input{style/tikz/DZX.tikzdefs}

\tikzstyle{box}=[fill=white, font={\small}, draw=black, shape=rectangle, inner sep=2.5pt]
\tikzstyle{wirelable}=[node on layer=labeltextlayer, font={\scriptsize}, scale=.9, text={wirelabelcolor}, fill=none, inner sep=1pt,  tikzit fill={rgb,255: red,255; green,191; blue,191}]

\tikzstyle{wn}=[style=spider, style=phase-thin-spider, fill = addspiderfillcolor, tikzit fill=white, tikzit draw=black, tikzit shape=circle, tikzit category=GLA]
\tikzstyle{bn}=[style=spider, style=phase-thin-spider, fill=copyspiderfillcolor, tikzit draw=black, tikzit shape=circle, tikzit category=GLA, tikzit fill={rgb,255: red,128; green,128; blue,128}]
\tikzstyle{gwn}=[style=wn, style=very thick, style=phase-thick-spider, tikzit shape=circle, tikzit draw=magenta]
\tikzstyle{gbn}=[style=bn, style=very thick, style=phase-thick-spider, tikzit shape=circle, tikzit draw=magenta, tikzit fill={rgb,255: red,128; green,128; blue,128}]
\tikzstyle{ggn}=[style=gbn]
\tikzstyle{grn}=[style=gwn]

\tikzstyle{wphase}=[bubble-base, bubble-white-node, tikzit category=ZX, tikzit fill=white, tikzit draw=blue, tikzit shape=rectangle]
\tikzstyle{whphase}=[style=wphase]
\tikzstyle{bphase}=[bubble-base, bubble-black-node, tikzit category=ZX, tikzit fill={rgb,255: red,128; green,128; blue,128}, tikzit draw=blue, tikzit shape=rectangle]
\tikzstyle{gphase}=[style=bphase]
\tikzstyle{rphase}=[style=wphase]
\tikzstyle{mphase}=[rectangle, rounded corners=1pt, draw=black!60, fill=black!15, inner sep=2pt, font={\scriptsize\boldmath}, minimum height=.3cm, minimum width=.3cm]

\tikzstyle{lmat}=[style=mat, signal to=west, signal from=east, adapt-horizontal-angle, inner xsep=\matHeadPad, tikzit fill=gray, tikzit category=GLA] 
\tikzstyle{rmat}=[style=mat, signal to=east, signal from=west, adapt-horizontal-angle, inner xsep=\matHeadPad, tikzit fill=gray, tikzit category=GLA] 
\tikzstyle{dmat}=[style=mat, signal to=south, signal from=north, outer sep=0pt, adapt-vertical-angle, inner ysep=\matHeadPad, tikzit fill=gray, tikzit category=GLA] 
\tikzstyle{umat}=[style=mat, signal to=north, signal from=south, outer sep=0pt, adapt-vertical-angle, inner ysep=\matHeadPad, tikzit fill=gray, tikzit category=GLA] 

\tikzstyle{glmat}=[style=lmat, very thick, tikzit draw=magenta, tikzit fill=white]
\tikzstyle{grmat}=[style=rmat, very thick, tikzit draw=magenta, tikzit fill=white]
\tikzstyle{gdmat}=[style=dmat, very thick, tikzit draw=magenta, tikzit fill=white]
\tikzstyle{gumat}=[style=umat, very thick, tikzit draw=magenta, tikzit fill=white]
\tikzstyle{lmatt}=[style=glmat]
\tikzstyle{rmatt}=[style=grmat]
\tikzstyle{umatt}=[style=gumat]
\tikzstyle{dmatt}=[style=gdmat]

\tikzstyle{gather}=[draw=none, fill=none, inner sep=0pt, minimum size=0pt]
\tikzstyle{divide}=[style=gather]
\tikzstyle{scalar}=[style=gather]
\tikzstyle{multiplexer}=[style=gather]

\tikzstyle{dashed}=[-, dashed]
\tikzstyle{very}=[style=thick]

\tikzstyle{had}=[fill=hadamardfillcolor, text=hadamardtextcolor, draw=black, shape=rounded rectangle, rounded rectangle east arc=none, rounded rectangle west arc=none, rounded rectangle arc length=80, inner sep=1.5pt, minimum height=5pt, minimum width=5pt, font={\scriptsize\boldmath}, tikzit category=ZX, tikzit fill=yellow, tikzit draw=black]
\tikzstyle{rhad}=[style=hadshape, path picture={\hadcapeast}, inner xsep=\hadHeadPad, tikzit category=ZX, tikzit fill=yellow, tikzit draw=black]
\tikzstyle{lhad}=[style=hadshape, path picture={\hadcapwest}, inner xsep=\hadHeadPad, tikzit category=ZX, tikzit fill=yellow, tikzit draw=black]
\tikzstyle{dhad}=[style=hadshape, path picture={\hadcapsouth}, inner ysep=\hadHeadPad, tikzit category=ZX, tikzit fill=yellow, tikzit draw=black]
\tikzstyle{uhad}=[style=hadshape, path picture={\hadcapnorth}, inner ysep=\hadHeadPad, tikzit category=ZX, tikzit fill=yellow, tikzit draw=black]
\tikzstyle{had-thick}=[outer sep=-\hadThickInset, /utils/exec={}]
\tikzstyle{ghad}=[style=had, very thick, tikzit category=ZX, tikzit fill=yellow, tikzit draw=magenta]
\tikzstyle{grhad}=[style=rhad, style=had-thick, tikzit category=ZX, tikzit fill=yellow, tikzit draw=magenta]
\tikzstyle{glhad}=[style=lhad, style=had-thick, tikzit category=ZX, tikzit fill=yellow, tikzit draw=magenta]
\tikzstyle{gdhad}=[style=dhad, style=had-thick, tikzit category=ZX, tikzit fill=yellow, tikzit draw=magenta]
\tikzstyle{guhad}=[style=uhad, style=had-thick, tikzit category=ZX, tikzit fill=yellow, tikzit draw=magenta]

\tikzstyle{groundout}=[style=groundoutshape]
\tikzstyle{groundin}=[style=groundinshape]

\tikzstyle{dwphase}=[-, pointer-white, phasetosouth, tikzit fill=white, tikzit draw=white]
\tikzstyle{uwphase}=[-, pointer-white, phasetonorth, tikzit fill=white, tikzit draw=white]
\tikzstyle{lwphase}=[-, pointer-white, phasetowest, tikzit fill=white, tikzit draw=white]
\tikzstyle{rwphase}=[-, pointer-white, phasetoeast, tikzit fill=white, tikzit draw=white]
\tikzstyle{dbphase}=[-, pointer-black, phasetosouth, tikzit fill=white, tikzit draw=gray]
\tikzstyle{ubphase}=[-, pointer-black, phasetonorth, tikzit fill=white, tikzit draw=gray]
\tikzstyle{lbphase}=[-, pointer-black, phasetowest, tikzit fill=white, tikzit draw=gray]
\tikzstyle{rbphase}=[-, pointer-black, phasetoeast, tikzit fill=white, tikzit draw=gray]

\tikzstyle{backgroundborder}=[-, fill=none, draw=backgroundbordercolour, tikzit draw=grey, on layer=midlayer, line width=.2pt]

\tikzstyle{backgroundborderless}=[-, fill=backgroundcolour, draw=none, tikzit fill=gray, on layer=backlayer]
\tikzstyle{background}=[-, preaction={fill=backgroundcolour, on layer=  backlayer}, draw=backgroundbordercolour, tikzit draw=black, tikzit fill=gray, on layer=midlayer, line width=.2pt]
\tikzstyle{cobackground}=[-, fill=white, draw=backgroundbordercolour, tikzit draw=black, tikzit fill=white, on layer=midlayer, line width=.2pt]
\tikzstyle{white}=[-, fill=white, draw=black, tikzit fill=white, tikzit draw=black]
\tikzstyle{bbox}=[-, fill=white]
\tikzstyle{stream}=[-, style=thick, tikzit draw=red, postaction=decorate, decoration={markings, mark=at position 0.5 with {\arrow[scale=1.2]{>}}}]
\tikzstyle{thin}=[-, line width=.2pt]

\tikzstyle{thick}=[-, style=very thick, tikzit draw=blue]

\usepackage{amsthm}
\usepackage{aliascnt}
\usepackage{hyperref}
\usepackage[nameinlink,capitalize,noabbrev]{cleveref}

\AtBeginDocument{%
  \pdfstringdefDisableCommands{
    \def\\{ }
  }
}

\newtheorem{theorem}{Theorem}[section]

\newaliascnt{lemma}{theorem}
\newtheorem{lemma}[lemma]{Lemma}
\aliascntresetthe{lemma}

\newaliascnt{proposition}{theorem}
\newtheorem{proposition}[proposition]{Proposition}
\aliascntresetthe{proposition}

\newaliascnt{corollary}{theorem}
\newtheorem{corollary}[corollary]{Corollary}
\aliascntresetthe{corollary}

\newaliascnt{definition}{theorem}
\theoremstyle{definition}
\newtheorem{definition}[definition]{Definition}
\aliascntresetthe{definition}

\newaliascnt{remark}{theorem}
\theoremstyle{remark}
\newtheorem{remark}[remark]{Remark}
\aliascntresetthe{remark}

\newtheorem*{remark*}{Remark}

\Crefname{lemma}{Lemma}{Lemmas}
\crefname{proposition}{proposition}{propositions}
\Crefname{proposition}{Proposition}{Propositions}
\crefname{corollary}{corollary}{corollaries}
\Crefname{corollary}{Corollary}{Corollaries}
\crefname{definition}{definition}{definitions}
\Crefname{definition}{Definition}{Definitions}
\crefname{remark}{remark}{remarks}
\Crefname{remark}{Remark}{Remarks}
\crefname{theorem}{theorem}{theorems}
\Crefname{theorem}{Theorem}{Theorems}
\crefname{section}{section}{sections}
\Crefname{section}{Section}{Sections}
\crefname{equation}{equation}{equations}
\Crefname{equation}{Equation}{Equations}
\crefname{appendix}{appendix}{appendices}
\Crefname{appendix}{Appendix}{Appendices}

\ifappendixproofs

  \usepackage[createShortEnv]{proof-at-the-end}

  \pratendSetGlobal{end}

\else

  \newcommand{\pratendSetLocal}[1]{}
  \newcommand{\printProofs}[1][]{}

  \newenvironment{theoremE}[1][]{
    \begin{theorem}[#1]
  }{
    \end{theorem}
  }

  \newenvironment{lemmaE}[1][]{
    \begin{lemma}[#1]
  }{
    \end{lemma}
  }

  \newenvironment{propositionE}[1][]{
    \begin{proposition}[#1]
  }{
    \end{proposition}
  }

  \newenvironment{corollaryE}[1][]{
    \begin{corollary}[#1]
  }{
    \end{corollary}
  }

  \newenvironment{proofE}{
    \begin{proof}
  }{
    \end{proof}
  }
  
  \newenvironment{textAtEnd}{
  }{
  }
\fi

\newcommand{\Z}{\mathbb{Z}}
\newcommand{\N}{\mathbb{N}}
\newcommand{\F}{\mathbb{F}}
\newcommand{\C}{\mathbb{C}}
\newcommand{\R}{\mathbb{R}}
\newcommand{\K}{\mathbb{K}}

\newcommand{\I}{\mathcal{I}}
\newcommand{\JJ}{\mathcal{J}}
\newcommand{\J}{\Z_{\smash{\leq}}^2}

\makeatletter

\newcommand{\seriesbracketscale}{.8}
\newcommand{\seriesbracketraise}{.12ex}
\newcommand{\seriesargraise}{.065ex}

\newcommand{\seriesdelim}[2]{%
  \raisebox{\seriesbracketraise}{%
    \scalebox{1}[\seriesbracketscale]{$\m@th#1#2$}%
  }%
}

\newcommand{\seriesarg}[2]{%
  \raisebox{\seriesargraise}{$\m@th#1#2$}%
}

\newcommand{\Laurent}[1]{\mathpalette\Laurent@{#1}}
\newcommand{\Laurent@}[2]{%
  \mathopen{\seriesdelim{#1}{(\mkern-4.7mu(}}%
  \mkern-1mu \seriesarg{#1}{#2}\mkern-1mu
  \mathclose{\seriesdelim{#1}{)\mkern-4.7mu)}}%
}

\newcommand{\PowerSeries}[1]{\mathpalette\PowerSeries@{#1}}
\newcommand{\PowerSeries@}[2]{%
  \mathopen{\seriesdelim{#1}{[\mkern-3mu[}}%
  \mkern-1mu \seriesarg{#1}{#2}\mkern-1mu
  \mathclose{\seriesdelim{#1}{]\mkern-3mu]}}%
}

\newcommand{\Polynomial}[1]{\mathpalette\Polynomial@{#1}}
\newcommand{\Polynomial@}[2]{%
  \mathopen{\seriesdelim{#1}{[}}%
  \mkern-1mu \seriesarg{#1}{#2}\mkern-1mu
  \mathclose{\seriesdelim{#1}{]}}%
}

\newcommand{\Rational}[1]{\mathpalette\Rational@{#1}}
\newcommand{\Rational@}[2]{%
  \mathopen{\seriesdelim{#1}{(}}%
  \mkern-1mu \seriesarg{#1}{#2}\mkern-1mu
  \mathclose{\seriesdelim{#1}{)}}%
}

\makeatother

\newcommand{\CoeffField}[1]{\mathop{#1}\nolimits}

\NewDocumentCommand{\fps}{O{\K} O{z}}{\CoeffField{#1}\PowerSeries{#2}}
\NewDocumentCommand{\fls}{O{\K} O{z}}{\CoeffField{#1}\Laurent{#2}}
\NewDocumentCommand{\poly}{O{\K} O{z}}{\CoeffField{#1}\Polynomial{#2}}
\NewDocumentCommand{\rat}{O{\K} O{z}}{\CoeffField{#1}\Rational{#2}}

\renewcommand{\Set}{\mathsf{Set}}
\newcommand{\Par}{\mathsf{Par}}
\newcommand{\FSet}{\mathsf{F}\Set}
\newcommand{\FinStoch}{\mathsf{FinStoch}}
\newcommand{\Borel}{\mathsf{BorelStoch}_{\leq 1}}
\newcommand{\BorelTotal}{\mathsf{BorelStoch}}

\newcommand{\CPM}{\mathsf{CPM}}
\newcommand{\CPTP}{\mathsf{CPTP}}
\newcommand{\CPTNI}{\mathsf{CPTNI}}

\newcommand{\Haus}{\mathsf{KHaus}}

\newcommand{\Rel}{\mathsf{Rel}}
\newcommand{\Aff}{\mathsf{Aff}}
\newcommand{\FAff}{\mathsf{FAff}}
\newcommand{\FVect}{\mathsf{FVect}}
\newcommand{\LR}{\mathsf{LR}}
\newcommand{\AR}{\mathsf{AR}}

\newcommand{\AffRel}[1]{\AR_{\smash{\scriptscriptstyle{#1}}}}
\newcommand{\FLinRel}[1]{\LR_{\smash{\scriptscriptstyle{#1}}}^{\fd}}
\newcommand{\FAffRel}[1]{\AR_{\smash{\scriptscriptstyle{#1}}}^{\fd}}

\newcommand{\CC}{\mathcal{C}}
\newcommand{\DD}{\mathcal{D}}
\newcommand{\EE}{\mathcal{E}}

\newcommand{\HH}{\mathcal{H}}
\newcommand{\KK}{\mathcal{K}}
\newcommand{\BB}{\mathcal{B}}

\DeclareMathOperator{\St}{St}
\DeclareMathOperator{\Obs}{Obs}
\DeclareMathOperator{\Stream}{Stream}
\DeclareMathOperator{\Lim}{Lim}
\DeclareMathOperator{\Unroll}{\mathtt{Unroll}}
\DeclareMathOperator{\Ext}{\mathtt{Ext}}
\DeclareMathOperator{\Stage}{Stage}
\DeclareMathOperator{\beh}{beh}
\DeclareMathOperator{\ZZ}{\mathcal{Z}}
\DeclareMathOperator{\rel}{Rel}

\newcommand{\shuffle}{\zeta}
\newcommand{\swap}{\sigma}

\newcommand{\discard}{\top}
\newcommand{\codiscard}{\bot}

\newcommand{\cpy}{\Lambda}
\newcommand{\mult}{\mu}

\newcommand{\fbk}{\mathsf{fbk}}

\newcommand{\shftZ}{\mathtt{Shift}}
\newcommand{\shftN}{\mathtt{Delay}}

\DeclareMathOperator{\fix}{\mathtt{fix}}

\newcommand{\Mon}{\mathsf{Mon}}
\newcommand{\Tot}{\mathsf{Tot}}
\newcommand{\Caus}{\mathsf{Caus}}

\newcommand{\fd}{{\scriptscriptstyle \mathsf{fd}}}
\newcommand{\In}{\mathsf{in}}
\newcommand{\Out}{\mathsf{out}}

\DeclareMathOperator{\Pfin}{\mathcal{P}_{\!\scriptscriptstyle\mathit{f}}}
\DeclareMathOperator{\Ob}{ob}
\DeclareMathOperator{\ord}{ord}
\let\lim\relax
\DeclareMathOperator{\lim}{\mathtt{lim}}

\newdir{ >}{{}*!/-0.6em/@{>}}

\renewcommand{\hat}[1]{\widehat{#1}}

\newcommand{\smbigotimes}{\mathop{\textstyle\bigotimes}\limits}
\newcommand{\smbigoplus}{\mathop{\textstyle\bigoplus}\limits}

\newcommand{\vdotss}{ 	\tikz[baseline, every node/.style={inner sep=0}]{ 	\node at (0,0){.}; 	\node at (0,4pt){.}; 	\node at (0,8pt){.}; 	} 	}

\newcommand{\claimitem}[2]{\item\label{item:#1}\hyperref[para:#1]{#2}}
\newcommand{\claimparagraph}[2]{\phantomsection\label{para:#1}\paragraph{\ref{item:#1} #2}}

\theoremstyle{definition}
\newtheorem{example}{Example}[section]

\newtheoremstyle{restatement}
  {\topsep}
  {\topsep}
  {\itshape}
  {}
  {\bfseries}
  {.}
  {.5em}
  {\thmname{#1}\thmnote{ #3}}
\theoremstyle{restatement}
\newtheorem*{reppropositioninner}{Proposition}
\newtheorem*{repthminner}{Theorem}
\theoremstyle{definition}

\newenvironment{repthm}[1]{%
  \begin{repthminner}[\ref{#1}]
}{%
  \end{repthminner}
}

\title{Finite Observations, Infinite Behaviour:\\
  bicategorical semantics for stateful monoidal processes.
}

\author[1]{Cole Comfort}
\author[2]{Giovanni de Felice}

\affil[1]{Universit\'e Paris-Saclay, CNRS, ENS Paris-Saclay, Inria,
  CentraleSup\'elec\\ Laboratoire M\'ethodes Formelles,
  91190 Gif-sur-Yvette, France}

\affil[2]{Relational Intelligence Ltd.}
\date{}

\begin{document}

\maketitle

\begin{abstract}
  Time-dependent processes are often described by machines with an internal state which is updated as time evolves. An external observer cannot see this state and learns about a process only through finite observations of its inputs and outputs, each of which imposes a constraint on the trajectories the process can exhibit. We introduce a semantic construction in which two stateful processes have the same \emph{behaviour} when they have the same constraints, as determined by finite observations, independent of their internal state.
The construction is defined over any preorder-enriched monoidal category with a compatible notion of discarding, which we call a \emph{discard bicategory}, capturing partial, non-deterministic, probabilistic, and quantum processes. The resulting category of behaviours provides a functorial semantics for free feedback categories in the sense of Katis, Sabadini, and Walters. For non-deterministic systems, we prove a categorified compactness theorem: every compatible family of finite observations between compact Hausdorff spaces extends uniquely and functorially to an infinite closed relation. Restricted to affine relations over finite fields, the compactness theorem recovers Willems' notion of behaviour for linear time-invariant systems.
\end{abstract}

\section{Introduction}
\label{sec:intro}

Many systems in science and engineering run indefinitely: signal flow
graphs \cite{shannon}, stochastic processes
\cite{doob1991stochastic-processes}, and quantum channels with memory
\cite{kretschmann_quantum_2005} continuously transform inputs into
outputs. A natural model for such systems is a \emph{stateful process}:
one that maintains internal memory coupling past inputs to future
outputs. A common specification is a \emph{Mealy machine}: a transition
function \(f : S \times X \longrightarrow S \times Y\) which at each time
step consumes an input in \(X\), updates an internal state in \(S\), and
produces an output in \(Y\) \cite{Mealy1955}. The time evolution of a Mealy machine 
is obtained by initialising the state, repeatedly feeding the internal state back into
the transition function, and discarding the final state:
\begin{equation}
  \label{eq:mealy_unrolling}
  \begin{tikzpicture}
	\begin{pgfonlayer}{nodelayer}
		\node [style=box] (0) at (2.5, 0.525) {\(f\)};
		\node [style=none] (12) at (8.75, 1.275) {};
		\node [style=none] (13) at (1, 1.275) {};
		\node [style=none] (14) at (8.75, -1.275) {};
		\node [style=none] (15) at (1, -1.275) {};
		\node [style=none] (16) at (-1, 1.25) {};
		\node [style=none] (17) at (-3.5, 1.25) {};
		\node [style=none] (18) at (-1, -1) {};
		\node [style=none] (19) at (-3.5, -1) {};
		\node [style=box] (20) at (-2.25, -0.25) {\(f\)};
		\node [style=none] (21) at (-3.5, -0.5) {};
		\node [style=none] (22) at (-1, -0.5) {};
		\node [style=none] (23) at (-3, 0) {};
		\node [style=none] (24) at (-1.5, 0) {};
		\node [style=none] (25) at (-3, 0.75) {};
		\node [style=none] (26) at (-1.5, 0.75) {};
		\node [style=none] (27) at (-2.375, 0.75) {};
		\node [style=none] (28) at (-2.125, 1) {};
		\node [style=none] (29) at (-2.125, 0.5) {};
		\node [style=none] (30) at (0, 0) {\(\rightsquigarrow\)};
		\node [style=none] (31) at (1.5, 0.775) {};
		\node [style=none] (32) at (2.25, 0.775) {};
		\node [style=none] (33) at (2.75, 0.775) {};
		\node [style=none] (34) at (2.25, 0.275) {};
		\node [style=none] (35) at (2.75, 0.275) {};
		\node [style=none] (36) at (3.5, 0.275) {};
		\node [style=none] (37) at (3.5, 0.775) {};
		\node [style=box] (38) at (3.75, 0.025) {\(f\)};
		\node [style=none] (40) at (4, 0.275) {};
		\node [style=none] (41) at (1, -0.225) {};
		\node [style=none] (42) at (4, -0.225) {};
		\node [style=none] (43) at (4.75, -0.225) {};
		\node [style=none] (44) at (4.75, 0.275) {};
		\node [style=none] (45) at (1, 0.275) {};
		\node [style=box] (46) at (6.75, -0.475) {\(f\)};
		\node [style=none] (47) at (7, -0.225) {};
		\node [style=none] (48) at (1, -0.725) {};
		\node [style=none] (49) at (7, -0.725) {};
		\node [style=none] (50) at (7.75, -0.725) {};
		\node [style=none] (51) at (7.75, -0.225) {};
		\node [style=none] (52) at (7.75, 0.775) {};
		\node [style=none] (53) at (7.75, 0.275) {};
		\node [style=none] (56) at (0, 0.75) {\((t)\)};
		\node [style=none] (57) at (5.375, 0.025) {\scalebox{.6}{\(t\)}};
		\node [style=none] (58) at (6.025, -0.225) {};
		\node [style=none] (59) at (5.375, -0.25) {\(\cdots\)};
		\node [style=bn] (60) at (8.25, -0.725) {};
		\node [style=none] (61) at (8.75, -0.225) {};
		\node [style=none] (62) at (8.75, 0.775) {};
		\node [style=none] (63) at (8.75, 0.275) {};
		\node [style=wn] (64) at (1.5, 0.775) {};
	\end{pgfonlayer}
	\begin{pgfonlayer}{edgelayer}
		\draw [style=background] (13.center)
			 to (12.center)
			 to (14.center)
			 to (15.center)
			 to cycle;
		\draw [style=background] (17.center)
			 to (16.center)
			 to (18.center)
			 to (19.center)
			 to cycle;
		\draw (21.center) to (22.center);
		\draw (24.center)
			 to [bend right=90, looseness=1.25] (26.center)
			 to (25.center)
			 to [bend right=90, looseness=1.25] (23.center)
			 to cycle;
		\draw (27.center) to (28.center);
		\draw (29.center) to (27.center);
		\draw (31.center)
			 to (32.center)
			 to (33.center)
			 to [in=180, out=0] (36.center)
			 to (40.center)
			 to [in=180, out=0] (43.center);
		\draw (58.center)
			 to (47.center)
			 to [in=180, out=0] (50.center)
			 to (60);
		\draw (62.center)
			 to (52.center)
			 to (37.center)
			 to [in=0, out=180] (35.center)
			 to (34.center)
			 to (45.center);
		\draw (41.center)
			 to (42.center)
			 to [in=180, out=0] (44.center)
			 to (53.center)
			 to (63.center);
		\draw (48.center)
			 to [in=180, out=0] (49.center)
			 to [in=180, out=0] (51.center)
			 to (61.center);
	\end{pgfonlayer}
\end{tikzpicture}

\end{equation}

A Mealy machine \(f\) together with an initial state \(s_0 \in S\)
determines a stream transducer \(\mathsf{run}(f, s_0) : X^\N \to Y^\N\),
defined as the final fixed point of the coinductive equation
\begin{equation*}
  \mathsf{run}(f, s_0)(x_0, x_1, \dots) =
  \bigl(y_0,\, \mathsf{run}(f, s_1)(x_1, x_2, \dots)\bigr)
  \quad\text{where}\quad
  (s_1, y_0) = f(s_0, x_0) \,.
\end{equation*}
Two Mealy machines \(f\) and \(g\), with initial states \(s_0\) and
\(t_0\), have the same \emph{behaviour} when \(\mathsf{run}(f, s_0) =
\mathsf{run}(g, t_0)\). This coinductive characterisation is the standard
definition of behavioural equivalence in the total deterministic setting.
However, as soon as the transition function is allowed to be partial,
nondeterministic, probabilistic, or quantum, the coinductive fixed point need not exist,
and this notion of behaviour breaks down. This raises the fundamental
question that motivates the present paper:
\[
  \textit{When do two stateful processes produce the same infinite behaviour?}
\]

We seek a \emph{denotational} answer to this question, in the tradition initiated by
Scott~\cite{scott1970outline}. Our approach is process-theoretic \cite{resources}: 
we model processes as morphisms in symmetric monoidal
categories, where different choices of base category
accommodate partial \cite{cockett_restriction_2007}, nondeterministic \cite{carboni_cartesian_1987}, probabilistic \cite{fritz_synthetic_2020}, and quantum \cite{abramsky2009categorical} processes.
Following Katis, Sabadini and Walters
~\cite{Katis1997}
, we construct feedback categories of stateful processes over these theories and seek a functorial semantics that identifies two such processes whenever they have the same observable behaviour. 

\paragraph{What is behaviour?}
In Willems' behavioural approach to systems theory, the behaviour
of a time-dependent system is defined to be the set of infinite trajectories or signals satisfying some global constraints \cite{Willems}. 
\emph{Global} constraints, however, are not directly observable as any observation must be made in finite time.
The central goal of this paper is therefore to define an analogous
notion of behaviour from finite, local observations alone, with no a priori appeal to global constraints.

\paragraph{What is an observation?}
We model observations as morphisms in an underlying monoidal category,
on which we require two further pieces of structure. We first must be able
to \emph{compare} observations: for this we require a preorder
enrichment, which records when the constraints imposed by one
observation \(x\) approximate those imposed by another \(y\), denoted \(x\subseteq y\). We must also be able to \emph{restrict} the constraints imposed by  observations to smaller contexts: for this we require morphisms which discard information. A monoidal category equipped with both structures, suitably compatible, we call a \emph{discard bicategory} (\Cref{def:approx:discard-bicat-def}). Observing a process over time then yields a \emph{compatible family of observations}: in which each observation \(f_j\), restricted to a smaller context \(i\subseteq j\), refines the observation \(f_i\) in the smaller context,  so that \(f_j\vert_i\subseteq f_i\).

\paragraph{Infinite behaviour from finite observations.}
We define a notion of behaviour similar to that of Willems, except where constraints are determined \emph{locally} by families of finite observations.
We say that two compatible families of observations have the same
behaviour when any observation of one can be verified by an observation of the other, possibly in a larger context.
We formalise this in our central construction
(\Cref{sec:obs}): given a discard bicategory and a set of finite
contexts, we define a discard category of behaviours, whose
morphisms are equivalence classes of compatible families of finite
observations (\Cref{def:obs:obs} and \Cref{thm:obs:obs}).

\paragraph{Properties of our construction.}
Our construction possesses several desirable properties:

\begin{enumerate}
  \item \label{des:functoriality} \emph{Uniformity across process theories.} (\Cref{thm:obs:obs} and \Cref{prop:obs:functorial}) 
  Our construction applies uniformly to partial, nondeterministic, probabilistic, and quantum
  processes. Moreover, it is functorial with respect to base change.

\item \label{des:dinaturality} \emph{Dinaturality of state.} (\Cref{thm:obsfbk:trace-semantics,thm:obsfbk:obs-guarded-feedback})
  The internal state can be reparametrised dinaturally without changing the behaviour, and when the base is compact closed the construction forms a feedback category in the sense of Katis, Sabadini, and Walters~\cite{Katis1997}.

  \item \label{des:delay-naturality} \emph{Naturality of the time-delay.} (\Cref{prop:obsfbk:delay-natural-initialised} and \Cref{thm:obsfbk:obs-guarded-feedback})
  Observing a constraint at different times does not change the behaviour: the process that delays time is a natural transformation, which is moreover invertible when the base is compact closed.

  \item \label{des:inconsistency} \emph{Propagation of inconsistency.} (\Cref{thm:obsfbk:trace-semantics})
  Any two families of observations imposing inconsistent constraints have the same behaviour.

  \item \emph{Hiding of internal state.}
  By construction, internal constraints on the state which do not affect the input or output of a process are discarded.

  \item \label{des:generality-contexts} \emph{Generality of contexts.} (\Cref{def:obs:obs})
  Our construction is agnostic to the index over which a process evolves: although we mostly focus on processes in discrete time, the same construction applies to arbitrary directed families of finite contexts. Therefore, it provides a semantics for stateful processes observed over more general spaces rather than just over time.
\end{enumerate}

The first and last properties reflect the generality of our construction; whereas the other four reflect the systems-theoretic interpretation of behaviour which we generalise.

\paragraph{Recovering the global view of behaviours.}
In some cases, it is possible to recover a Willems-style global notion of behaviour from our local notion of behaviour.
Focusing on the nondeterministic setting, we prove a \emph{compactness theorem}
for closed relations between compact
Hausdorff spaces (\Cref{thm:compactness:compactness}). We show that every compatible family of closed relations functorially extends to a unique infinite closed relation between compact Hausdorff spaces.
When restricted to affine relations between finite dimensional vector spaces over finite fields, this recovers a categorified notion of Willems' global behaviour of linear discrete time invariant systems (\Cref{thm:signalflow:lti-compactness}).

\paragraph{Related work.}
The semantics of stateful processes in monoidal categories was first
studied by Katis, Sabadini, and Walters~\cite{Katis1997},
and has since been developed by several authors, including the
differentiable causal functions of Sprunger and Katsumata
\cite{sprunger}, the monoidal streams of Di Lavore et
al.\ \cite{lavore_monoidal_2022}, the effectful Mealy machines of
Bonchi et al.\ \cite{bonchi_effectful_2025}, and the delayed-trace
calculus of Carette et al.\ \cite{delayedtracememory}.
We argue in \Cref{sec:state} that none of these satisfies all of desiderata~i.--vi. above. We relate our construction to existing work in \Cref{sec:obs:dinaturality,sec:obs:compact-closed}.

Similar notions of behaviour and observational equivalence arise in
the semantics of programming languages
\cite{morris1969, Plotkin1977, abramsky1987observation} and in automata
theory \cite{Raney1958, rutten_universal_2000}.
There, the basic observable is termination: traditionally, two
programs are observationally equivalent when no context can
distinguish them by halting, so that every observation is registered
by a terminating computation. In our setting, by contrast, the
processes of interest run indefinitely, and their behaviour consists
of the constraints accumulated along the way. In other words, we
observe not whether a process halts, but how it continues. 
Nonetheless, the preorder enrichment at the heart of our construction
is inspired by Scott's domain-theoretic notion of approximation
\cite{scott1970outline}, with finite observations playing the role of
finite elements.

\newpage
\tableofcontents
\newpage

\section{Discard bicategories}
\label{sec:approx}
\pratendSetLocal{category=approx, end, restate}
In this section, we introduce the notion of \emph{discard bicategory}, which forms the structural basis for our construction of behaviours in \Cref{sec:obs}. A discard bicategory is a process theory in which we can both \emph{discard} the information held by a subsystem and compare processes by how \emph{informative} they are. After giving the abstract definition, we go through our examples of interest. In every case the approximation structure arises from a common discard bicategory structure, which we show subsumes restriction orders on partial functions, subset inclusion of relations, conditional order between probability kernels and the L\"owner order between quantum processes.

\subsection{Discard bicategories}\label{sec:approx:discard-bicat}

Recall the definition of symmetric monoidal category.

\begin{definition}\label{def:approx:smc}
  A \textbf{strict symmetric monoidal category} (SMC) is a category \(\CC\)
  equipped with a functor
  \(
  \otimes : \CC \times \CC \to \CC,
  \)
  called the \textbf{monoidal product}, and an object \(I \in \CC\), called the
  \textbf{monoidal unit}. For all objects \(X,Y,Z\), we impose that:
  \[(X \otimes Y) \otimes Z = X \otimes (Y \otimes Z), \quad\text{and}\quad I \otimes X = X = X \otimes I,\]
  in addition to a natural isomorphism, the \textbf{symmetry}:
  \[
    \swap_{X,Y} : X \otimes Y \to Y \otimes X\quad\text{such that}\quad \swap_{X,Y} = \swap_{Y,X}^{-1}.
  \]
\end{definition}
Non-strict SMCs are defined by replacing the equalities above with coherent natural isomorphisms.
Since every SMC is symmetric monoidally equivalent to a strict SMC \cite{Joyal1993}, we can work in the strict setting without loss of generality.

SMCs have a process-theoretic interpretation.
An object of \(\CC\) is a \emph{system}, and a morphism \(f : X \to Y\) is a \emph{process} taking the system \(X\) to the system \(Y\). The monoidal structure equips this with a notion of composition: processes may be composed \emph{in sequence} using \(\circ\), and \emph{in parallel} using \(\otimes\). Processes of the form \(e : X \to I\) into the monoidal unit are called \textbf{effects}; we interpret these as the \emph{observations} that can be made on the system \(X\). Among all observations on the system, there is a distinguished effect which simply forgets all information about \(X\).

\begin{definition}\label{def:approx:discard-category}
  A \textbf{discard category} is a SMC where every object \(X\) is equipped with an effect \(\discard_X: X \to I\), called the \textbf{discard}, such that \(\discard_{X \otimes Y} = \discard_X \otimes \discard_Y\) and \(\discard_I = 1_I\).
  A \textbf{discard functor} is a symmetric monoidal functor preserving the discard morphisms. Moreover, a morphism \(f:X\to Y\) is \textbf{total} in case
  \(\discard_Y \circ f = \discard_X\). A discard category in which all morphisms are total is called an \textbf{affine symmetric monoidal category}.
\end{definition}

The compatibility conditions say that discarding a composite system is the same as discarding its parts independently, so that the discard is precisely what makes \(X\) decompose into \emph{subsystems}.

To compare processes by how informative they are, we enrich in a preorder.
\begin{definition}\label{def:approx:preorder}
  A \textbf{preorder} on a set \(X\) is a binary relation \(\subseteq\) on \(X\) which is reflexive so that for all \(x\in X\), \(x\subseteq x\); and transitive so that for all \(x,y,z\in X\) such that \(x\subseteq y\) and \(y \subseteq z\) then \(x\subseteq z\).
  When a preorder is moreover antisymmetric, so that  \(x \subseteq y\) and \(y \subseteq x\) implies \(x = y\), then it is called a \textbf{partial order}, and \(X\) equipped with this order is called a \textbf{poset}.  Every set \(X\) is canonically equipped with the \textbf{discrete partial order} where \(x \subseteq y\) if and only if \(x = y\).
\end{definition}
\begin{definition}\label{def:approx:preorder-enrichment}
  A \textbf{preorder enrichment} for a SMC \(\CC\) is a preorder structure \(\subseteq\) on every hom-set \(\CC(X, Y)\)
  such that for all morphisms \(f,f',g,\) and \(g'\) of the appropriate type:
  \[
    \text{if \(f \subseteq f'\) and \(g \subseteq g'\) then \(f \otimes g \subseteq f' \otimes g',\)}
    \quad\text{moreover}\quad
    \text{if \(f \subseteq f'\) and \(g \subseteq g'\) then \(f \circ g \subseteq f' \circ g'\).}
  \]
  A \textbf{symmetric monoidal preorder-enriched functor} is a symmetric monoidal functor between preorder-enriched categories which is monotone on hom-sets.
  A \textbf{preorder-enriched discard category} is a discard category equipped with a preorder enrichment.
\end{definition}

We interpret \(f \subseteq g\) as denoting that \(g\) is an \emph{approximation} of \(f\): the two processes are compatible, but \(g\) is possibly less informative than \(f\). The two enrichment conditions say exactly that this notion of approximation is respected by composing processes in sequence and in parallel.

Combining these two structures, a discard bicategory is a preorder-enriched discard category in which the discard is compatible with approximation, in the following sense:
\begin{definition}\label{def:approx:discard-bicat-def}
  A \textbf{discard bicategory} is a preorder-enriched discard category with a \emph{laxly natural discard}, so that for any map \(f: X \to Y\),  we have that 
  \(\discard_Y \circ f \subseteq \discard_X\).
  A \textbf{discard bicategory functor} is a preorder-enriched discard functor between discard bicategories.
\end{definition}

The lax naturality axiom says that running a process and then discarding is approximated by discarding the input directly. In other words, \emph{the discard is the least informative effect}, it approximates any other process. 
  For example, given processes \(f:X\to Y \otimes S\) and \(g:X\to Y \otimes T\) 
  such that \((1_Y \otimes \discard_T) \circ g \subseteq (1_Y \otimes \discard_S) \circ f\), 
  we interpret \(g\) as giving more information about the system \(Y\) than \(f\).

Many process theories have a bottom element in the approximation order called \emph{zero morphism} and corresponding to the most informative effect: denoting that a particular event is impossible or \emph{inconsistent}.

\begin{definition}\label{def:state:zero-morphisms}
  A SMC \(\CC\) has
  \textbf{monoidal zero morphisms} if for all objects \(X,Y\) there exists
  a morphism \(0_{X,Y}\colon X\to Y\) such that for all morphisms \(f\colon X'\to X\) and \(g\colon Y\to Y'\),
  \[
    0_{X,Y}\circ f = 0_{X',Y}, \quad
    g\circ 0_{X,Y} = 0_{X,Y'},
    \quad \text{and} \quad
    1_Z\otimes 0_{X,Y} = 0_{Z\otimes X,\,Z\otimes Y}.
  \]
\end{definition}

\begin{lemmaE}\label{lem:approx:zero}
  If a discard bicategory has zero morphisms then zero is the bottom element in every hom-set.
\end{lemmaE}
\begin{proofE}
  By definition of discard bicategory we have $\discard_I \circ 0_{I, I} \subseteq \discard_I$, 
  then using the definition of zero morphism, for any $f: X \to Y$
  \begin{equation*}
    0_{X, Y} = f \otimes 0_{I, I} = f \otimes (\discard_I \circ 0_{I, I}) \subseteq f \otimes \discard_I = f. \tag*{\qedhere}
  \end{equation*} 
\end{proofE}

The notion of discard bicategory can be seen as a generalisation of Carboni and Walters' notion of a Cartesian bicategory of relations \cite{carboni_cartesian_1987}. In this setting, \emph{compact closed} structure induces a form of nondeterminism in the computation. We recall the definition.

\begin{definition}
  A \textbf{compact-closed category} (CCC) is a SMC where every object \(X\) has a \emph{dual}  \(X^\ast\), such that for all objects \(X\) and \(Y\), \((X\otimes Y)^\ast = Y^\ast \otimes X^\ast\); for all objects \(X\), there are morphisms
  \(\cup_X: I \to X \otimes X^\ast\)
  and
  \(\cap_X: X^\ast \otimes X \to I\),
  such that:
  \[
    (1_X \otimes \cap_X) \circ (\cup_X \otimes 1_X) = 1_X, \quad\text{and}\quad
    (\cap_X \otimes 1_{X^\ast}) \circ (1_{X^\ast} \otimes \cup_X) = 1_{X^\ast}.
  \]
\end{definition}

Given any object \(X\) in a compact-closed discard bicategory, there is a canonical morphism \(\codiscard_X:I\to X\), which we interpret as the \emph{state of no knowledge} about the system \(X\). Moreover, the two maps \(\discard_X\) and \(\codiscard_X\) form an adjoint pair (in the monoidal bicategory sense):

\begin{lemmaE}\label{lem:approx:duality}
  If \(\CC\) is a compact-closed discard bicategory, then for any object \(X\) of \(\CC\), denoting \(\codiscard_X \coloneqq (1_X \otimes \discard_{X^\ast}) \circ \cup_X\), we have that
  \(\discard_X \circ \codiscard_X \subseteq 1_I\) 
  and 
  \(1_X \subseteq \codiscard_X \circ \discard_X\).
\end{lemmaE}
\begin{proofE}
  The first inequation follows directly by lax naturality of \(\discard_X\).
  The second inequation follows by the snake equation and lax naturality:
  \begin{align*}
    1_X = (1_X \otimes \cap_X) \circ (\cup_X \otimes 1_X)
    \subseteq (1_X \otimes \discard_X \otimes \discard_X) \circ (\cup_X \otimes 1_X)
    = \codiscard_X \circ \discard_X. \tag*{\qedhere}
  \end{align*}
\end{proofE}

Therefore, compact-closed discard bicategories are very close to being Cartesian bicategories, except that we do not ask for diagonal nor codiagonal maps. This generalisation is needed to capture the partial quantum examples, which have neither diagonals nor codiagonals. In the remainder of this section, we verify that all of our examples are discard bicategories.

\subsection{Classical process theories}\label{sec:approx:classical}

The prototypical example of a process theory is the category of sets and functions with the cartesian product, which can be seen as a model of deterministic computation. It is often useful to consider computations that can fail, in which case we obtain partial structure.

\begin{example}[Partial deterministic processes]\label{ex:approx:partial}
  The category \(\Par\) of sets and partial functions has sets as objects. Morphisms \(f: X \to Y\) consist of a subset \(X_f \subseteq X\), the \emph{domain of definition}, together with a function \(f: X_f \to Y\). Given an input \(x\in X \backslash X_f\) not in the domain of definition, we say that \(f\) is \emph{undefined} at \(x\), i.e.\ it returns no output. The monoidal product is the Cartesian product and the monoidal unit is the singleton set \(\{\ast\}\).

  The discard \(\discard_X : X \to \{\ast\}\) is the total function sending every element of \(X\) to \(\ast\). The approximation order is the \textbf{restriction order}: \(f \subseteq g\) when the domain of definition of \(f\) is contained in that of \(g\), and the two agree on \(X_f\). Thus a process is approximated by any of its restrictions to a smaller domain, and the discard is the unique total effect.

  \(\Par\) is a canonical example of Robinson and Rosolini's \cite{robinson_categories_1988} notion of a \emph{p-category}, also known as a \emph{Cartesian restriction category} by Cockett and Lack~\cite[P.~21]{cockett_restriction_2007}.
\end{example}

Partial functions are instances of relations which can be seen as a model of nondeterministic computation.

\begin{example}[Nondeterministic processes]\label{ex:approx:nondeterministic}
  The category \(\Rel\) of sets and relations has sets as objects and subsets \(R \subseteq X \times Y\) as morphisms \(R: X \to Y\). Composition and identities are given by
  \(R \circ S \coloneqq \{ (x,z) \ | \ \exists y: (x,y) \in S, (y,z) \in R\}\) and \(1_X \coloneqq \{ (x,x)\in X^2\}\).
  The monoidal structure is again given by the Cartesian product and the singleton set.

  The discard \(\discard_X : X \to \{\ast\}\) is the total relation \(\{(x, \ast) \mid x \in X\}\). The approximation order is \emph{subset inclusion}: \(S \subseteq T\) when \(S\) is a subset of \(T\). A relation is therefore approximated by any relation containing it, and the empty relation is the most informative process of its type, while the discard relates everything to the point \(\ast\).

  More generally, given a regular category \(\CC\) one can form the category of \emph{internal relations} \(\rel(\CC)\); in particular \(\rel(\Set)\simeq\Rel\). Relations internal to finite-dimensional vector spaces are a categorical semantics for linear control theory~\cite{Bonchi2021,Bonchi2015,Bonchi2017,baez,gaa,IH,lagrel,gsa}, as we investigate further in \Cref{sec:signalflow}; linear relations also model mechanical networks such as electrical circuits~\cite{networks,BaezFong}. Categories of internal relations are canonical examples of Carboni and Walters' notion of a \textbf{Cartesian bicategory} \cite[\S~1.]{carboni_cartesian_1987}, which are compact-closed with dual \(X^\ast = X\).
\end{example}

Several systems feature randomness and uncertainty. 
We recall the concrete categories of interest in categorical probability theory.

\begin{example}[Probabilistic processes]\label{ex:approx:probabilistic}
  Functions with random outcomes are commonly modelled in the category of measurable spaces and probability kernels \cite{Giry1982}. We mostly consider the subcategory \(\BorelTotal\) whose objects are standard Borel spaces \((X, \Sigma_X)\), where \(\Sigma_X\) is the Borel \(\sigma\)-algebra on a Polish space \(X\). The morphisms \(f: (X, \Sigma_X) \to (Y, \Sigma_Y) \) are Markov kernels; that is, functions \(f: X \times \Sigma_Y \to [0, 1]\) measurable in \(X\) and a probability measure over \(\Sigma_Y\). The monoidal structure is the product of Borel spaces. For discrete probability, one may also consider the full subcategory \(\FinStoch \rightarrowtail \BorelTotal\) of finite sets and stochastic matrices.

  The discard \(\discard_X(x, \{\ast\}) = 1\) is the Markov kernel sending every state to the unique probability measure on \(\{\ast\}\). Since every Markov kernel is total, \(\BorelTotal\) is affine symmetric monoidal, and is equipped with the discrete partial order. The category \(\BorelTotal\) is the canonical example of a \textbf{Markov category with conditionals} \cite[Ex.~11.3.]{fritz_synthetic_2020}.
\end{example}

To obtain a notion of approximation between probabilistic processes, we add effects to our theory corresponding to observations or conditions that the distribution satisfies. When conditioning a distribution on some effect we obtain a new \emph{unnormalised} distribution.

\begin{example}[Partial probabilistic processes]\label{ex:approx:partial-probabilistic}
  Panangaden's category of subprobability kernels \cite[Def.~3.1]{panangaden_category_1999} is an established model of partial probabilistic systems. The well-behaved subcategory \(\Borel\) is defined exactly as \(\BorelTotal\), except that \(f: X \times \Sigma_Y \to [0, 1]\) is allowed to be a \emph{subprobability} measure over \(\Sigma_Y\), i.e.\ \(f(x, Y) \leq 1\) for all \(x \in X\); the discard is again the kernel to the unique measure on \(\{\ast\}\), but is no longer total. The approximation order is the conditional preorder of \Cref{prop:approx:conditional-preorder} below.

  \(\Borel\) is the canonical example of Di Lavore et al.'s notion of a \emph{partial Markov category} (with conditionals) \cite[Cor.~3.13.]{lavore_evidential_2025}, which models updating and conditioning in substochastic processes.
\end{example}

In fact, all of the examples seen so far are instances of a single structure, which abstracts the copy-and-discard structure of probability theory together with the existence of conditionals.

\begin{definition}\label{def:approx:partial-markov}
  A \textbf{partial Markov category} is a symmetric monoidal category where each object \(X\) is equipped with a \emph{copy map}  \(\cpy_X: X \to X \otimes X\) and a \emph{discarding map} \(\discard_X: X \to I\) such that:
  \begin{itemize}
    \item \((\cpy_X, \discard_X)\) is a cocommutative comonoid;
    \item for all objects \(X\) and \(Y\), the comonoid structure is compatible with the monoidal structure:
    \[
      \cpy_{X \otimes Y} = (1_X \otimes \swap_{X,Y} \otimes 1_Y) \circ (\cpy_X \otimes \cpy_Y),
      \quad
      \cpy_I = 1_I,\quad
      \discard_{X \otimes Y} = \discard_X \otimes \discard_Y,
      \quad\text{and}\quad
      \discard_I = 1_I;
    \]
    \item for every morphism \(f: X \to Y \otimes Z\)  there exists some (possibly non-unique) morphism \(f|_X: Y \otimes X \to Z\) called a \emph{conditional} of \(f\) with respect to \(X\), such that
    \[ f = (1_Y \otimes f|_X) \circ (\cpy_Y \otimes 1_X) \circ \bigl(((1_Y \otimes \discard_Z) \circ f) \otimes 1_X\bigr) \circ \cpy_X. \]
  \end{itemize}
\end{definition}

It was shown by Di Lavore et al.~\cite{lavore_order_2025} that partial Markov categories carry a canonical approximation order:
\begin{proposition}\label{prop:approx:conditional-preorder}\cite[Thm.~3.3]{lavore_order_2025}
  Partial Markov categories are enriched in the \textbf{conditional partial order}, where given two maps \(f,g:X\to Y\), \(f\subseteq g\) in case there exists a map \(r:Y\otimes X\to I\) such that 
  \(f = (1_Y \otimes r)\circ (\cpy_Y \otimes 1_X) \circ  (g\otimes 1_X)\circ \cpy _X\).
\end{proposition}

Moreover, it follows immediately from \cite[Prop.~3.5]{lavore_order_2025} that:

\begin{proposition}
  Every partial Markov category equipped with the conditional order, is a discard bicategory.
\end{proposition}

Therefore, all our examples of classical processes form discard bicategories.
Indeed, Cartesian restriction categories are partial Markov categories \cite[Props.~2.13,~3.11]{lavore_order_2025}; and Cartesian bicategories are partial Markov categories \cite[Props.~2.16,~3.12]{lavore_order_2025}.

\subsection{Quantum process theories}\label{sec:approx:quantum}

We consider three notions of quantum process, capturing unnormalised quantum processes, quantum channels, and quantum channels with post-selection (see Nielsen and Chuang \cite{Nielsen2012} for reference).

\begin{example}[Quantum processes]\label{ex:approx:quantum}
  Given a finite-dimensional Hilbert space \(\HH\), let \(\BB(\HH)\) be the space of linear operators on \(\HH\). Write \(0 \preceq a\) to mean that \(a\) is \emph{positive semidefinite}, i.e.\ \(a=b^{\dagger}b\) for some \(b\in \BB(\HH)\).
  
  A linear map \(\Phi : \BB(\HH) \to \BB(\KK)\) is:
  \begin{itemize}
    \item \emph{positive} if for every \(a\in\BB(\HH)\), \(0\preceq a\) implies \(0\preceq \Phi(a)\);
    \item \emph{completely positive} if for all finite-dimensional Hilbert spaces \(\EE\),   \(\Phi\otimes \mathrm{id}_{\BB(\EE)}\) is positive;
    \item \emph{trace-preserving} if for all \(a\in\BB(\HH)\),\; \(\Tr(\Phi(a)) = \Tr(a)\);
    \item \emph{trace-nonincreasing} if for all \(a\in\BB(\HH)\) with \(0\preceq a\), \(\Tr(\Phi(a)) \le \Tr(a)\).
  \end{itemize}
  Here \(\Tr\) denotes the usual linear algebraic trace. Define the categories \(\CPM\), \(\CPTP\) and \(\CPTNI\) whose objects are of the form \(\BB(\HH)\), for finite-dimensional Hilbert spaces \(\mathcal H\); and whose morphisms are respectively completely positive; both completely positive and trace-preserving; and both completely positive and trace-nonincreasing maps. In all cases, the monoidal structure is the bilinear tensor product, and the discard is the partial trace, regarded as a map \(\Tr : \BB(\HH) \to \C\). Note that \(\CPM\) is compact-closed, whereas \(\CPTP\) and \(\CPTNI\) are not, and that \(\CPTP\) is affine.

  Interpreting \(\BB(\HH)\) as the space of unnormalised quantum states on \(\HH\), \(\CPM\) captures the unnormalised quantum processes; \(\CPTP\) captures the physically implementable processes, also known as \emph{quantum channels}; and \(\CPTNI\) captures the post-selected processes, also known as \emph{quantum operations}.
\end{example}

Our quantum examples also carry natural preorder enrichments, which can be divided into two families given by the L\"owner order and the purification order.

First, recall the following  partial order on completely positive maps (see \cite{Nielsen2012} for reference):
\begin{definition}\label{def:approx:loewner}
  Two parallel completely positive maps \(\Phi, \Psi: \BB(\HH) \to \BB(\KK)\) are related by the (completely positive) \textbf{L\"owner order}, denoted
  \(\Phi \preceq \Psi\), in case \(\Psi-\Phi\) is completely positive.
\end{definition}

It is easy to see that this gives us the following poset enrichments:
\begin{lemmaE}\label{lem:approx:loewner-enrichment}
  The L\"owner order is a symmetric monoidal poset enrichment for \(\CPTNI\) and is a discrete poset enrichment for \(\CPTP\).
\end{lemmaE}
\begin{proofE}
  Selinger remarked that the L\"owner order is a poset enrichment for \(\CPTNI\)  \cite[Lem.~6.4]{selinger2004towards}, where DCPO-enrichment is formally proved by Cho~\cite[Thm.~5.5]{cho_semantics_2014}. 

  To see that the L\"owner order between completely positive trace-preserving maps is discrete, suppose \(\Phi \preceq \Psi\), then \(\Psi - \Phi\) is completely positive. Then for any normalised state \(a \in \BB(\HH)\), \(\Psi(a) - \Phi(a)\) is  positive. Moreover since these maps preserve the trace \(\Tr(\Psi(a)) = \Tr(\Phi(a)) = \Tr(a) = 1\). Therefore \(\Tr(\Psi(a) - \Phi(a)) = 0\) by linearity of the trace. The only CP map that sends all normalised states to states of trace \(0\) is the zero map, therefore \(\Psi = \Phi\).
\end{proofE}

Second, recall the following order introduced by Carette et al. \cite[Def.~5]{delayedtracememory}:
\begin{definition}\label{def:approx:purification}
  Two parallel morphisms \(f,g:X\to Y\) in a discard category are related by the \textbf{purification order} \(f\subseteq g\) in case there exists some \(g_0:X\to S\otimes Y\) and \(f_0:S\to I\) such that
  \[f= (f_0 \otimes 1_Y) \circ g_0, \quad\text{and}\quad g = (\discard_S \otimes 1_Y) \circ g_0.\]
\end{definition}

It is a direct consequence of \cite[Lem.~6]{delayedtracememory} that:
\begin{lemma}\label{lem:approx:purification-enrichment}
  The purification order is a symmetric monoidal preorder enrichment for \(\CPM\) and a discrete poset enrichment for \(\CPTP\).
\end{lemma}

Carette et al.\ do not formally compare the purification order with the L\"owner order, which we now do:
\begin{propositionE}\label{prop:approx:loewner-purification}
  The L\"owner and purification orders coincide in \(\CPTP\) and \(\CPTNI\) but differ in \(\CPM\).
\end{propositionE}
\begin{proofE}
  We prove each claim separately.
 
  First, the orders coincide in \(\CPTP\) by the uniqueness of the discrete partial order structure.

  Next we show that the orders coincide in \(\CPTNI\). First, suppose that \(\Phi \preceq \Psi\), then \(\Xi = \Psi - \Phi\) is completely positive. Consider the qubit space \(\C^2\) with an orthonormal basis \(\ket{0}, \ket{1}\). Define \(g_0: \BB(\HH) \to \BB(\KK \otimes \C^2)\) and \(f_0: \BB(\C^2) \to \C\) by:
  \[g_0 = \Phi \otimes \ket{0}\bra{0} + \Xi \otimes \ket{1}\bra{1} \quad\text{and}\quad f_0(a) = \bra{0} a \ket{0}.\]
  for any \(a \in \BB(\C^2)\). It is easy to check that these maps are trace non-increasing and that \(\Phi= (1 \otimes f_0) \circ g_0\) and \(\Psi = (1 \otimes \Tr) \circ g_0\).
  Therefore \(\Phi \subseteq \Psi\) in the purification order. Now for the other direction, suppose that \(\Phi \subseteq \Psi\). Then there exists trace non-increasing maps \(g_0, f_0\) such that \(\Phi= (1 \otimes f_0) \circ g_0\) and \(\Psi = (1 \otimes \Tr) \circ g_0\). Therefore, \(\Psi - \Phi = (1 \otimes (\Tr - f_0)) \circ g_0 \); moreover, since \(f_0\) is trace non-increasing \(\Tr - f_0\) is completely positive. Since the composition of completely positive maps is completely positive we obtain \(\Psi \subseteq \Phi\).

  To see that these two orders do not coincide we can simply look at scalars. Note that the scalars in \(\CPM\) are precisely the non-negative real numbers \(\R^{\geq 0}\). For any pair of scalars \(a, b \in \R^{\geq 0}\), taking \(g_0 = b\) and \(f_0 = \frac{a}{b}\) we have that \(a \subseteq b\) in the purification order. However if \(b\leq a\) as real numbers, the inequality does not hold in the L\"owner order.
\end{proofE}

Equipping the quantum categories with the purification order makes each a discard bicategory, completing our list of examples.

\begin{propositionE}\label{prop:approx:quantum-discard-bicat}
  When equipped with the purification order, \(\CPM\), \(\CPTP\) and \(\CPTNI\) are discard bicategories.
\end{propositionE}
\begin{proofE}
  Lax naturality of the discard is immediate from \Cref{def:approx:purification}: for any \(f : X \to Y\), taking \(g_0 = f\) and \(f_0 = \discard_Y\) exhibits \(\discard_Y \circ f \subseteq \discard_X\). The remaining preorder-enrichment axioms hold by \cite[Lem.~6]{delayedtracememory}. For \(\CPTP\) and \(\CPTNI\) the purification order coincides with the L\"owner order by \Cref{prop:approx:loewner-purification}.
\end{proofE}

Note that \(\CPM\) equipped with the L\"owner order is a poset enrichment but \emph{not} a discard bicategory structure, since lax naturality fails. Apart from this exception, all the process theories reviewed in this section are discard bicategories. They are in fact discard bicategories equipped with the purification order, as the following proposition shows.

\begin{propositionE}
  In any partial Markov category, the purification order and the conditional order coincide.
\end{propositionE}
\begin{proofE}
  Denote by \(\subseteq_P\) the purification order and by \(\subseteq_C\) the conditional order.
  Suppose that \(f \subseteq_P g\), then there exists \(f_0\) and \(g_0\) such that 
  \(f= (f_0 \otimes 1_Y) \circ g_0\) and  \(g = (\discard_S \otimes 1_Y) \circ g_0\).
  Then, since \(f_0 \subseteq_C \discard_S\) \cite[Prop.~3.5]{lavore_order_2025} and \(\subseteq_C\) respects the monoidal structure, we have:
  \[ f = (f_0 \otimes 1_Y) \circ g_0 \subseteq_C (\discard_S \otimes 1_Y) \circ g_0 = g\]
  For the other direction, suppose that \(f \subseteq_C g\), then there exists \(r:Y\otimes X\to I\) such that 
  \(f = (1_Y \otimes r)\circ (\cpy_Y \otimes 1_X) \circ  (g\otimes 1_X)\circ \cpy _X\). 
  Let \(E = Y \otimes X\), \(f_0 = r : Y \otimes X \to I\) and \(g_0 = (\cpy_Y \otimes 1_X) \circ  (g\otimes 1_X)\circ \cpy _X\). Then \(f= (f_0 \otimes 1_Y) \circ g_0\) and \(g = (\discard_S \otimes 1_Y) \circ g_0\), therefore \(f \subseteq_P g\).
\end{proofE}

\section{Stateful processes}
\label{sec:state}
\pratendSetLocal{category=state, end, restate}

Having seen how process theories have the structure of discard bicategories, we now ask how to perform \emph{stateful computation} inside them. This section is a detailed review of the existing categorical semantics for stateful processes, examining in each case what equivalence it induces and where it fails on the examples we care about. We will see that no existing approach simultaneously interprets state functorially, lets the time delay be natural, and propagates zero morphisms globally. These shortcomings are precisely what motivate the notion of \emph{behaviour} developed in \Cref{sec:obs}. Throughout this section we work in a plain symmetric monoidal category, the approximation structure becoming relevant only later, when we use it to compare partial observations.

\subsection{Mealy machines in monoidal categories}\label{sec:state:mealy}

The most basic specification of a stateful process is a \emph{Mealy machine} \cite{Mealy1955}: a transition function \(S \times X \to S \times Y\) which consumes an input in \(X\), updates an internal state in \(S\), and produces an output in \(Y\). In Cartesian closed categories, such processes are often packaged using the \emph{state monad} \(T_S(X) \coloneqq (X \times S)^S\) \cite{Moggi1991}: a Kleisli arrow for the state monad is exactly a Mealy machine, and iterating the Kleisli composition gives the \(n\)-step time evolution illustrated in \Cref{eq:mealy_unrolling}.

We are interested in a generalised notion of Mealy machine, in which the underlying transition \(S \otimes X \to S \otimes Y\) need not be a function, but can instead live in any symmetric monoidal category.

\begin{definition}\label{def:state:mealy}
  Let \(\CC\) be a symmetric monoidal category and let \(S\) be an object of \(\CC\).
  A \textbf{symmetric monoidal Mealy machine} is a morphism \( f\colon S\otimes X \to S\otimes Y\), where \(X\) is interpreted as the input, \(Y\) as the output and \(S\) as the internal state.
  For \(n\ge 1\), its \(n\)-step evolution is the canonical composite
  \(
  f^{(n)} \colon S\otimes X^{\otimes n} \longrightarrow  Y^{\otimes n} \otimes S
  \)
  obtained by wiring together \(n\) copies of \(f\) so that the output state of the  \(k\)-th copy is fed back to
  the input state of the \((k+1)\)-st copy, repeatedly using \(\swap\), as in \Cref{eq:mealy_unrolling}.

  Given an \textbf{initial state} \(s_0\colon I\to S\), the \textbf{evaluation} at time step \(n\) is given by
  \[
  f^{(n)}\circ (s_0\otimes 1_{X^{\otimes n}})\colon X^{\otimes n}\to  Y^{\otimes n} \otimes S.
  \]
\end{definition}

Instantiating \Cref{def:state:mealy} in each of our guiding examples recovers a well-studied class of stateful process:
\begin{itemize}
  \item In \(\Par\), a Mealy machine is a \textbf{partial Mealy machine}: at each step the transition may be undefined, halting the computation. Partial Mealy machines are a standard example in the coalgebraic theory of state-based systems, where their behaviour is obtained through the (generalised) powerset/determinization construction \cite{silva_generalizing_2013}.

  \item In \(\Rel\), a Mealy machine is a \textbf{nondeterministic Mealy machine}: each input may lead to several output-and-state pairs. These coincide with the \emph{observable nondeterministic finite-state machines} studied in model-based testing and automata learning \cite{khalili_learning_2014}.

  \item In \(\BorelTotal\) or \(\FinStoch\), a Mealy machine has a random transition function, and its time evolution generates a \textbf{controlled stochastic process} \cite[Def.~7.1]{lavore_monoidal_2022}: a Markov chain whose dynamics are driven by the input stream.

  \item In \(\CPTP\), a Mealy machine is a \textbf{quantum channel with finite memory} \cite{kretschmann_quantum_2005}, where the state object \(S\) is the memory carried between uses of the channel. Such channels are widely studied in the contexts of spin chains, quantum coding, and quantum communication \cite{plenio_spin_2007,caruso_quantum_2014,ollivier_description_2003}. Going beyond this trace-preserving literature, our framework also accommodates the non–trace-preserving setting: working over \(\CPTNI\) or \(\CPM\) additionally captures conditioning and post-selection.
\end{itemize}
Running any of these machines from an initial state should, intuitively, produce a process that transforms input streams into output streams; but when should two such machines be deemed equivalent? This depends on what one takes the behaviour of a process to be. For the partial and nondeterministic examples, the cited works adopt a computational view, in which the behaviour of a machine is the language it accepts or the computation it performs. We instead regard the behaviour of a process as the family of constraints it imposes, in the systems-theoretic tradition of Willems~\cite{Willems}. In the remainder of this section we review several categorical semantics for stateful processes that formalise the computational view, and argue that none of them accommodates this constraint-based notion of behaviour, which we develop in \Cref{sec:obs}.

\subsection{The state construction}\label{sec:state:stateful}
It is natural to ask: \emph{what is the categorical semantics of Mealy machines over an arbitrary symmetric monoidal category?}
At a minimum, such a semantics should provide a principled way to
\emph{hide} the internal state \(S\), turning a stateful description into an external process from \(X\) to
\(Y\) as in \Cref{eq:mealy_unrolling}.
Following the approach of Katis, Sabadini, and Walters~\cite[\S~3.1]{Katis1997}
(and later by various others such as \cite{sprunger,Goncharov}), we axiomatise this state-hiding operation with the notion of a \emph{feedback category}.
Our formulation is a mild variant of that of Katis, Sabadini, and Walters, adapted to the guarded/delayed setting needed later:\footnote{More generally one could replace  \(\Gamma, \Delta\) and \(\DD\) with a monoidal endoprofunctor on \(\mathcal C\), modifying the sliding equation appropriately.}
\begin{definition}\label{def:state:feedback}
  A \textbf{feedback category} is a tuple \((
  \CC,\DD, \Gamma, \Delta, \fbk)\) where:
  \begin{itemize}
    \item \(\Gamma\colon\DD\to \CC\) is a symmetric monoidal functor,  called the \textbf{guard functor};
    \item \(\Delta\colon\DD\to \DD\) is a symmetric monoidal endofunctor, called the \textbf{delay endofunctor};
    \item \(\fbk_S\colon \CC(\Gamma(\Delta S) \otimes X, \Gamma(S) \otimes Y) \to \CC(X, Y)\) is a family of functions, called the \textbf{feedback operator}.
  \end{itemize}
  Moreover  for all \(f\in \CC(\Gamma(\Delta S)\otimes X,\Gamma(S)\otimes Y)\), \(u\in \CC(X',X)\),  \(v\in \CC(Y,Y')\), and \(h \in \DD(S,T)\), we ask that the following equations hold: 
  \begin{itemize}
    \item \textbf{Tightening:}
    \(v \circ \fbk_S(f) \circ u = \fbk_S((1_{\Gamma S} \otimes v) \circ f \circ ( 1_{\Gamma (\Delta S)} \otimes u))\);
    \item \textbf{Vanishing:} \(\fbk_I(f) = f\);
    \item \textbf{Joining:} \(\fbk_T (\fbk_S(f)) = \fbk_{S \otimes T}(f)\);
    \item \textbf{Strength:} \(\fbk_S(f) \otimes g = \fbk_{S}(f \otimes g)\);
    \item \textbf{Sliding:}
    \(\fbk_T\bigl((\Gamma(h)\otimes 1_Y)\circ f\bigr)
  =\fbk_S\bigl(f\circ(\Gamma(\Delta h)\otimes 1_X)\bigr)\).
  \end{itemize}
  In case \(\Gamma\) is faithful call \(\DD\) the \textbf{guarded subcategory}.

  A \textbf{feedback functor}  \((F, F|_{\DD})\colon (\CC, \DD, \Gamma, \Delta, \fbk) \to (\CC',\DD', \Gamma', \Delta', \fbk')\) between feedback categories
  consists of symmetric monoidal functors \(F\colon\CC \to \CC'\) and \(F|_{\DD}\colon\DD\to \DD'\), satisfying   \(\Gamma'\circ F|_{\DD} = F \circ \Gamma\),
  \(\Delta'\circ F|_{\DD}=F|_{\DD}\circ \Delta\), and \(F(\fbk_S(f)) = \fbk'_{F(S)}(F(f))\).
\end{definition}

We interpret \(\Gamma\colon\DD\to \CC\) as parametrising what is allowed to be passed through the state at successive timesteps, and \(\Delta\colon\DD\to \DD\) as an effect which is applied when something is passed through the state; graphically:
\begin{equation}
  \label{eq:state:feedback-sliding}
  \begin{tikzpicture}
	\begin{pgfonlayer}{nodelayer}
		\node [style=none] (16) at (-1, 1.5) {};
		\node [style=none] (17) at (-5.25, 1.5) {};
		\node [style=none] (18) at (-1, -1.375) {};
		\node [style=none] (19) at (-5.25, -1.375) {};
		\node [style=box] (20) at (-3.75, -0.375) {\(f\)};
		\node [style=none] (21) at (-5.25, -0.75) {};
		\node [style=none] (22) at (-1, -0.75) {};
		\node [style=none] (23) at (-4.5, 0) {};
		\node [style=none] (24) at (-1.75, 0) {};
		\node [style=none] (25) at (-4.5, 1) {};
		\node [style=none] (26) at (-1.75, 1) {};
		\node [style=none] (27) at (-3.875, 1) {};
		\node [style=none] (28) at (-3.625, 1.25) {};
		\node [style=none] (29) at (-3.625, 0.75) {};
		\node [style=none] (30) at (-0.5, 0) {\(=\)};
		\node [style=box] (31) at (-2.5, 0.075) {\(\Gamma h\)};
		\node [style=none] (32) at (5.25, 1.5) {};
		\node [style=none] (33) at (0, 1.5) {};
		\node [style=none] (34) at (5.25, -1.375) {};
		\node [style=none] (35) at (0, -1.375) {};
		\node [style=box] (36) at (3.75, -0.375) {\(f\)};
		\node [style=none] (37) at (0, -0.75) {};
		\node [style=none] (38) at (5.25, -0.75) {};
		\node [style=none] (39) at (0.75, 0) {};
		\node [style=none] (40) at (4.5, 0) {};
		\node [style=none] (41) at (0.75, 1) {};
		\node [style=none] (42) at (4.5, 1) {};
		\node [style=none] (43) at (3.625, 1) {};
		\node [style=none] (44) at (3.875, 1.25) {};
		\node [style=none] (45) at (3.875, 0.75) {};
		\node [style=box] (46) at (1.875, 0.075) {\(\Gamma(\Delta h)\)};
	\end{pgfonlayer}
	\begin{pgfonlayer}{edgelayer}
		\draw [style=background] (17.center)
			 to (16.center)
			 to (18.center)
			 to (19.center)
			 to cycle;
		\draw (21.center) to (22.center);
		\draw (24.center)
			 to [bend right=90, looseness=1.25] (26.center)
			 to (25.center)
			 to [bend right=90, looseness=1.25] (23.center)
			 to cycle;
		\draw (27.center) to (28.center);
		\draw (29.center) to (27.center);
		\draw [style=background] (33.center)
			 to (32.center)
			 to (34.center)
			 to (35.center)
			 to cycle;
		\draw (37.center) to (38.center);
		\draw (40.center)
			 to [bend right=90, looseness=1.25] (42.center)
			 to (41.center)
			 to [bend right=90, looseness=1.25] (39.center)
			 to cycle;
		\draw (43.center) to (44.center);
		\draw (45.center) to (43.center);
	\end{pgfonlayer}
\end{tikzpicture}

\end{equation}

In particular, when  \(\Gamma=\Delta=1_{\CC}\), a morphism
\(f\colon S\otimes X\to S\otimes Y\) is a symmetric monoidal Mealy machine in the sense of \Cref{def:state:mealy}, where  \(\fbk_S(f)\colon X\to Y\) is interpreted as the process given by hiding the memory of \(f\) as an internal state.

Note that a \emph{traced symmetric monoidal  structure} \cite{Joyal1996} on a SMC \(\CC\) is precisely a feedback structure with  \(\Delta=\Gamma=1_{\CC}\) satisfying the additional \emph{yanking identity}:  \(\Tr_X(\swap_{X,X})=1_X\), which trivialises the time dependency.

Generalising the ``Circ'' construction of Katis et al.~\cite{katis_feedback_2002}, now more commonly referred to in the theoretical computer science literature as the \emph{state construction} (see for example~\cite{lavore_monoidal_2022}), we can \emph{freely} construct a feedback category to give a categorical semantics for stateful processes over any base symmetric monoidal category:
\begin{definition}[State construction]\label{def:state:state}
  Given SMCs  \(\CC\) and \(\DD\) with symmetric monoidal functors \(\Gamma\colon\DD\to \CC\) and  \(\Delta\colon \DD \to \DD\),
  the SMC \(\St_{\DD}(\CC, \Delta)\) has the same objects as \(\CC\),
  where the morphisms given by elements of the set:
  \[
    \St_{\DD}(\CC, \Delta)(X,Y) \coloneqq \int^{S \in \DD}
    \CC(\Gamma(\Delta S)\otimes X, \Gamma(S) \otimes Y ).
  \]
  The rest of the structure defining a symmetric monoidal category is given by the evident natural transformations.

  A morphism in \(\St_{\DD}(\CC, \Delta)\) is called a \textbf{stateful morphism}.
  In case, \(\Gamma= 1_\CC\) or \(\Delta = 1_{\DD}\), we omit  the labels \(\DD\) or  \(\Delta\).
\end{definition}

\begin{remark}
  This integral symbol denotes a coend \cite{Loregian2021}, so that
  elements of \(\St_{\DD}(\CC, \Delta)(X,Y)\) are pairs
  \[(S\in \Ob(\DD), f\in \CC(\Gamma(\Delta S)\otimes X, \Gamma(S) \otimes Y )),\]
  modulo the largest equivalence relation such that for any \(r\colon S \to S'\) and any \(h\colon S' \otimes X \to S \otimes Y\):
  \[[S', (\Gamma(s) \otimes 1_Y) \circ h] = [S, h \circ (\Gamma(\Delta(s)) \otimes 1_X)].\]
\end{remark}

The state construction is the \emph{free} feedback category built from  \(\CC\), \(\DD\), \(\Gamma\), and \(\Delta\), see \Cref{prop:state:state-free} for a detailed proof.
Given \(f \in \St_{\DD}(\CC, \Delta)(\Gamma(\Delta S) \otimes X,
\Gamma(S) \otimes Y)\), the process \(\fbk(f) \colon X \to Y\) is
obtained by adjoining \(S\) to the memory. By the definition of the
coend, two stateful processes are identified precisely when they are
related by finitely many applications of the sliding rule
\eqref{eq:state:feedback-sliding}.

\begin{textAtEnd}

  \begin{lemma}\label{lem:state:induced-feedback}
    Take  two symmetric monoidal functors \(\Gamma\colon\DD\to \CC\) and \(\Delta\colon\DD\to\DD\), and let
    \(\eta\colon\CC\to \St_{\DD}(\CC,\Delta)\) be the canonical symmetric monoidal functor.
    Then \(\St_{\DD}(\CC,\Delta)\) carries a canonical feedback structure given by \(\Delta\colon\DD\to \DD\),
    \[
      \hat\Gamma\coloneqq \eta\circ\Gamma\colon\DD\to \St_{\DD}(\CC,\Delta),
      \qquad\text{and}\qquad
      \fbk_S([M,f])\coloneqq [M\otimes S,f].
    \]
  \end{lemma}
  \begin{proof}
    We need to show that the feedback operator:
    \begin{enumerate}
      \claimitem{state:state-fbk:typed}{is well-typed;}
      \claimitem{state:state-fbk:wd}{is well-defined on equivalence classes;}
      \claimitem{state:state-fbk:axioms}{satisfies the feedback category axioms.}
    \end{enumerate}

    Take objects \(S\in\Ob(\DD)\) and \(X,Y\in\Ob(\CC)\).

    First, because \(\eta\) is the identity on objects, we have that
    \(
    \hat\Gamma(S)=\Gamma(S)
    \) and 
    \(
    \hat\Gamma(\Delta S)=\Gamma(\Delta S).
    \)

    Next, take a representative
    \(
    [M,f]\in \St_{\DD}(\CC,\Delta)\bigl(\hat\Gamma(\Delta S)\otimes X,\ \hat\Gamma(S)\otimes Y\bigr),
    \)
    where
    \[
    M\in\Ob(\DD)
    \quad\text{and}\quad
    f\in \CC\bigl(\Gamma(\Delta M)\otimes \Gamma(\Delta S)\otimes X,\ \Gamma(M)\otimes \Gamma(S)\otimes Y\bigr).
    \]

    Define the feedback operator on the representatives \([M,f]\) by \(\fbk_S([M,f])\coloneqq [M\otimes S,f]\).
    
    \claimparagraph{state:state-fbk:typed}{Feedback operator is well-typed.}
    Because \(\Delta\) and \(\Gamma\) are strict symmetric monoidal, we have that
    \(
    \Gamma(\Delta(M\otimes S))=\Gamma(\Delta M)\otimes \Gamma(\Delta S)
    \)
    and
    \(
    \Gamma(M\otimes S)=\Gamma(M)\otimes \Gamma(S),
    \)
    so that \(f\) is also an element of
    \(
    \CC\bigl(\Gamma(\Delta(M\otimes S))\otimes X,\ \Gamma(M\otimes S)\otimes Y\bigr),
    \)
    and thus,  \([M\otimes S,f]\in \St_{\DD}(\CC,\Delta)(X,Y)\). Therefore, \(\fbk_S\) is well-typed.

    \claimparagraph{state:state-fbk:wd}{Well-definedness of feedback operator.}
    Fix \(t\in\DD(M,M')\) and
    \[
    f\in \CC\bigl(\Gamma(\Delta M)\otimes \Gamma(\Delta S)\otimes X,\ \Gamma(M')\otimes \Gamma(S)\otimes Y\bigr).
    \]
    Moreover, write \(A\coloneqq \hat\Gamma(\Delta S)\otimes X\) and \(B\coloneqq \hat \Gamma(S)\otimes Y\). 

    Observe that the feedback operation in \(\St_{\mathcal D}(\mathcal C, \Delta)\) satisfies the sliding equation on representatives of the equivalence class:
    \begin{align*}
      \fbk_{M'}(M'\otimes S,(\hat \Gamma(t) \otimes 1_B)\circ f)
      &  = (S,(\hat \Gamma(t) \otimes 1_B)\circ f)
       = (S,f\circ(\hat \Gamma(\Delta t) \otimes 1_A))\\
      & =\fbk_M(M\otimes S,f\circ (\hat \Gamma(\Delta t) \otimes 1_A)).
    \end{align*}

    Therefore \(\fbk_M\) is well-defined.

    \claimparagraph{state:state-fbk:axioms}{Feedback category axioms.}
    The other feedback category axioms follow from the evident diagram chases.
  \end{proof}
  This is the \emph{free} feedback category, in the following sense:
  \begin{proposition}\label{prop:state:state-free}
    Take two symmetric monoidal functors  \(\Gamma\colon\DD\to \CC\) and
    \(\Delta\colon\DD\to \DD\).
    Given a feedback category \((\CC',\DD', \Gamma', \Delta', \fbk')\) and two symmetric monoidal functors
    \(F\colon\CC\to \mathcal{C'}\) and \(F|_{\DD}\colon \DD\to \DD'\) satisfying
    \(
    \Gamma'\circ F|_{\DD} = F \circ \Gamma
    \) and \(\Delta'\circ F|_{\DD}=F|_{\DD}\circ \Delta,
    \)
    there is a unique feedback functor
    \(
    ( F^\sharp, F|_{\mathcal D}) \colon\St_{\DD}(\CC,\Delta)\to \CC'
    \)
    making the diagram commute:
    \[
    \xymatrix@C=3.5em@R=1.0em{
      \CC \ar[r]^\eta \ar[dr]_{F} &
      \St_{\DD}(\CC,\Delta) \ar@{..>}[d]^{(F^\sharp, F|_{\mathcal D})} \\
      & \CC'}
    \]
  \end{proposition}
  \begin{proof}

    Define \(F^\sharp\colon\St_{\DD}(\CC,\Delta)\to\CC'\) on objects by \(F^\sharp(X)\coloneqq F(X)\).
    Given some \([S,f]\in \St_{\DD}(\CC,\Delta)(X,Y)\), define
    \(
    F^\sharp([S,f])\coloneqq \fbk'_{F|_{\DD}(S)}\bigl(F(f)\bigr)\in \CC'(F(X),F(Y)).
    \)

    We need to show that \(F^\sharp\):
    \begin{enumerate}
      \claimitem{state:state-free:typed}{is well-typed;}
      \claimitem{state:state-free:wd}{is well-defined on equivalence classes;}
      \claimitem{state:state-free:funct}{is functorial and symmetric monoidal;}
      \claimitem{state:state-free:triangle}{satisfies \(F^\sharp\circ\eta=F\);}
      \claimitem{state:state-free:compat}{is compatible with \(\Gamma\) and \(\Delta\);}
      \claimitem{state:state-free:fbk}{preserves the feedback operator;}
      \claimitem{state:state-free:unique}{is the unique such functor.}
    \end{enumerate}

    \claimparagraph{state:state-free:typed}{Well-typedness.}
    Since
    \(\Gamma'\circ F|_{\DD}=F\circ\Gamma\)
    and
    \(\Delta'\circ F|_{\DD}=F|_{\DD}\circ\Delta\), we have that
    \[
    F(\Gamma(S))=\Gamma'(F|_{\DD}(S)),
    \quad\text{and}\quad
    F(\Gamma(\Delta S))=\Gamma'(\Delta'(F|_{\DD}(S))).
    \]

    \claimparagraph{state:state-free:wd}{Well-definedness.}
    Take morphisms
    \(s\in\DD(S,S')\) and 
    \(f\in \CC(\Gamma(\Delta S')\otimes X,\ \Gamma(S)\otimes Y)\).
    Observe that in \(\St_{\DD}(\CC,\Delta)\) we have
    \(
    [S',(\Gamma(s)\otimes 1_Y)\circ f]=[S,f\circ(\Gamma(\Delta s)\otimes 1_X)].
    \)
    On the other hand in \(\CC'\), we have
    {\allowdisplaybreaks
    \begin{align*}
    \fbk&'_{F|_{\DD}(S')}\Bigl(F\bigl((\Gamma(s)\otimes 1_Y)\circ f\bigr)\Bigr)\\
    &=
    \fbk'_{F|_{\DD}(S')}\Bigl((F(\Gamma(s))\otimes 1_{F(Y)})\circ F(f)\Bigr)\\
    &=
    \fbk'_{F|_{\DD}(S')}\Bigl((\Gamma'(F|_{\DD}(s))\otimes 1_{F(Y)})\circ F(f)\Bigr)\\
    &=
    \fbk'_{F|_{\DD}(S)}\Bigl(F(f)\circ(\Gamma'(\Delta'(F|_{\DD}(s)))\otimes 1_{F(X)})\Bigr)\\
    &=
    \fbk'_{F|_{\DD}(S)}\Bigl(F(f)\circ(\Gamma'(F|_{\DD}(\Delta s))\otimes 1_{F(X)})\Bigr)\\
    &=
    \fbk'_{F|_{\DD}(S)}\Bigl(F(f)\circ(F(\Gamma(\Delta s))\otimes 1_{F(X)})\Bigr)\\
    &=
    \fbk'_{F|_{\DD}(S)}\Bigl(F(f)\circ F(\Gamma(\Delta s)\otimes 1_X)\Bigr)\\
    &=
    \fbk'_{F|_{\DD}(S)}\Bigl(F\bigl(f\circ(\Gamma(\Delta s)\otimes 1_X)\bigr)\Bigr).
    \end{align*}
    }

    Hence \(F^\sharp\) is well-defined.

    \claimparagraph{state:state-free:funct}{Functoriality and symmetric monoidality.}
    Follow from the evident diagram chases.

    \claimparagraph{state:state-free:triangle}{Triangle commutes.}
    Given \(f\in\CC(X,Y)\) we have
    \begin{align*}
      F^\sharp(\eta(f))
      =
      F^\sharp([I,f])
      =
      \fbk'_{F|_{\DD}(I)}(F(f))
      =
      \fbk'_I(F(f))
      =
      F(f).
    \end{align*}

    \claimparagraph{state:state-free:compat}{Compatibility with \(\Gamma\) and \(\Delta\).}
    By assumption we have
    \(\Delta'\circ F|_{\DD}=F|_{\DD}\circ\Delta;\)
    moreover,
    \[
      F^\sharp\circ {\hat \Gamma}
      =
      F^\sharp\circ\eta\circ\Gamma
      =
      F\circ\Gamma
      =
      \Gamma'\circ F|_{\DD}.
    \]

    \claimparagraph{state:state-free:fbk}{Preservation of feedback.}
    Let \(S\in\Ob(\DD)\) and take a representative
    \[[M,h]\in \St_{\DD}(\CC,\Delta)\bigl({\hat \Gamma}(\Delta S)\otimes X,\ {\hat \Gamma}(S)\otimes Y\bigr)\,.\]
    Then
    \begin{align*}
    F^\sharp(\fbk_S([M,h]))
    &=
    F^\sharp([S\otimes M,h])
    =
    \fbk'_{F|_{\DD}(S\otimes M)}(F(h))
    =
    \fbk'_{F|_{\DD}(S)}\Bigl(\fbk'_{F|_{\DD}(M)}(F(h))\Bigr)\\
    &=
    \fbk'_{F|_{\DD}(S)}\bigl(F^\sharp([M,h])\bigr).
    \end{align*}

    \claimparagraph{state:state-free:unique}{Uniqueness.}
    Take another feedback functor
    \((H,F|_{\DD})\colon\St_{\DD}(\CC,\Delta)\to\CC',\)
    such that \(H\circ\eta=F\). Moreover  take any representative
    \((S,f)\in \St_{\DD}(\CC,\Delta)(\hat \Gamma (\Delta S) \otimes X, (\hat \Gamma S) \otimes Y)\,.\)
    Therefore
    \begin{align*}
      H([S,f])
      =
      H(\fbk_S([I,f]))
      =
      H(\fbk_S(\eta f))
      =
      \fbk'_{F|_{\DD}(S)}\bigl(H(\eta f))\bigr)
      =
      \fbk'_{F|_{\DD}(S)}\bigl(F(f)\bigr)
      =
  F^\sharp([S,f]).
    \end{align*}
    Hence \(H=F^\sharp\), so \((F^\sharp,F|_{\DD})\) is unique.
  \end{proof}
\end{textAtEnd}

We call the map \(\delta_X\coloneqq \fbk_X (\swap_{\Gamma(\Delta X), \Gamma X}) \), the \textbf{delay morphism} at \(X\). Generalising the work of Katis et al.~\cite[Prop.~3]{Katis1997}, we find that given some weak  assumptions on the guard functor, feedback structures in compact-closed categories are determined by the delay morphisms \(\{\delta_X\}_{X \in \Ob(\mathcal D)}\):
\begin{propositionE}\label{prop:state:feedback-compact-closed}
  Take a compact-closed category \(\CC\), a faithful symmetric monoidal functor
  \(\Gamma\colon\DD \rightarrowtail \CC\)
  such that for all \(X \in \Ob(\mathcal D)\), \(\Gamma^{-1}(\delta_X ) \neq \emptyset\), and a symmetric monoidal endofunctor
  \(\Delta\colon\DD\to \DD\). The existence of a feedback operator is equivalent
  to a monoidal natural transformation \(\delta\colon1_{\DD}\Rightarrow \Delta\), called the
  \textbf{delay natural transformation}. Moreover, such a natural transformation is necessarily a natural isomorphism.
\end{propositionE}
\begin{proofE}
  We need to prove that:
  \begin{enumerate}
    \claimitem{state:fbk-cc:fbk-implies-delay}{the existence of a feedback operator implies a monoidal natural transformation \(\delta\colon1_{\DD}\Rightarrow\Delta\);}
    \claimitem{state:fbk-cc:delay-implies-fbk}{a monoidal natural transformation \(\delta\colon1_{\DD}\Rightarrow\Delta\) implies a feedback operator;}
    \claimitem{state:fbk-cc:iso}{the delay \(\delta\) is necessarily a natural isomorphism.}
  \end{enumerate}

  \claimparagraph{state:fbk-cc:fbk-implies-delay}{Feedback operator implies delay.}
  Assume that for all objects \(S \in \DD\), we have a feedback operator 
  \(
    \fbk_S\colon \CC(\Gamma(\Delta S) \otimes X, \Gamma(S) \otimes Y) \to \CC(X, Y).
  \)

  Using the fact that  \(\Gamma\colon\DD\rightarrowtail \CC\)
  is a faithful symmetric monoidal functor,
  for all \(S\in \Ob(\DD)\), define \(\delta_S \in \DD(S,\Delta S)\) by 
  \(
    \delta_S \coloneqq
    \Gamma^{-1}\bigl(\fbk_S\bigl(\swap_{\Gamma(\Delta S),\Gamma S}\bigr)\bigr)
    \in \CC(\Gamma(S), \Gamma(\Delta S)).
  \)
  Note that \(\delta_S\) exists by assumption and is well-defined due to the faithfulness of \(\Gamma\).

  Now we prove naturality.  Define
  \[
    f \coloneqq \swap_{\Gamma(\Delta S),\Gamma(S)}
    \in \CC(\Gamma(\Delta T)\otimes \Gamma(S),\, \Gamma(S)\otimes \Gamma(\Delta T)).
  \]
  Take any \(h\in \DD(S,T)\). Using sliding, we have
  \[
    \fbk_T\bigl((\Gamma(h)\otimes 1_{\Gamma(\Delta T)})\circ f\bigr)
    =
    \fbk_S\bigl(f\circ(\Gamma(\Delta h)\otimes 1_{\Gamma(S)})\bigr).
  \]
  Using the naturality of \(\swap\), we have
  \[
    (\Gamma(h)\otimes 1)\circ \swap_{\Gamma(\Delta T),\Gamma(S)}
    =
    \swap_{\Gamma(\Delta T),\Gamma(T)}\circ(1\otimes \Gamma(h)),
  \]
  and similarly
  \[
    \swap_{\Gamma(\Delta T),\Gamma(S)}\circ(\Gamma(\Delta h)\otimes 1)
    =
    (1\otimes \Gamma(\Delta h))\circ \swap_{\Gamma(\Delta S),\Gamma(S)}.
  \]
  Substituting these into the sliding equation gives:
  \begin{align*}
    \fbk_T\bigl(\swap_{\Gamma(\Delta T),\Gamma(T)}\circ(1\otimes \Gamma(h))\bigr)
    = \fbk_S\bigl((1\otimes \Gamma(\Delta h))\circ \swap_{\Gamma(\Delta S),\Gamma(S)}\bigr).
  \end{align*}
  Now applying tightening twice yields
  \begin{align*}
    \fbk_T\bigl(\swap_{\Gamma(\Delta T),\Gamma(T)}\circ(1\otimes \Gamma(h))\bigr)
    = \fbk_T\bigl(\swap_{\Gamma(\Delta T),\Gamma(T)}\bigr)\circ \Gamma(h)
    =\Gamma(\delta_T)\circ \Gamma(h),
  \end{align*}
  and similarly
  \begin{align*}
    \fbk_S \bigl((1\otimes \Gamma(\Delta h))\circ \swap_{\Gamma(\Delta S),\Gamma(S)}\bigr)
    = \Gamma(\Delta h)\circ \fbk_S\bigl(\swap_{\Gamma(\Delta S),\Gamma(S)}\bigr)
    = \Gamma(\Delta h)\circ \Gamma(\delta_S).
  \end{align*}
  Hence
  \[
  \Gamma(\Delta h)\circ \Gamma(\delta_S)=\Gamma(\delta_T)\circ \Gamma(h).
  \]
  Since \(\Gamma\) is faithful, this implies that \(\Delta h\circ \delta_S=\delta_T\circ h\),
  therefore \(\delta\colon1_{\DD}\Rightarrow \Delta\) is natural.
  Monoidality of \(\delta\) follows from strength, joining, and vanishing together with coherence of \(\swap\).

  \claimparagraph{state:fbk-cc:delay-implies-fbk}{Delay implies feedback operator.}
  Conversely, assume we have a monoidal natural transformation \(\delta\colon1_{\DD}\Rightarrow \Delta\).
  Take \(S\in\Ob(\DD)\) and
  \(f\in \CC(\Gamma(\Delta S)\otimes X, \Gamma(S)\otimes Y)\).
  Define
  \begin{align*}
    \fbk_S(f)
    \coloneqq
    (\cap_{\Gamma(S)} \otimes 1_Y)&\circ(\swap_{\Gamma(S),\Gamma(S)^*}\otimes 1_Y)
    \circ (1_{\Gamma(S)}\otimes \swap_{Y,\Gamma(S)^*})\\
    & \circ\bigl((f\circ(\Gamma(\delta_S)\otimes 1_X))\otimes 1_{\Gamma(S)^*}\bigr)\\
    & \circ (1_{\Gamma(S)}\otimes \swap_{\Gamma(S)^*,X}) \circ (\cup_{\Gamma(S)}\otimes 1_X).
  \end{align*}
  The feedback axioms follow from compact-closure together with the fact that \(\delta\) is a monoidal natural transformation, using essentially the same argument as is spelled out in \cite{lavore_spangraph_2022}.

  \claimparagraph{state:fbk-cc:iso}{The delay is a natural isomorphism.}
  Consider the delay natural transformation \(\delta \colon 1_{\DD}\Rightarrow \Delta\) from \Cref {prop:state:feedback-compact-closed}.
  Define \(\gamma\colon\Delta\Rightarrow 1_{\DD}\) componentwise by
  \[
  \gamma_X \coloneqq
  (\cap_{\Delta(X^*)}\otimes 1_X)\circ
  (1_{\Delta X}\otimes \delta_{X^*}\otimes 1_X)\circ
  (1_{\Delta X}\otimes \cup_{X^*}).
  \]
  Naturality of \(\gamma\) follows by applying the dual \((-)^*\) to the naturality square for \(\delta\).

  Now we show that \(\gamma\) is inverse to \(\delta\).
  First note that because \(\Delta\) is a symmetric monoidal functor and \(\DD\) is compact-closed, for all objects \(D\) in \(\DD\) we have that
  \(\cup_{\Delta Z}=\Delta(\cup_Z)\) and \(\cap_{\Delta Z}=\Delta(\cap_Z).\)

  For the one side of the inverse, note that naturality of \(\delta\) at \(\cap_{X^*}\colon X\otimes X^*\to I\), together with monoidality of \(\delta\), gives
  \(
  \cap_{\Delta(X^*)}\circ(\delta_X\otimes \delta_{X^*})=\cap_{X^*}.
  \)
  Therefore
  \begin{align*}
  \gamma_X\circ \delta_X
  &=
  (\cap_{\Delta(X^*)}\otimes 1_X)\circ
  (1_{\Delta X}\otimes \delta_{X^*}\otimes 1_X)
   \circ
  (1_{\Delta X}\otimes \cup_{X^*})\circ \delta_X\\
  &=
  (\cap_{\Delta(X^*)}\otimes 1_X)\circ
  (\delta_X\otimes \delta_{X^*}\otimes 1_X)\circ
  (1_X\otimes \cup_{X^*})\\
  &=
  (\cap_{X^*}\otimes 1_X)\circ(1_X\otimes \cup_{X^*})
  =
  1_X.
  \end{align*}
  On the other hand, 
  the naturality of \(\delta\) at \(\cup_{X^*}\colon I\to X^*\otimes X\), together with the monoidality of \(\delta\), imply that
  \(
  (\delta_{X^*}\otimes \delta_X)\circ \cup_{X^*}=\cup_{\Delta(X^*)}.
  \).
  Therefore
  \begin{align*}
  \delta_X\circ \gamma_X
  &=
  \delta_X\circ
  (\cap_{\Delta(X^*)}\otimes 1_X)\circ
  (1_{\Delta X}\otimes \delta_{X^*}\otimes 1_X)
  \circ
  (1_{\Delta X}\otimes \cup_{X^*})\\
  &=
  (\cap_{\Delta(X^*)}\otimes 1_{\Delta X})\circ
  (1_{\Delta X}\otimes \delta_{X^*}\otimes \delta_X)\circ
  (1_{\Delta X}\otimes \cup_{X^*})\\
  &=
  (\cap_{\Delta(X^*)}\otimes 1_{\Delta X})\circ
  (1_{\Delta X}\otimes \cup_{\Delta(X^*)})
  =
  1_{\Delta X}. \tag*{\qedhere}
  \end{align*}
\end{proofE}

\subsection{Stateful morphism sequences}

Applying the state construction to \(\mathcal C\) identifies stateful processes built from \(\mathcal C\) that differ only by a reparametrisation of the internal state.  However, this semantics is purely formal: it does not specify how to iterate a process through time.
To model iteration explicitly, we pass from \(\CC\) to categories of time-indexed families in
\(\CC\).  The key extra ingredient is a canonical endofunctor expressing  a one timestep delay:
\begin{definition}\label{def:state:time-indexed}
  Let \(\Z\) and \(\N\) denote the discrete symmetric monoidal categories whose objects
  are respectively the integers and the natural numbers, where the monoidal product and unit are given by addition and \(0\).

  Given a symmetric monoidal category \(\CC\), the functor categories
  \([\Z,\CC]\) and \([\N,\CC]\)  have objects and morphisms given by \mbox{(bi)infinite} sequences of objects and morphisms in \(\CC\).
  Equipping these functor categories with the
  \emph{pointwise} symmetric monoidal structure, we find that there are symmetric monoidal functors:
  \begin{align*}
    \shftZ\colon [\Z,\CC] \to [\Z,\CC];\qquad
    &(\shftZ X)_{j+1} \coloneqq X_{j},\quad  (\shftZ f)_{j+1} \coloneqq f_{j};\\
    \shftN\colon [\N,\CC] \to [\N,\CC];\qquad
    &(\shftN X)_0 \coloneqq I,\quad (\shftN X)_{j+1}\coloneqq X_j,\ \\
    &(\shftN f)_0 \coloneqq 1_I,\quad (\shftN f)_{j+1}\coloneqq f_j.
  \end{align*}
  With this data, we can build feedback categories of \textbf{extensional monoidal streams}~\cite[\S~3.2]{lavore_monoidal_2022}.
  Call a morphism in \(\St([\N, \CC], \shftN )\) an \textbf{initialised stateful morphism sequence} in \(\CC\); and a morphism in \(\St([\Z, \CC], \shftZ )\) an \textbf{uninitialised stateful morphism sequence}.
\end{definition}

A guard functor \(\Gamma\colon \DD \to \CC\) induces functors \(\Gamma\colon [\N, \DD] \to [\N, \CC]\) and \(\Gamma\colon [\Z, \DD] \to [\Z, \CC]\) which can be used to restrict the sliding rule in stateful morphism sequences. For instance, initialised stateful morphism sequences with sliding guarded by \(\Gamma\) live in \(\St_{[\N, \DD]}([\N, \CC], \shftN)\). We however do not require this restriction in our results.

In the uninitialised setting, we can unroll stateful morphisms from the free feedback category into stateful morphism sequences analogously to \Cref{eq:mealy_unrolling}:
\begin{lemmaE}\label{lem:state:z-stateful}
  For every SMC \(\CC\), there is a unique feedback functor \(U\colon\St(\CC) \to \St([\Z, \CC], \shftZ)\) which is the identity on objects and acts on morphisms by:
  \[\Bigl( [S,f]\colon X\to Y\Bigr) \quad \longmapsto \quad \Bigl(\bigl[( f\colon S\otimes X\to  S\otimes Y)_{j\in \Z}\bigr] \colon X^\Z\to Y^\Z \Bigr)\,.\]
\end{lemmaE}
\begin{proofE}
  Consider the symmetric monoidal functor:
  \[
    J\colon\CC\to \St([\Z,\CC],\shftZ ) ;\quad (f\colon X\to Y) \mapsto [ (I \otimes f)_{j \in \Z}].
  \]
  Using the canonical functor \(\eta\colon\CC\to \St(\CC)\),
  the universal property of the state construction given in \Cref{prop:state:state-free} induces a unique feedback functor
  \(
    J^\sharp\colon\St(\CC)\to \St([\Z,\CC],\shftZ),
  \)
  such that for all \(f\in\CC(X,Y)\)
  \(J^\sharp([I,f]) = [ (f)_{j \in \Z} ]\).
  Moreover, for any \([S,f] \in \St(\CC)(X,Y)\), since  \(J^\sharp\) is a feedback functor, it follows that
  \(J^\sharp ([S,f]) = \shftZ(J^\sharp ([S,f]))\,.\)
\end{proofE}
\begin{remark}\label{rem:state:initialised}
  On the other hand,  unless \(\CC\) is a \emph{monoid}, there is no feedback functor \(F\colon \St(\CC) \to \St([\N, \CC], \shftN)\) acting on objects by \(X \mapsto (X, X, \dots)\), since we would then have that \((X, X, \dots) = F(1_{\CC}(X)) = \shftN(F(X)) = (I, X, X, \dots)\), and thus \(I=X\) for all objects \(X\).

  Nevertheless, one may still obtain an \(\N\)-indexed interpretation of stateful morphisms
  by equipping morphisms in \(\St\) with an initial state as in the work of Bonchi et al.~\cite[\S~III.E.]{bonchi_effectful_2025}, but this requires changing the definition of the feedback operation \cite[Def.~III.15]{bonchi_effectful_2025}.
  Here, we will use \(\St([\N, \CC], \shftN)\) directly as a syntax for initialised automata for convenience.
\end{remark}

Extensional streams provide a categorical semantics for Mealy
machines, though in a different sense from the state construction over
the base category: they identify processes whose state can be
propagated finitely far into the future or the past. The functor of
\Cref{lem:state:z-stateful} may therefore be seen as equipping the
more primitive semantics of the previous subsection with a notion of
discrete-time dynamics.

This finite propagation, however, is also a defect. If a process is
found to be inconsistent at some point in time, then the process as a
whole should be interpreted as inconsistent. Because extensional streams
permit morphisms to slide only \emph{finitely} far into the past or
future, an inconsistency propagates only finitely far, and it follows
immediately that:
\begin{lemma}\label{lem:state:dinatural-zero}
  Take a symmetric monoidal category \(\CC\) with monoidal zero morphisms. Consider  a representative dinatural stream \([\vb f]\) in  \(\St([\Z, \CC],\shftZ)\) or \(\St([\N, \CC], \shftN)\) with some \(f_j = 0\). Then in general \([\vb f]\) will not be a monoidal zero morphism.
\end{lemma}

This is only one example of the many shortcomings of extensional streams.
For example, using the feedback operator, any endomorphism can be turned into a scalar that never interfaces with the inputs or outputs of the stream. Consequently, extensionally indistinguishable processes need not lie in the same equivalence class.
However many of these problems are resolved via taking a more sophisticated coalgebraic approach.

\subsection{Coalgebraic streams}\label{sec:state:coalgebraic}
Stateful morphism sequences only allow for the internal state at one time to be reparametrised finitely far into the past or the future. Therefore, in a certain sense, the induced equivalence relation is very coarse as it does not account for truly infinite behaviour. In order to give a properly infinite stream semantics, in this subsection, we recall the coalgebraic notion of monoidal streams first developed by Di Lavore et al.~\cite[Def.~4.1]{lavore_monoidal_2022} and later generalised by Bonchi et al.~\cite[Def.~IV.4.]{bonchi_effectful_2025}.

Given a symmetric monoidal category \(\CC\) and a symmetric monoidal subcategory \(\Gamma\colon\DD\rightarrowtail \CC\), Bonchi et al.~\cite{bonchi_effectful_2025} define a feedback category, which we denote by  \(\Stream_{\DD}(\CC)\), whose objects are infinite lists of objects in \(\CC\), where the hom-sets are given by the final fixpoint of profunctors:\footnote{In fact they give a more general definition for premonoidal categories; however, we are only interested in monoidal setting in this paper.}
\begin{equation}\label{eq:state:stream-fixpoint}
    \Stream_{\DD} (\CC)(\vb X, \vb Y)\cong
    \;\;
    \int^{M \in \DD}
    \CC \bigl(X_0, (\Gamma M) \otimes Y_0 \bigr)
    \times
    \Stream_{\DD}(\CC) \bigl((\Gamma M) \cdot \vb X^+,\vb Y^+ \bigr),
\end{equation}
where given  \(\vb X, \vb Y \in \Ob(\CC)^\N\) and  \(N\in \Ob(\CC)\),  we denote
\[N \cdot \vb X \coloneqq (N \otimes X_0, X_1, X_2, \dots),\quad\text{and}\quad   \vb X^+ \coloneqq (X_1, X_2, \dots)\,.\]

In case \(\Gamma = 1_{\CC }\) denote this category by \(\Stream(\CC)\) as in \cite{lavore_monoidal_2022}.
The category \(\Stream(\CC)\)  elegantly captures total stateful processes;  for example, stateful processes built from Cartesian monoidal categories or Markov categories with ranges and conditionals \cite[Rem.~4.10]{lavore_monoidal_2022}.
However, we argue that in the presence of compact-closure or monoidal zero morphisms, which are often present in nondeterministic partial and quantum settings, coalgebraic monoidal streams are poorly behaved.

\begin{textAtEnd}
  To further understand coalgebraic streams, we must further expose the work of Di Lavore et al.~\cite{lavore_monoidal_2022}.

  \begin{definition}[{\cite[Def.~4.2]{lavore_monoidal_2022}}]
    \label{def:state:stage}
    Fix a symmetric monoidal category \((\CC,\otimes,I)\) and \(\vb X,\vb Y\in\Ob(\CC)^{\N}\).
    Define \(M_{-1}\coloneqq I\) and  for all \(n\in\N\) define
    \[
      \Stage_n(\vb X,\vb Y)\coloneqq
      \int^{\vb{M}\in\CC^{n-1}}
      \prod_{i=0}^{n}\CC(M_{i-1}\otimes X_i,\ M_i\otimes Y_i),
    \]
    denoting representatives by  \([\ \langle f_0|\cdots|f_n|\ ] \in \Stage_n(\vb X,\vb Y) \).
    Define  the truncation maps by:
    \begin{gather*}
      \pi_n\colon\Stage_{n+1}(\vb X,\vb Y)\to\Stage_n(\vb X,\vb Y); \quad
      \langle f_0|\cdots|f_{n+1}| \mapsto  \langle f_0|\cdots|f_n| .
    \end{gather*}
  \end{definition}

  Next we recall the following notion under which coalgebraic streams are ``well-behaved'':
  \begin{definition}[{\cite[Def.~9]{lavore_monoidal_2022}}]
    \label{def:state:productive}
    A SMC is \textbf{productive} if for all \(\alpha\in\Stage_0(\vb X,\vb Y)\) there exists a representative
    \(\alpha_0\colon X_0\to M_0\otimes Y_0\) such that, for every representative \(\alpha'\colon X_0\to M\otimes Y_0\) of \(\alpha\),
    there exists \(s\colon M_0\to M\) with \(\alpha'=(s\otimes 1_{Y_0})\circ\alpha_0\).
  \end{definition}

  \begin{lemma}\label{lem:state:productive}
    The following  SMCs are productive:
    \begin{itemize}
      \item Cartesian monoidal categories;
      \item compact-closed categories;
      \item SMCs with monoidal zero morphisms.
    \end{itemize}
  \end{lemma}
  \begin{proof}
    The first two claims are proved by Di Lavore et al.~\cite[Remark~4.10]{lavore_monoidal_2022}.
    Therefore, we only prove the final claim involving monoidal zero morphisms.

    Fix \(\vb X,\vb Y\in\Ob(\CC)^{\N}\) and let \(\alpha\in\Stage_0(\vb X,\vb Y)\).
    Choose a representative \(\alpha=\langle \alpha_0|\) with
    \(\alpha_0\colon X_0\to M\otimes Y_0\) for some \(M\in\Ob(\CC)\).
    The axioms of monoidal zero morphisms imply that
    \[
      0_{M,I}\otimes 1_{Y_0}
      =
      \swap_{I,Y_0}\circ(1_{Y_0}\otimes 0_{M,I})\circ \swap_{M,Y_0}
      =
      0_{M\otimes Y_0,\,I\otimes Y_0}.
    \]
    Moreover, by dinaturality
    \(
      \langle \alpha_0|
      =
      \langle (0_{M,I}\otimes 1_{Y_0})\circ \alpha_0|
      =
      \langle 0_{X_0,\,I\otimes Y_0}|.
    \) 
    Therefore \(\Stage_0(\vb X,\vb Y)\) is a singleton set. Hence every \(0\)-stage process is terminating so that \(\CC\) is productive.
  \end{proof}

  \begin{lemma}[{\cite[Thm.~4.11.]{lavore_monoidal_2022}}]
    \label{lem:state:productive-stage}
    If \(\CC\) is productive, then for all \(\vb X,\vb Y\in\Ob(\CC)^{\N}\) there is a canonical bijection
    \(\Stream(\CC)(\vb X,\vb Y)\cong \lim_{n\in\N}\Stage_n(\vb X,\vb Y)\),
    where the limit is taken along the truncations \(\pi_n\).
  \end{lemma}
\end{textAtEnd}

In the presence of compact-closure, coalgebraic monoidal streams without a guarded subcategory are degenerate:
\begin{lemmaE}
  \label{lem:state:nogo-compact-closed}
  \(\Stream(\CC)\) is indiscrete for compact-closed \(\CC\).
\end{lemmaE}
\begin{proofE}
  Fix \(\vb X,\vb Y\in\Ob(\CC)^{\N}\).
  Take some
  \(
    \langle f_0|f_1|\cdots|f_n|\in \Stage_n(\vb X,\vb Y)
  \)
  with witnesses \newline\(f_i\colon M_{i-1}\otimes X_i\to M_i\otimes Y_i\) and \(M_{-1}\coloneqq I\).
  Define
  \[
    \eta_{M,X,Y}\coloneqq 1_{M\otimes X}\otimes(\swap_{Y,Y^*}\circ\cup_Y)\colon
    M\otimes X\to (M\otimes X\otimes Y^*)\otimes Y,
  \]
  Moreover for any \(f\colon M\otimes X\to N\otimes Y\), define
  \[
    s_f\coloneqq (1_N\otimes \cap_Y)\circ(1_N\otimes \swap_{Y,Y^*})\circ(f\otimes 1_{Y^*})
    \colon M\otimes X\otimes Y^*\to N.
  \]
  Therefore it follows that \((s_f\otimes 1_Y)\circ \eta_{M,X,Y}=f\).
  Write \(N_{-1}\coloneqq I\) and \(N_k\coloneqq N_{k-1}\otimes X_k\otimes Y_k^*\).
  Then by dinaturality of the coends and the above factorisation, we have that
  \begin{align*}
    \langle f_0| \cdots|f_{n-1}|f_n|
    \approx{} &
    \langle f_0| \cdots|f_{n-1}|\eta_{M_{n-1},X_n,Y_n}|\\
    \approx{}&
    \langle f_0| \cdots|f_{n-2}|\eta_{M_{n-2},X_{n-1},Y_{n-1}}|
    \eta_{M_{n-2}\otimes X_{n-1}\otimes Y_{n-1}^*,\,X_n,\,Y_n}|\\
    \vdots\; & \\
    ={} &
    \langle \eta_{N_{-1},X_0,Y_0}|\eta_{N_0,X_1,Y_1}|\cdots|\eta_{N_{n-1},X_n,Y_n}|.
  \end{align*}
  Therefore for all \(n\in\N\), \(\Stage_n(\vb X,\vb Y)\) is a singleton set.
  By \Cref{lem:state:productive}, \(\CC\) is productive; therefore by \Cref{lem:state:productive-stage} we have \(\Stream(\CC)\) is indiscrete.
\end{proofE}

As remarked by Bonchi et al.~\cite[Rem.~V.12]{bonchi_effectful_2025} this  problem  can be sidestepped by working in \(\Stream_{\mathcal D}(\CC)\)
for some non-compact-closed symmetric monoidal subcategory \(\mathcal D\) of \(\mathcal C\). When \(\CC\) is a partial Markov category, they choose \(\DD\) to be the category of total maps \(\Tot(\CC)\).

Without a guard, in the presence of monoidal zero morphisms, coalgebraic monoidal streams are also degenerate:
\begin{lemmaE}\label{lem:state:nogo-zero}
  If \(\CC\) is a SMC with monoidal zero morphisms, then \(\Stream(\CC)\) is indiscrete.
\end{lemmaE}
\begin{proofE}
  Fix \(\vb X,\vb Y\in\Ob(\CC)^{\N}\).
  Take some
  \(\langle f_0|f_1|\cdots|f_n| \in \Stage_n(\vb X,\vb Y)\)
  with witnesses\newline
  \(f_i\colon M_{i-1}\otimes X_i\to M_i\otimes Y_i\) and \(M_{-1}\coloneqq I\).
  By dinaturality and the coherences for monoidal zero morphisms we have that
  \begin{align*}
    \langle f_0| \cdots|f_{n-1}|f_n|
    \approx{}&
    \langle f_0| \cdots|f_{n-1}|(0_{M_{n},I}\otimes 1_{Y_n})\circ f_n|\\
    ={} & \langle f_0| \cdots|f_{n-1}|0_{M_{n}\otimes X_n,I\otimes Y_n}|\\
    \vdots\;  & \\
    ={} & \langle 0_{X_0, Y_0}| \cdots|0_{X_{n-1}, Y_{n-1}}|0_{X_n, Y_n}|.
  \end{align*}
  Therefore for all \(n\in \N\),  \(\Stage_n(\vb X,\vb Y)\) is a singleton set.
  By \Cref{lem:state:productive}, \(\CC\) is productive; therefore by \Cref{lem:state:productive-stage}, \(\Stream(\CC)\) is indiscrete.
\end{proofE}
This problem is also addressed for partial Markov categories by guarding over the category of total maps.
However, guarding over the total maps disallows the delay from being a natural isomorphism.
Similarly, when one takes this approach of guarding over the total maps, just as is the case for extensional streams in \Cref{lem:state:dinatural-zero}, the presence of a zero morphism at some time in the stream will not make the stream itself zero, and thus will not detect inconsistencies in a satisfactory manner.  

\subsection{Causal and monotone sequences}\label{sec:state:causal-monotone}

Another approach to categorical semantics for stateful processes is to represent a stream \emph{concretely} as a family of morphisms in the base category of \emph{increasingly larger} or \emph{increasingly informative} approximations.

First, consider the following notion due to Raney \cite{Raney1958}, which constructs stream transducers so that the value of the stream at every time step does not depend on the value of the stream in the future. Let \(A\) and \(B\) be sets. A function \(f\colon A^{\N}\to B^{\N}\) is \emph{causal} (called sequential by Raney) if for all \(x,y\in A^{\N}\) and \(n\in\N\), 
\(\bigl(\forall i\le n\colon\ x_i=y_i\bigr)\ \Longrightarrow\ \bigl(\forall i\le n\colon\ f(x)_i=f(y)_i\bigr)\).

Categorifying this concept we obtain:

\begin{definition}\label{def:state:causal-seq}
  Let \(\CC\) be a discard category. The affine symmetric monoidal category of \textbf{causal sequences} \(\Caus_{\N}(\CC)\) has:
  \begin{itemize}
    \item \textbf{Objects}: infinite sequences \(\vb X\in\Ob(\CC)^{\N}\);
    \item \textbf{Morphisms}:  \(\vb f\colon\vb X\to \vb Y\) are given by \(\N\)-indexed sequences of morphisms in \(\CC\) with components
    \(
      \{f_t\colon X_0 \otimes \dots \otimes X_t \to Y_0 \otimes \dots \otimes Y_t\}_{t\in \N}
    \)
    such that for all \(t\in\N\),
    \begin{equation}\label{eq:state:strict-causal-seq}
      (1_{Y_0 \otimes \dots \otimes Y_{t}} \otimes \discard_{Y_{t + 1}}) \circ f_{t + 1}
      =
      f_t \otimes \discard_{X_{t+1}};
    \end{equation}
  \end{itemize}
  where the rest of the affine symmetric monoidal category  structure is defined pointwise.
\end{definition}

This notion of causality is well-behaved in categories where every morphism is total. For example, it recovers the causal functions of Katsumata and Sprunger \cite[Thm.~14]{sprunger} in Cartesian categories, and controlled stochastic processes in \(\FinStoch\) \cite[Def.~7.1]{lavore_monoidal_2022}.

However, this is not the correct notion in categories of partial processes. For example, a process \(f\) in \(\Caus_{\N}(\Par)\) such that \(f_0\) is total, must remain total for all \(n > 0\). Similarly, for any constant sequence \(\{f^{\otimes n}\} \in \Caus_{\N}(\CC)\) we must have \(\discard_Y \circ f = \discard_X\), i.e. \(f\) must be total. Therefore, causal streams do not accommodate inconsistencies, which suggests that the equality from \Cref{def:state:causal-seq} needs to be relaxed to an \emph{inequality}

\begin{equation}\label{eq:state:inequality}
  (1_{Y_0 \otimes \dots \otimes Y_{t}} \otimes \discard_{Y_{t + 1}}) \circ f_{t + 1} \subseteq   f_t \otimes \discard_{X_{t+1}}.
\end{equation}

Finally, with all of these examples in mind, we can relax \Cref{def:state:causal-seq} to the preorder-enriched setting:
\begin{definition}\label{def:state:monotone-seq}
  Let \(\CC\) be a preorder-enriched discard category. The discard category of  \textbf{monotone sequences} \(\Mon_{\N}(\CC)\) has the same objects as \(\Caus_{\N}(\CC)\). A morphism \(f\colon X\to Y\) is a {monotone sequence}, i.e.\ a family of morphisms
  \(
    \{f_t\colon X_0 \otimes \dots \otimes X_t \to Y_0 \otimes \dots \otimes Y_t\}_{t\in \N}
  \) satisfying the inequality of \Cref{eq:state:inequality},
  where the rest of the discard category structure is defined pointwise.
\end{definition}

Bonchi et al.\ show \cite[Prop.~6.7]{bonchi_effectful_2025} that, in a Markov category with conditionals, their notion of ``causal sequences'' (which conflicts with the terminology that we use in this paper) are exactly monotone sequences.
Carette et al. define a similar construction for discard categories with pseudo-purification, where they impose moreover that every sequence satisfies the very strong assumption of regularity \cite[Def.~8]{delayedtracememory}.

Monotone sequences accommodate inconsistency: they can contain a component \(f_t=0\) without forcing the whole family \(\{f_t\}_{t\in\N}\) to be the pointwise-zero morphism in \(\Mon_{\N}(\CC)\).
However, inconsistencies are still not detected at the level of the whole stream: an inconsistency that appears only at a later time step changes only later components in the family.

This motivates imposing an equivalence relation which takes the information theoretic interpretation of preorder enrichment more seriously,
identifying two monotone sequences in case they are compatible and neither is more informative than the other.

\section{Observational semantics}
\label{sec:obs}
\pratendSetLocal{category=obs, end, restate}
In the previous section, we reviewed various semantics for stateful processes, and saw that each makes a different compromise. Stateful morphism sequences are purely local: a constraint enforced at one time step only propagates finitely far into the future or the past. Coalgebraic streams sit at the opposite, where a fixed-point equation constrains the entire infinite behaviour at once, but collapses in the presence of partiality and nondeterminism. Causal and monotone sequences are concrete and well-behaved, but they fail to identify dinaturally equivalent streams.

In this section, we build a categorical semantics in which the \emph{global behaviour} of a process is determined by the totality of its \emph{local observations}, glued together by a consistency condition: observations made in a larger context may only \emph{refine} those made in a smaller one. Two processes are then identified when they are \emph{observationally indistinguishable}, that is, when every observation of one can be verified by observing the other, possibly in a larger context. We show that this notion satisfies the desiderata set out in \Cref{sec:intro}. A further advantage over existing stream semantics is that the contexts need not be discrete intervals of time. We can therefore interpret processes evolving through more general spaces.

\subsection{The bicategory of behaviours}\label{sec:obs:behaviours}

Our semantic construction is parametrised by an arbitrary \textbf{index set} \(\I\) and a collection \(\JJ \subseteq \Pfin(\I)\) of \emph{finite} subsets of \(\I\) called \textbf{contexts}, which define the local finite regions from which a global system can be observed.
Given two contexts \(j \subseteq j'\), we denote the set-theoretic difference by  \(j' \setminus j\). We assume a total order on \(\I\) throughout so that the indexed monoidal product \(\otimes_{i \in j}\) is well-defined for any subset \(j \in \JJ\). We often omit the symmetry morphisms for simplicity and readability, or explicitly use the \textbf{shuffle morphisms}, defined for arbitrary \(\vb{X}, \vb{Y} \in (\Ob(\CC))^\I\) and contexts \(j \in \JJ\) as the following isomorphisms:
\[(\shuffle_{\vb{X},\vb{Y}})_{j}\colon\big(\smbigotimes_{k \in j} X_k\big)\otimes\big(\smbigotimes_{k \in j} Y_k\big)\cong \smbigotimes_{k \in j}(X_k\otimes Y_k).\]
We represent consistent families of observations using the following structure:
\begin{definition}\label{def:obs:monotone-family}
  Take a discard bicategory  \(\CC\), index set \(\I\), a set of contexts \(\JJ\subseteq \Pfin(\I)\), and a   \(\I\)-indexed family of objects \(\vb{X}\) and \(\vb{Y}\) in \(\CC\).

  A \textbf{monotone family} \(f\colon \vb{X} \to \vb{Y}\) in \(\CC\) is a \(\JJ\)-indexed family of morphisms
  \[\{f_j\colon \smbigotimes_{i \in j} X_i \to \smbigotimes_{i \in j} Y_i \}_{j \in \JJ},\]
  where \(Y_{j} \coloneqq \bigotimes_{k \in j} Y_{k}\) and \(X_{j} \coloneqq \bigotimes_{k \in j} X_{k}\),
  such that for all \(j \subseteq j' \in \JJ\):
  \begin{equation*}
    (1_{Y_j} \otimes \discard_{Y_{j' \setminus j}}) \circ f_{j'} \subseteq f_j \otimes \discard_{X_{j' \setminus j}}\, .
  \end{equation*}

  A monotone family is \textbf{causal} when all  inclusions are equalities.
\end{definition}

Note that monotone families generalise the notion of monotone sequences. A monotone sequence is a  monotone family with
\(\I \coloneqq \N\) and \(\JJ \coloneqq \{[0,n] \ | \ \forall n\in \N\}\).

We interpret a context as a finite region in which properties of some global system can be observed (for monotone sequences, this finite region is interpreted as an interval of time). Under this interpretation, the condition of monotonicity imposes that the observations made in any context \emph{approximate} the observations which can be made in a larger context, so that all observations are consistent with each other.

We introduce a preorder between monotone families that captures when one family of observations approximates another:
\begin{definition}\label{def:obs:approx}
  Given two monotone families \(\vb f,\vb g\colon\vb X\to\vb Y\), say that \(\vb g\) \textbf{approximates} \(\vb f\), written
  \(\vb f\subseteq \vb g\), if for every \(j\in\JJ\) there exists some
  \(j'\in\JJ\) with \(j\subseteq j'\) such that
  \[
    (1_{Y_j}\otimes \discard_{Y_{j'\setminus j}})\circ f_{j'}
    \;\subseteq\;
    g_j\otimes \discard_{X_{j'\setminus j}}\, .
  \]
\end{definition}

\begin{lemmaE}\label{lem:obs:approx-preorder}
  Approximation is a preorder.
\end{lemmaE}
\begin{proofE}
  We need to prove that:
  \begin{enumerate}
    \claimitem{obs:approx:refl}{approximation is reflexive;}
    \claimitem{obs:approx:trans}{approximation is transitive.}
  \end{enumerate}

  \claimparagraph{obs:approx:refl}{Reflexivity.}
  Immediate by taking \(j'=j\).

  \claimparagraph{obs:approx:trans}{Transitivity.}
  Suppose \(\vb f\subseteq \vb g\) and \(\vb g\subseteq \vb h\).
  Fix \(j\in\JJ\).
  By \(\vb g\subseteq \vb h\), there exists \(j_0\in\JJ\) with \(j\subseteq j_0\) such that
  \(
    (1_{Y_j}\otimes \discard_{Y_{j_0\setminus j}})\circ g_{j_0}
    \subseteq
    h_j\otimes \discard_{X_{j_0\setminus j}}\, .
  \)
  Moreover, since \(\vb f\subseteq \vb g\), by restricting to the context \(j_0\), there exists some
  \(j_1\in\JJ\) with \(j_0\subseteq j_1\) such that
  \(
    (1_{Y_{j_0}}\otimes \discard_{Y_{j_1\setminus j_0}})\circ f_{j_1}
    \subseteq
    g_{j_0}\otimes \discard_{X_{j_1\setminus j_0}}\, .
  \)
  Therefore
  \(
    (1_{Y_j}\otimes \discard_{Y_{j_1\setminus j}})\circ f_{j_1}
    \subseteq
    h_j\otimes \discard_{X_{j_1\setminus j}}\,,
  \)
  so that \(\vb f\subseteq \vb h\).
\end{proofE}

As the name suggests, we interpret \(\vb f\subseteq \vb g\) to denote that the observations made in \(\vb f\) refine those made in \(\vb g\), so that \(\vb f\) is a more accurate description of the underlying process which is being observed.

We define a discard bicategory of monotone families ordered by approximation:

\begin{definition}\label{def:obs:mon-family}
  Given a set of contexts \(\JJ \subseteq \Pfin(\I)\) and a discard bicategory \(\CC\),
  the discard bicategory of \textbf{monotone families}, \(\Mon_{\I,\JJ}(\CC)\),  has:
  \begin{itemize}
    \item \textbf{objects}: families \(\vb X=\{X_i\}_{i\in\I}\);
    \item \textbf{morphisms}: monotone families \(\vb f\colon\vb X\to\vb Y\);
    \item \textbf{poset-enrichment}:  given by approximation;
    \item \textbf{discard category structure:} contextwise.
  \end{itemize}
\end{definition}
\begin{propositionE}\label{prop:obs:mon}
  \(\Mon_{\I,\JJ}(\CC)\) is a discard bicategory.
\end{propositionE}
\begin{proofE}
  We need to show that:
  \begin{enumerate}
    \claimitem{obs:mon:structural}{structural morphisms are monotone families;}
    \claimitem{obs:mon:comp}{composition of monotone families is monotone;}
    \claimitem{obs:mon:tensor}{monoidal product of monotone families is monotone;}
    \claimitem{obs:mon:axioms}{the discard bicategory axioms are satisfied.}
  \end{enumerate}

  \claimparagraph{obs:mon:structural}{Structural morphisms are monotone.}
  Identities, symmetry, and discard are defined contextwise by tensoring the corresponding
  structural morphisms in \(\CC\):
  \begin{gather*}
    (1_{\vb X})_j = 1_{X_j}\, ,\quad
    (\discard_{\vb X})_j = \bigotimes_{i\in j}\discard_{X_i}\, , \quad\text{and}\quad
    (\swap_{\vb X,\vb Y})_j = \bigotimes_{i\in j}\swap_{X_i,Y_i}\, .
  \end{gather*}

  Take contexts \(j\subseteq j'\).
  All structural morphisms  \(f\) in \(\CC\) are total meaning that \(\discard\circ f=\discard\); therefore, since all structural morphisms in \(\Mon_{\I,\JJ}(\CC)\) are built from those in \(\CC\) we have that
  \[
  (1_{Y_j}\otimes \discard_{Y_{j'\setminus j}})\circ f_{j'}
  =
  f_j\otimes \discard_{X_{j'\setminus j}}\, .
  \]
  Therefore all structural morphisms are causal (and thus monotone) families.

  \claimparagraph{obs:mon:comp}{Composition preserves monotonicity.}
  Take two monotone \(\vb f\colon\vb X\to\vb Y\) and \(\vb g\colon\vb Y\to\vb Z\) families, and fix contexts \(j\subseteq j'\) in \(\JJ\).
  Using monotonicity of \(\vb f\) and \(\vb g\), it follows that
  \begin{align*}
  (1_{Z_j}\otimes \discard_{Z_{j'\setminus j}})\circ (g_{j'}\circ f_{j'})
  =
  \bigl((1_{Z_j}\otimes \discard_{Z_{j'\setminus j}})\circ g_{j'}\bigr)\circ f_{j'}
  \subseteq
  \bigl(g_j\otimes \discard_{Y_{j'\setminus j}}\bigr)\circ f_{j'}
  \subseteq
  (g_j\circ f_j)\otimes \discard_{X_{j'\setminus j}}\, .
  \end{align*}
  Therefore \(\vb g\circ\vb f\) is monotone.

  \claimparagraph{obs:mon:tensor}{Monoidal product preserves monotonicity.}
  Take two  monotone families \(\vb f\colon\vb X\to\vb Y, \vb h\colon\vb X'\to\vb Y'\), and fix  contexts \(j\subseteq j'\). It follows that
  \begin{align*}
  &(1_{(Y\otimes Y')_j}\otimes \discard_{(Y\otimes Y')_{j'\setminus j}})
  \circ (\shuffle_{\vb Y,\vb Y'})_j\circ (f_{j'}\otimes h_{j'})\circ (\shuffle_{\vb X,\vb X'})_j^{-1} \\
  &\quad =
  (\shuffle_{\vb Y,\vb Y'})_j
  \circ (1_{Y_j}\otimes \discard_{Y_{j'\setminus j}}\otimes
  1_{Y'_j}\otimes \discard_{Y'_{j'\setminus j}})
  \circ (f_{j'}\otimes h_{j'})\circ (\shuffle_{\vb X,\vb X'})_j^{-1}\\
  &\quad \subseteq
  (\shuffle_{\vb Y,\vb Y'})_j\circ (f_j\otimes h_j)\circ (\shuffle_{\vb X,\vb X'})_j^{-1}.
  \end{align*}
  Therefore \(\vb f\otimes \vb h\) is monotone.

  \claimparagraph{obs:mon:axioms}{Discard bicategory axioms hold.}
  \ref{item:obs:mon:structural},\ref{item:obs:mon:comp}, and \ref{item:obs:mon:tensor} imply that approximation is a monoidal preorder enrichment; where all other discard-bicategory axioms, including lax naturality of the discard, hold contextwise.
\end{proofE}

In particular, we have that \(\Mon_{\{[0,n] \ | \ \forall n\in \N\}, \N} (\CC)\cong \Mon_\N(\CC)\); however, by changing the set of contexts we can model processes evolving through more general spaces, rather than merely time.

Consider the equivalence relation on monotone families generated by mutual approximation:

\begin{definition}\label{def:obs:obs-equiv}
  We say that \(\vb{f}\)  and \(\vb{g}\)  are \textbf{observationally equivalent}, in case \(\vb{f}\) and  \(\vb{g}\) approximate each other,
  i.e. \(\vb{f} \subseteq \vb{g}\) and \(\vb{g} \subseteq \vb{f}\).
\end{definition}
We interpret two processes to be observationally equivalent in case their observations never contradict each other; so that every observation of one process can be shown to approximate an observation in the other and vice versa.

For monotone families indexed by arbitrary sets of contexts, observational equivalence is not necessarily a congruence with respect to composition.  To ensure that this holds, we ask that the set of contexts possesses the following property:
\begin{definition}\label{def:obs:directed}
  \(\JJ \subseteq \Pfin(\I)\)  is \textbf{upward-directed}, in case for any pair of contexts \(j, j' \in \JJ\) there exists some \(j'' \supseteq j \cup j'\) in \(\JJ\).
  Call a monotone family a \textbf{monotone net}\footnote{We use the term ``monotone net'' as a categorification of the notion of a ``net'' coined by Kelley, which denotes a function whose domain is a directed set \cite{Kelley1950}.} if it is indexed by an upward-directed set of contexts.
\end{definition}

Quotienting the category of monotone \emph{nets} by observational equivalence yields the following category, which is the central construction of this article:

\begin{definition}\label{def:obs:obs}
  Given an upward-directed set of contexts \(\JJ \subseteq \Pfin(\I)\) and a discard bicategory \(\CC\),
  the (poset-enriched) discard bicategory of \textbf{behaviours}, \(\Obs_{\I,\JJ}(\CC)\), is given by quotienting \(\Mon_{\I,\JJ}(\CC)\) by observational equivalence.
\end{definition}

\begin{theoremE}\label{thm:obs:obs}
  \(\Obs_{\I,\JJ}(\CC)\) is a poset-enriched discard bicategory.
\end{theoremE}
\begin{proofE}
  We need to show that observational equivalence is a congruence for:
  \begin{enumerate}
    \claimitem{obs:obs:approx}{approximation;}
    \claimitem{obs:obs:comp}{composition;}
    \claimitem{obs:obs:tensor}{monoidal product.}
  \end{enumerate}
  From which it follows that the discard bicategory structure of \(\Mon_{\I,\JJ}(\CC)\) from \Cref{prop:obs:mon} descends to the quotient, where the preorder enrichment becomes a poset since mutually-approximating families are identified.

  \claimparagraph{obs:obs:approx}{Congruence for approximation.}
  This follow by transitivity.

  \claimparagraph{obs:obs:comp}{Congruence for composition.}
  Take composable \(\vb f\subseteq \vb f'\) and \(\vb g\subseteq \vb g'\) and fix some \(j\in\JJ\).
  Choose \(j_f,j_g\in\JJ\) with \(j\subseteq j_f\) and \(j\subseteq j_g\)
  witnessing that \(\vb f'\) approximates \(\vb f\) and \(\vb g'\) approximates \(\vb g\).
  Using upward directedness of \(\JJ\), let \(j'\in\JJ\) be such that \(j_f\cup j_g\subseteq j'\).
  Then:
  \begin{align*}
    (1_{Z_j}\otimes \discard_{Z_{j'\setminus j}})\circ g_{j'}\circ f_{j'}
    \subseteq
    g'_j\circ (1_{Y_j}\otimes \discard_{Y_{j'\setminus j}})\circ f_{j'}
    \subseteq
    (g'_j\circ f'_j)\otimes \discard_{X_{j'\setminus j}}\,.
  \end{align*}
  Thus \(\vb g\circ\vb f\subseteq \vb g'\circ\vb f'\).

  \claimparagraph{obs:obs:tensor}{Congruence for monoidal product.}
  Take \(\vb f\subseteq \vb f'\) and \(\vb h\subseteq \vb h'\).
  Fix \(j\in\JJ\).
  Choose \(j_f,j_h\in\JJ\) with \(j\subseteq j_f\) and \(j\subseteq j_h\)
  witnessing that \(\vb f'\) approximates \(\vb f\) and \(\vb h'\) approximates \(\vb h\).
  Using upward directedness of \(\JJ\), let \(j'\in\JJ\) be such that \(j_f\cup j_h\subseteq j'\).
  Then:
  \begin{align*}
    (1_{(Y\otimes Y')_j}\otimes \discard_{(Y\otimes Y')_{j'\setminus j}})
    \circ (\shuffle_{\vb Y,\vb Y'})_j\circ (f_{j'}\otimes h_{j'})\circ (\shuffle_{\vb X,\vb X'})_j^{-1}
    \subseteq
    (\shuffle_{\vb Y,\vb Y'})_j\circ (f'_j\otimes h'_j)\circ (\shuffle_{\vb X,\vb X'})_j^{-1}.
  \end{align*}
  Thus \(\vb f\otimes\vb h\subseteq \vb f'\otimes\vb h'\).
\end{proofE}

The construction of behaviours is functorial:

\begin{propositionE}\label{prop:obs:functorial}
  Let \(F\colon\CC\to \DD\) be a discard bicategory functor.
  There is an induced discard bicategory functor
  \(
    \Obs_{\I,\JJ}(F)\colon\Obs_{\I,\JJ}(\CC)\to \Obs_{\I,\JJ}(\DD)
  \)
  defined contextwise.
  
  Moreover,  if \(F\) reflects the preorder then \(\Obs_{\I,\JJ}(F)\) is faithful. If \(F\) reflects the preorder and is full, then \(\Obs_{\I,\JJ}(F)\) is full and faithful.
\end{propositionE}

\begin{proofE}
  We need to show that \(\Obs_{\I,\JJ}(F)\):
  \begin{enumerate}
    \claimitem{obs:obs-fun:mono}{preserves the preorder on hom-sets;}
    \claimitem{obs:obs-fun:wd}{is well-defined on observational equivalence classes;}
    \claimitem{obs:obs-fun:funct}{is functorial and preserves the discard category structure;}
    \claimitem{obs:obs-fun:faith}{is faithful when \(F\) reflects the preorder;}
    \claimitem{obs:obs-fun:full}{is full when \(F\) reflects the preorder and is full.}
  \end{enumerate}

  \claimparagraph{obs:obs-fun:mono}{Preservation of monotonicity.}
  Fix contexts \(j\subseteq j'\) in \(\JJ\).
  By monotonicity of \(\vb f\), we have
  \[
    (1_{Y_j}\otimes \discard_{Y_{j'\setminus j}})\circ f_{j'}
    \subseteq
    f_j\otimes \discard_{X_{j'\setminus j}}\, .
  \]
  Using the fact that \(F\) preserves the preorder, it follows that
  \[
    F\bigl((1_{Y_j}\otimes \discard_{Y_{j'\setminus j}})\circ f_{j'}\bigr)
    \subseteq
    F\bigl(f_j\otimes \discard_{X_{j'\setminus j}}\bigr).
  \]
  Since \(F\) is a discard functor and symmetric monoidal, it follows that
  \[
    (1_{F(Y)_j}\otimes \discard_{F(Y)_{j'\setminus j}})\circ F(f_{j'})
    \subseteq
    F(f_j)\otimes \discard_{F(X)_{j'\setminus j}}\, .
  \]
  Therefore \(\{F(f_j)\}_{j\in\JJ}\) is a monotone net in \(\mathcal D\).

  \claimparagraph{obs:obs-fun:wd}{Well-definedness on equivalence class.}
  Take two monotone families \(\vb f\subseteq \vb g\) in \(\Mon_{\I,\JJ}(\CC)\), and fix some context \(j\in\JJ\).
  By definition of approximation, there exists \(j'\in\JJ\) with \(j\subseteq j'\) such that
  \(
    (1_{Y_j}\otimes \discard_{Y_{j'\setminus j}})\circ f_{j'}
    \subseteq
    g_j\otimes \discard_{X_{j'\setminus j}}\, .
  \)
  Therefore it follows that
  \[
    (1_{F(Y)_j}\otimes \discard_{F(Y)_{j'\setminus j}})\circ F(f_{j'})
    \subseteq
    F(g_j)\otimes \discard_{F(X)_{j'\setminus j}}\, .
  \]
  Thus \(\Obs_{\I,\JJ}(F)([\vb f])\subseteq \Obs_{\I,\JJ}(F)([\vb g])\); therefore, \(\Obs_{\I,\JJ}(F)\) is well-defined on observational equivalence classes.

  \claimparagraph{obs:obs-fun:funct}{Functoriality.}
  Functoriality holds contextwise.

  \claimparagraph{obs:obs-fun:faith}{Faithfulness.}
  Assume \(F\) reflects the preorder, and take two monotone families \(\vb f\) and \(\vb g\) such that
  \(\Obs_{\I,\JJ}(F)([\vb f])=\Obs_{\I,\JJ}(F)([\vb g]).\) 
  Then \(\Obs_{\I,\JJ}(F)([\vb f])\) and \(\Obs_{\I,\JJ}(F)([\vb g])\)
  approximate each other.

  Next, take some context \(j\in\JJ\).
  Then there exists some \(j'\supseteq j\) such that
  \[
  (1_{F(Y)_j}\otimes \discard_{F(Y)_{j'\setminus j}})\circ F(f_{j'})
  \subseteq
  F(g_j)\otimes \discard_{F(X)_{j'\setminus j}}\, .
  \]
  Rewrite both sides as \(F(u)\subseteq F(v)\).
  By reflection of the preorder, \(u\subseteq v\), hence \(\vb f\subseteq \vb g\).
  The reverse inclusion holds by symmetry, so that \([\vb f]=[\vb g]\).
  Thus \(\Obs_{\I,\JJ}(F)\) is faithful.

  \claimparagraph{obs:obs-fun:full}{Fullness.}
  Assume \(F\) reflects the preorder and is full.
  
  Let \([\vb h]\colon\Obs_{\I,\JJ}(F)(\vb X)\to\Obs_{\I,\JJ}(F)(\vb Y)\)
  be a morphism in \(\Obs_{\I,\JJ}(\mathcal D)\).
  For each \(j\in\JJ\), using fullness, choose some \(f_j\) in \(\CC\) such that
  \(F(f_j)=h_j\), taking \(\vb f\coloneqq \{f_j\}_{j\in\JJ}\).

  Monotonicity of \(\vb h\) and the reflection of the preorder
  imply monotonicity of \(\vb f\).
  Since \(F(f_j)=h_j\) for all \(j\in\JJ\), it follows that \(\Obs_{\I,\JJ}(F)([\vb f])=[\vb h]\).
\end{proofE}

\subsection{Propagation of inconsistency}\label{sec:obs:inconsistency}

Recall from \Cref{sec:state} that the existence of symmetric monoidal zero morphisms makes the category of coalgebraic streams satisfying \Cref{eq:state:stream-fixpoint} codiscrete. On the other hand, in extensional streams a zero morphism occurring at some time step cannot propagate infinitely far into the future, so the global process is not a monoidal zero morphism; and for monotone sequences, zero morphisms propagate only into the future, never the past. By contrast, behaviours \emph{detect inconsistencies} globally: if a process is inconsistent in any single finite context, then its whole behaviour is the zero morphism.
\begin{propositionE}\label{prop:obs:zero}
  Take a discard bicategory \(\CC\) with monoidal zero morphisms, and upward-directed \(\JJ\subseteq \Pfin(\I)\).
  Then \(\Obs_{I, \JJ}(\CC)\) has monoidal zero morphisms.
  Moreover, if a monotone net \(\vb f\colon\vb X\to\vb Y\) satisfies
  \(f_{j_0}=0\) for some \(j_0\in\JJ\), then
  \([\vb f]=[\{0\}_{j\in\JJ}]\) is a zero morphism.
\end{propositionE}
\begin{proofE}
  Fix  a contexts \(j_0,j, j'\in\JJ\) with \(j\cup j_0\subseteq j'\).
  By monotonicity of \(\vb f\), \(f_{j'}\) factors through \(f_{j_0}=0\),
  hence
  \(
    (1_{Y_j}\otimes \discard_{Y_{j'\setminus j}})\circ f_{j'}=0
  \).
  This witnesses \(\vb f\subseteq \{0\}\), where the converse is immediate.
\end{proofE}

\subsection{Dinaturality of state}\label{sec:obs:dinaturality}

We now turn to the question: \emph{under which conditions can stateful processes be interpreted as behaviours?} We answer it by exhibiting functors from \emph{stateful morphism sequences} into categories of behaviours. Crucially, these functors equate machines when they are related by dinatural reparametrisation of the internal state. Note that the same mapping onto monotone sequences would be ill-defined as dinaturally equivalent machines can give rise to distinct monotone sequences.
We first study the case of stateful processes with a specified initial state. Their behaviours are then indexed by the natural numbers.
\begin{definition}\label{def:obsfbk:initialised}
  Let \(\N_{\leq}^2 \coloneqq \{ [0,t]\ |\ t \in \N \} \) denote the set of  finite intervals in \(\N\) containing \(0\).
  Given a discard bicategory \(\CC\), denote the category of \textbf{initialised behaviours} over \(\CC\) by \(\Obs_\N(\CC) \coloneqq \Obs_{\N,  \N_{\leq}^2}(\CC)\).
\end{definition}
For convenience, we index morphisms in \(\Obs_\N(\CC)\) by \(t\in\N\), writing \(f_t\) for the component at context \([0,t]\); note that \(X_t\) denotes the single object at time~\(t\), whereas \(f_t\) has domain \(X_0\otimes\dots\otimes X_t\).

\begin{definition}
  Define the delay endofunctor  \(\shftN\colon\Obs_\N(\CC) \to \Obs_\N(\CC)\)  on objects by 
  \[
      \shftN(\vb{X}) \coloneqq \{X_{t-1}\}_{t \in \N},
    \quad \text{with}\quad 
      X_{-1} \coloneqq I,
  \]
  and on morphisms by
  \[\Bigl([\vb{f}]\colon\vb{X}\to \vb{Y}\Bigr)\quad  \longmapsto\quad \Bigl( [\{ 1_I \otimes f_{t-1} \}_{t \in \N}]\colon\shftN(\vb{X}) \to \shftN(\vb{Y})\Bigr).\]
\end{definition}

The following theorem justifies our claim that behaviours provide a categorical semantics for Mealy machines:
\begin{theoremE}\label{thm:obsfbk:trace-semantics}
  If \(\CC\) is a discard bicategory
  then there is a symmetric monoidal functor \(\St([\N, \CC], \shftN) \to \Obs_\N(\CC)\).
\end{theoremE}
\begin{proofE}
  The functor \(U\colon \St([\N, \CC], \shftN) \to \Obs_\N(\CC)\) is defined on a representative stateful morphism sequence by unrolling and discarding the output memory. By lax naturality, it is easy to show that this gives a monotone sequence and thus a valid initialised behaviour. To show that two representatives of the same equivalence class of \( \St([\N, \CC], \shftN)\) yield the same initialised behaviour, the only non-trivial thing to prove is that the sliding equation holds.
  Let the grey circle denote the discard. We use lax naturality to prove the first inequality
  \[\begin{tikzpicture}
	\begin{pgfonlayer}{nodelayer}
		\node [style=none] (21) at (-0.025, -0.125) {\(\rotatebox{-90}{\(\subseteq\)}\)};
		\node [style=none] (72) at (0, 6.25) {\(U\bigl(\fbk^{\vb{S}} \bigl((\vb{r}\otimes 1_{\vb Y})\circ \vb{f} \bigr) \bigr)_t\)};
		\node [style=none] (73) at (0, 5.375) {\(\rotatebox{-90}{\(=\)}\)};
		\node [style=none] (74) at (0, -6.5) {\(    U\bigl(\fbk^{\vb{S}} \bigl(\vb{f} \circ (\shftN(\vb{r})\otimes 1_{\vb X}) \bigr) \bigr)_t\,,\)};
		\node [style=none] (75) at (0, -5.625) {\(\rotatebox{-90}{\(=\)}\)};
		\node [style=box] (147) at (-7.25, 3.9) {\(f_0\)};
		\node [style=none] (148) at (7.5, 4.775) {};
		\node [style=none] (149) at (-8.5, 4.775) {};
		\node [style=none] (150) at (7.5, 0.475) {};
		\node [style=none] (151) at (-8.5, 0.475) {};
		\node [style=none] (168) at (-7.5, 4.275) {};
		\node [style=none] (169) at (-5.5, 4.275) {};
		\node [style=none] (170) at (-7.5, 3.525) {};
		\node [style=none] (171) at (-5.5, 3.525) {};
		\node [style=none] (172) at (-4.75, 3.525) {};
		\node [style=none] (173) at (-4.75, 4.275) {};
		\node [style=box] (174) at (-4.25, 3.15) {\(f_1\)};
		\node [style=none] (175) at (-2.5, 3.525) {};
		\node [style=none] (176) at (-8.5, 2.775) {};
		\node [style=none] (177) at (-2.5, 2.775) {};
		\node [style=none] (178) at (-1.75, 2.775) {};
		\node [style=none] (179) at (-1.75, 3.525) {};
		\node [style=none] (180) at (-8.5, 3.525) {};
		\node [style=box] (181) at (3.5, 1.65) {\(f_t\)};
		\node [style=none] (182) at (4.25, 2.025) {};
		\node [style=none] (183) at (-8.5, 1.275) {};
		\node [style=none] (184) at (4.25, 1.275) {};
		\node [style=none] (185) at (5, 1.275) {};
		\node [style=none] (186) at (5, 2.025) {};
		\node [style=none] (187) at (7.5, 4.275) {};
		\node [style=none] (188) at (7.5, 3.525) {};
		\node [style=none] (190) at (1.875, 1.65) {\scalebox{.8}{\(t\)}};
		\node [style=none] (191) at (2.525, 2.025) {};
		\node [style=none] (192) at (1.875, 1.95) {\(\cdots\)};
		\node [style=box] (194) at (-1.25, 2.4) {\(f_2\)};
		\node [style=none] (195) at (0.5, 2.775) {};
		\node [style=none] (196) at (-8.5, 2.025) {};
		\node [style=none] (197) at (0.5, 2.025) {};
		\node [style=none] (198) at (1.25, 2.025) {};
		\node [style=none] (199) at (1.25, 2.775) {};
		\node [style=none] (200) at (7.5, 2.775) {};
		\node [style=box] (201) at (-6, 4.25) {\(r_{0}\)};
		\node [style=box] (202) at (-3, 3.5) {\(r_{1}\)};
		\node [style=box] (203) at (0, 2.75) {\(r_{2}\)};
		\node [style=box] (204) at (5.75, 1.25) {\(r_{t}\)};
		\node [style=bn] (205) at (6.75, 1.275) {};
		\node [style=none] (206) at (7.5, 2.025) {};
		\node [style=box] (207) at (-7.25, -1.6) {\(f_0\)};
		\node [style=none] (208) at (7.5, -0.725) {};
		\node [style=none] (209) at (-8.5, -0.725) {};
		\node [style=none] (210) at (7.5, -5.025) {};
		\node [style=none] (211) at (-8.5, -5.025) {};
		\node [style=none] (212) at (-7.5, -1.225) {};
		\node [style=none] (213) at (-5.5, -1.225) {};
		\node [style=none] (214) at (-7.5, -1.975) {};
		\node [style=none] (215) at (-5.5, -1.975) {};
		\node [style=none] (216) at (-4.75, -1.975) {};
		\node [style=none] (217) at (-4.75, -1.225) {};
		\node [style=box] (218) at (-4.25, -2.35) {\(f_1\)};
		\node [style=none] (219) at (-2.5, -1.975) {};
		\node [style=none] (220) at (-8.5, -2.725) {};
		\node [style=none] (221) at (-2.5, -2.725) {};
		\node [style=none] (222) at (-1.75, -2.725) {};
		\node [style=none] (223) at (-1.75, -1.975) {};
		\node [style=none] (224) at (-8.5, -1.975) {};
		\node [style=box] (225) at (3.5, -3.85) {\(f_t\)};
		\node [style=none] (226) at (4.25, -3.475) {};
		\node [style=none] (227) at (-8.5, -4.225) {};
		\node [style=none] (228) at (4.25, -4.225) {};
		\node [style=none] (229) at (5, -4.225) {};
		\node [style=none] (230) at (5, -3.475) {};
		\node [style=none] (231) at (7.5, -1.225) {};
		\node [style=none] (232) at (7.5, -1.975) {};
		\node [style=none] (233) at (1.875, -3.85) {\scalebox{.8}{\(t\)}};
		\node [style=none] (234) at (2.525, -3.475) {};
		\node [style=none] (235) at (1.875, -3.55) {\(\cdots\)};
		\node [style=box] (236) at (-1.25, -3.1) {\(f_2\)};
		\node [style=none] (237) at (0.5, -2.725) {};
		\node [style=none] (238) at (-8.5, -3.475) {};
		\node [style=none] (239) at (0.5, -3.475) {};
		\node [style=none] (240) at (1.25, -3.475) {};
		\node [style=none] (241) at (1.25, -2.725) {};
		\node [style=none] (242) at (7.5, -2.725) {};
		\node [style=box] (243) at (-6, -1.25) {\(r_{0}\)};
		\node [style=box] (244) at (-3, -2) {\(r_{1}\)};
		\node [style=box] (245) at (0, -2.75) {\(r_{2}\)};
		\node [style=bn] (247) at (6.75, -4.225) {};
		\node [style=none] (248) at (7.5, -3.475) {};
	\end{pgfonlayer}
	\begin{pgfonlayer}{edgelayer}
		\draw [style=background] (149.center)
			 to (148.center)
			 to (150.center)
			 to (151.center)
			 to cycle;
		\draw (168.center)
			 to (169.center)
			 to [in=180, out=0] (172.center)
			 to (175.center)
			 to [in=180, out=0] (178.center)
			 to (195.center)
			 to [in=180, out=0] (198.center);
		\draw (191.center)
			 to (182.center)
			 to [in=180, out=0] (185.center);
		\draw (187.center)
			 to (173.center)
			 to [in=0, out=180] (171.center)
			 to (170.center)
			 to (180.center);
		\draw (176.center)
			 to (177.center)
			 to [in=180, out=0] (179.center)
			 to (188.center);
		\draw (183.center)
			 to (184.center)
			 to [in=180, out=0] (186.center);
		\draw (196.center)
			 to (197.center)
			 to [in=180, out=0] (199.center)
			 to (200.center);
		\draw (185.center) to (205);
		\draw (206.center) to (186.center);
		\draw [style=background] (209.center)
			 to (208.center)
			 to (210.center)
			 to (211.center)
			 to cycle;
		\draw (212.center)
			 to (213.center)
			 to [in=180, out=0] (216.center)
			 to (219.center)
			 to [in=180, out=0] (222.center)
			 to (237.center)
			 to [in=180, out=0] (240.center);
		\draw (234.center)
			 to (226.center)
			 to [in=180, out=0] (229.center);
		\draw (231.center)
			 to (217.center)
			 to [in=0, out=180] (215.center)
			 to (214.center)
			 to (224.center);
		\draw (220.center)
			 to (221.center)
			 to [in=180, out=0] (223.center)
			 to (232.center);
		\draw (227.center)
			 to (228.center)
			 to [in=180, out=0] (230.center);
		\draw (238.center)
			 to (239.center)
			 to [in=180, out=0] (241.center)
			 to (242.center);
		\draw (229.center) to (247);
		\draw (248.center) to (230.center);
	\end{pgfonlayer}
\end{tikzpicture}
\]
  and similarly for the converse
  \[\begin{tikzpicture}
	\begin{pgfonlayer}{nodelayer}
		\node [style=none] (21) at (-0.25, -0.125) {\(\rotatebox{-90}{\(\subseteq\)}\)};
		\node [style=none] (72) at (0, -6.5) {\(U\bigl(\fbk^{\vb{S}} \bigl((\vb{r}\otimes 1_{\vb Y})\circ \vb{f} \bigr) \bigr)_t\otimes \discard_{X_{t+1}}\,.\)};
		\node [style=none] (73) at (0, -5.625) {\(\rotatebox{-90}{\(=\)}\)};
		\node [style=none] (74) at (0, 6.25) {\( (1_{Y_{[0,t]}} \otimes \discard_{Y_{t+1}}) \circ \bigl( U\bigl(\fbk^{\vb{S}} \bigl(\vb{f} \circ (\shftN(\vb{r})\otimes 1_{\vb X}) \bigr) \bigr) \bigr)_{t+1} \)};
		\node [style=none] (75) at (0, 5.375) {\(\rotatebox{-90}{\(=\)}\)};
		\node [style=box] (122) at (-8, 3.9) {\(f_0\)};
		\node [style=none] (123) at (9.25, 4.775) {};
		\node [style=none] (124) at (-9.25, 4.775) {};
		\node [style=none] (125) at (9.25, 0.475) {};
		\node [style=none] (126) at (-9.25, 0.475) {};
		\node [style=none] (127) at (-8.25, 4.275) {};
		\node [style=none] (128) at (-6.25, 4.275) {};
		\node [style=none] (129) at (-8.25, 3.525) {};
		\node [style=none] (130) at (-6.25, 3.525) {};
		\node [style=none] (131) at (-5.5, 3.525) {};
		\node [style=none] (132) at (-5.5, 4.275) {};
		\node [style=box] (133) at (-3.25, 3.15) {\(f_{t-1}\)};
		\node [style=none] (134) at (-0.5, 3.525) {};
		\node [style=none] (135) at (-9.25, 2.775) {};
		\node [style=none] (136) at (-0.5, 2.775) {};
		\node [style=none] (137) at (0.25, 2.775) {};
		\node [style=none] (138) at (0.25, 3.525) {};
		\node [style=none] (139) at (-9.25, 3.525) {};
		\node [style=box] (140) at (4.25, 1.65) {\(f_{t+1}\)};
		\node [style=none] (141) at (7, 2.025) {};
		\node [style=none] (142) at (-9.25, 1.275) {};
		\node [style=none] (143) at (7, 1.275) {};
		\node [style=none] (144) at (7.75, 1.275) {};
		\node [style=none] (145) at (7.75, 2.025) {};
		\node [style=none] (146) at (9.25, 4.275) {};
		\node [style=none] (147) at (9.25, 3.525) {};
		\node [style=box] (151) at (1, 2.4) {\(f_t\)};
		\node [style=none] (152) at (2.75, 2.775) {};
		\node [style=none] (153) at (-9.25, 2.025) {};
		\node [style=none] (154) at (2.75, 2.025) {};
		\node [style=none] (155) at (3.5, 2.025) {};
		\node [style=none] (156) at (3.5, 2.775) {};
		\node [style=none] (157) at (9.25, 2.775) {};
		\node [style=box] (158) at (-6.75, 4.25) {\(r_{0}\)};
		\node [style=box] (159) at (-1.25, 3.5) {\(r_{t-1}\)};
		\node [style=box] (160) at (2.25, 2.75) {\(r_{t}\)};
		\node [style=box] (161) at (6.25, 2) {\(r_{t+1}\)};
		\node [style=bn] (162) at (8.5, 1.275) {};
		\node [style=bn] (163) at (8.5, 2.025) {};
		\node [style=none] (164) at (2.5, 3.525) {};
		\node [style=none] (165) at (-4.5, 3.525) {};
		\node [style=none] (166) at (-5, 3.15) {\scalebox{.8}{\(t+1\)}};
		\node [style=none] (167) at (-5, 3.45) {\(\cdots\)};
		\node [style=box] (168) at (-6.125, -1.6) {\(f_0\)};
		\node [style=none] (169) at (7.125, -0.725) {};
		\node [style=none] (170) at (-7.375, -0.725) {};
		\node [style=none] (171) at (7.125, -5.025) {};
		\node [style=none] (172) at (-7.375, -5.025) {};
		\node [style=none] (173) at (-6.375, -1.225) {};
		\node [style=none] (174) at (-4.375, -1.225) {};
		\node [style=none] (175) at (-6.375, -1.975) {};
		\node [style=none] (176) at (-4.375, -1.975) {};
		\node [style=none] (177) at (-3.625, -1.975) {};
		\node [style=none] (178) at (-3.625, -1.225) {};
		\node [style=box] (179) at (-1.375, -2.35) {\(f_{t-1}\)};
		\node [style=none] (180) at (1.375, -1.975) {};
		\node [style=none] (181) at (-7.375, -2.725) {};
		\node [style=none] (182) at (1.375, -2.725) {};
		\node [style=none] (183) at (2.125, -2.725) {};
		\node [style=none] (184) at (2.125, -1.975) {};
		\node [style=none] (185) at (-7.375, -1.975) {};
		\node [style=none] (187) at (6.375, -3.475) {};
		\node [style=none] (188) at (-7.375, -4.225) {};
		\node [style=none] (189) at (6.375, -4.225) {};
		\node [style=none] (192) at (7.125, -1.225) {};
		\node [style=none] (193) at (7.125, -1.975) {};
		\node [style=box] (194) at (2.875, -3.1) {\(f_t\)};
		\node [style=none] (195) at (4.625, -2.725) {};
		\node [style=none] (196) at (-7.375, -3.475) {};
		\node [style=none] (197) at (4.625, -3.475) {};
		\node [style=none] (198) at (5.375, -3.475) {};
		\node [style=none] (199) at (5.375, -2.725) {};
		\node [style=none] (200) at (7.125, -2.725) {};
		\node [style=box] (201) at (-4.875, -1.25) {\(r_{0}\)};
		\node [style=box] (202) at (0.625, -2) {\(r_{t-1}\)};
		\node [style=box] (203) at (4.125, -2.75) {\(r_{t}\)};
		\node [style=bn] (205) at (6.375, -4.225) {};
		\node [style=bn] (206) at (6.375, -3.475) {};
		\node [style=none] (207) at (4.375, -1.975) {};
		\node [style=none] (208) at (-2.625, -1.975) {};
		\node [style=none] (209) at (-3.125, -2.35) {\scalebox{.8}{\(t\)}};
		\node [style=none] (210) at (-3.125, -2.05) {\(\cdots\)};
	\end{pgfonlayer}
	\begin{pgfonlayer}{edgelayer}
		\draw [style=background] (124.center)
			 to (123.center)
			 to (125.center)
			 to (126.center)
			 to cycle;
		\draw (127.center) to (128.center);
		\draw [in=180, out=0] (128.center) to (131.center);
		\draw (165.center)
			 to (134.center)
			 to [in=180, out=0] (137.center)
			 to (152.center)
			 to [in=180, out=0] (155.center)
			 to (141.center)
			 to [in=180, out=0] (144.center)
			 to (162);
		\draw (146.center)
			 to (132.center)
			 to [in=0, out=180] (130.center)
			 to (129.center)
			 to (139.center);
		\draw (135.center)
			 to (136.center)
			 to [in=180, out=0] (138.center)
			 to (147.center);
		\draw (142.center)
			 to (143.center)
			 to [in=180, out=0] (145.center);
		\draw (153.center)
			 to (154.center)
			 to [in=180, out=0] (156.center)
			 to (157.center);
		\draw (163) to (145.center);
		\draw [style=background] (170.center)
			 to (169.center)
			 to (171.center)
			 to (172.center)
			 to cycle;
		\draw (173.center) to (174.center);
		\draw [in=180, out=0] (174.center) to (177.center);
		\draw (208.center) to (180.center);
		\draw [in=180, out=0] (180.center) to (183.center);
		\draw (183.center) to (195.center);
		\draw [in=180, out=0] (195.center) to (198.center);
		\draw (198.center) to (187.center);
		\draw (192.center)
			 to (178.center)
			 to [in=0, out=180] (176.center)
			 to (175.center)
			 to (185.center);
		\draw (181.center)
			 to (182.center)
			 to [in=180, out=0] (184.center)
			 to (193.center);
		\draw (188.center) to (189.center);
		\draw (196.center)
			 to (197.center)
			 to [in=180, out=0] (199.center)
			 to (200.center);
	\end{pgfonlayer}
\end{tikzpicture}
\]
  Note that the second inequality requires looking ahead to the context \([0, t+1] \supseteq [0, t]\). Therefore, both sides of the sliding equation are observationally equivalent and thus \(U\) is well-defined on equivalence classes. Moreover, symmetric monoidal functoriality follows immediately from the axiom \(\discard_X \otimes \discard_Y = \discard_{X \otimes Y}\).  
\end{proofE}

In general, \(\Obs_\N(\CC)\) is not known to form a feedback category for monoidal \(\CC\), although the existence of a functor from \(\St([\N, \CC])\) ensures that dinaturally equivalent stateful morphisms will have the same interpretation. Even so, the delay is natural: delaying an observation by one time step does not change the behaviour.
\begin{propositionE}\label{prop:obsfbk:delay-natural-initialised}
  For any discard bicategory \(\CC\), there is a monoidal natural transformation
  \[\delta\colon1_{\Obs_\N(\CC)}\Rightarrow\shftN;\quad \delta_{\vb X}\coloneqq\big[ \big\{ \big(\smbigotimes_{k=0}^{t-1} 1_{X_k}\big)\otimes\discard_{X_t} \big\}_{t\in\N} \big]\,.\]
\end{propositionE}
\begin{proofE}
  Naturality at any \(\vb f\colon\vb X\to\vb Y\) amounts to the equation \(\shftN(\vb f)\circ\delta_{\vb X}=\delta_{\vb Y}\circ\vb f\), which follows by lax naturality of the discard upon looking ahead one context, exactly as in the sliding equation of \Cref{thm:obsfbk:trace-semantics}. Monoidality follows from \(\discard_X\otimes\discard_Y=\discard_{X\otimes Y}\).
\end{proofE}

\begin{example}[Universality of stateful processes]\label{ex:obsfbk:full}
  A natural question to ask is whether all behaviours can be built as stateful morphism sequences. Categorically, this corresponds to showing that the functor \(\St([\N, \CC]) \to \Obs_\N(\CC)\) is full. This holds for all our examples of interest.

  \emph{Classical processes.} For a partial Markov category \(\CC\), the fullness of the functor \(U: \St([\N, \CC]) \to \Obs_\N(\CC)\) can be shown using the property of conditionals, encoding \emph{Bayes rule}. Given a representative monotone sequence, the conditional preorder supplies an effect witnessing each monotonicity inequality \cite[\S~7]{lavore_monoidal_2022}. Note however that this functor does not factor through the category of monotone sequences \(\Mon_{\N}(\CC)\), because it doesn't have a natural delay. Still, the quotient functor \(\Mon_{\N}(\CC) \to \Obs_\N(\CC)\) is full and it was shown by Bonchi et al.\ \cite[Thm.~7.19]{bonchi_effectful_2025} that, if \(\CC\) has \emph{ranges}, monotone sequences can be identified with coalgebraic streams guarded over the total maps, \(\Stream_{\Tot(\CC)}(\CC) \simeq \Mon_{\N}(\CC)\). Composing with the quotient, the functor \(\Stream_{\Tot(\CC)}(\CC) \to \Obs_\N(\CC)\) is full. This means we can interpret coinductive equations on string diagrams \cite{DiLavore2025}, as long as the rule is restricted to total maps. A \emph{complete} coalgebraic calculus of initialised behaviours would have to combine the two: finite sliding of arbitrary morphisms together with coinductive rules restricted to a subclass.

  \emph{Quantum processes.} A similar story should hold in the quantum setting. Here fullness of \(U\) arises instead from the \emph{Stinespring dilation theorem} for completely positive maps, which plays the role that conditionals play classically. This has been axiomatised by Carette et al.~\cite{delayedtracememory} via the notion of a \emph{pseudo-purifiable} discard category equipped with the purification order of \Cref{def:approx:purification}, which can be seen as a strengthening of the notion discard bicategory where every morphism is uniquely expressed by a pure morphism with a minimal discarded environment. Following the same argument as \cite[Prop.~2]{delayedtracememory}, \(U\) is full for all \(\CPTP\), \(\CPTNI\), and \(\CPM\) with the purification order; by \Cref{prop:approx:loewner-purification}, the same holds for \(\CPTNI\) with the L\"owner order.
\end{example}

\subsection{Compact-closed behaviours}\label{sec:obs:compact-closed}

When working in compact-closed categories, the simplest way to construct a feedback category of stateful processes is to freely add a \emph{delay} natural transformation, see \Cref{prop:state:feedback-compact-closed}. This is convenient because, if we have an axiomatisation of a compact-closed category \(\CC\) in terms of generators and relations, we directly obtain an axiomatisation of \(\St(\CC)\) by adding one generator for delay and imposing that it is natural with respect to all other generators.

However, it is not clear what the time evolution of these processes should be, as there cannot be a functor from \(\St(\CC)\) to initialised streams or behaviours. For example, signal flow graphs add a natural delay transformation to the category of affine relations \cite{Bonchi2015, Bonchi2017}, but their time-domain interpretation has so far not been functorial.

As we proceed to show, stateful morphisms over compact-closed discard bicategories can be interpreted functorially as behaviours indexed by the integers rather than the natural numbers.

\begin{definition}\label{def:obsfbk:uninitialised}
  Let \(\J\coloneqq \{ [\ell,r]\ |\ \ell,r \in \Z\colon \ell\leq r \} \) denote  the set of finite closed intervals in \(\Z\).
  Given a discard bicategory \(\CC\), denote the category of \textbf{uninitialised behaviours} over \(\CC\) by
  \(\Obs_\Z(\CC) \coloneqq \Obs_{\Z, \J}(\CC)\).
\end{definition}
\begin{definition}
  Define the delay endofunctor \(\shftZ\colon\Obs_\Z(\CC) \to \Obs_\Z(\CC)\)  on objects by 
  \[\shftZ(\vb{X}) \coloneqq \{X_{k-1}\}_{k \in \Z}, \] 
  and on morphisms by:
  \[ \bigl([\vb{f}]\colon\vb{X}\to \vb{Y}\bigr)\quad  \longmapsto\quad \bigl([ \{ f_{[\ell-1, r-1]} \}_{[\ell,r] \in \J}]\colon\shftZ(\vb{X})\to \shftZ(\vb{Y}) \bigr)\, .\]
\end{definition}

It follows from \Cref{lem:approx:duality} that:

\begin{theoremE}\label{thm:obsfbk:obs-guarded-feedback}
  Given a compact-closed discard bicategory \(\CC\) with \(\codiscard_X \coloneqq (1_X \otimes \discard_{X^*}) \circ \cup_X\),
  there is a monoidal natural isomorphism
  \[\delta\!\colon\!1_{\Obs_\Z(\CC)}\!\Rightarrow \shftZ;\quad \delta_{\vb{X}} \!\coloneqq\! \big[ \big\{ \codiscard_{X_{\ell - 1}} \otimes \big(\smbigotimes_{k=\ell}^{r-1} 1_{X_k}\big) \otimes \discard_{X_r} \big\}_{[\ell,r]\in \J} \big]\,,\]
  making  \(\Obs_\Z(\CC)\) into a compact-closed feedback category.
\end{theoremE}
\begin{proofE}
  Using \Cref{prop:state:feedback-compact-closed}, it suffices to prove that:
  \begin{enumerate}
    \claimitem{obsfbk:obs-guarded:compact}{\(\Obs_\Z(\CC)\) is compact-closed;}
    \claimitem{obsfbk:obs-guarded:mono}{the components of \(\delta\) are monotone nets;}
    \claimitem{obsfbk:obs-guarded:nat}{\(\delta\) is a monoidal natural transformation.}
  \end{enumerate}

  \claimparagraph{obsfbk:obs-guarded:compact}{Compact closure.}
  For every object \(\vb{X}\) of \(\Obs_\Z(\CC)\), we define the dual as \((\vb X)^*\coloneqq(X_k^*)_{k\in\Z}\).
  The cups and caps are defined component-wise by:
  \begin{align*}
      \cup_{\vb X}\coloneqq \left[\left\{(\shuffle_{\vb{X^*},\vb{X}})_{[\ell, r]}\circ\bigotimes_{k=\ell}^r\cup_{X_k}\right\}_{[\ell,r]\in \J}\right]
    \quad\text{and}\quad
      \cap_{\vb X} \coloneqq  \left[\left\{ \bigotimes_{k=\ell}^r\cap_{X_k}\circ(\shuffle_{\vb{X},\vb{X^*}})_{[\ell, r]}^{-1}\right\}_{[\ell,r]\in \J}\right].
  \end{align*}
  Lax naturality ensures that \(\cap_{X_k} \subseteq \discard_{X_k^* \otimes X_k} \) and \( \discard_{X_k \otimes X_k^*} \circ \cup_{X_k} \subseteq 1_I = \discard_I\). Therefore, it follows that \(\cup_{\vb X}\) and \(\cap_{\vb X}\) are monotone nets. The snake equations can be shown to hold component-wise.

  \claimparagraph{obsfbk:obs-guarded:mono}{Monotonicity of delay.}
  We start by showing monotonicity of each component of \(\delta\). Note that for any \([\ell, r] \in \J\),
  \begin{align*}
    (\discard_{X_{\ell - 2}} \otimes 1_{X_{[\ell - 1, r - 1]}}) \circ (\delta_{\vb{X}})_{[\ell - 1, r]} &= 1_{X_{[\ell - 1, r - 1]}} \otimes \discard_{X_r}
     \subseteq \discard_{X_{\ell - 1}} \otimes (\codiscard_{X_{\ell - 1}} \otimes 1_{X_{[\ell - 1, r - 1]}} \otimes \discard_{X_r})\\
    & = \discard_{X_{\ell - 1}} \otimes (\delta_{\vb{X}})_{[\ell, r]}\, ,
  \end{align*}
  where the inequality follows from the inequation \(1_{X_{\ell - 1}} \subseteq \codiscard_{X_{\ell - 1}} \circ \discard_{X_{\ell - 1}}\).
  On the other hand we have:
  \[(1_{X_{[\ell, r]}} \otimes \discard_{X_{r + 1}}) \circ (\delta_{\vb{X}})_{[\ell, r + 1]} = (\delta_{\vb{X}})_{[\ell, r]} \otimes \discard_{X_{r + 1}}\,.\]
  Therefore \(\delta_{\vb{X}}\) is a monotone net.

  \claimparagraph{obsfbk:obs-guarded:nat}{Monoidal naturality of delay.}
  To show naturality, fix any \(\vb f\colon \vb X \to \vb Y\) and consider the compositions
  \(\shftZ(\vb f) \circ \delta_{\vb{X}}\colon \vb X \to \shftZ(\vb Y)\) and
  \(\delta_{\vb{Y}} \circ \vb f\colon \vb X \to \shftZ(\vb Y)\).
  Fix any \([\ell, r]\in \J\).
  
  First, to show \(\shftZ(\vb{f}) \circ \delta_{\vb{X}} \subseteq \delta_{\vb{Y}} \circ \vb{f}\)
  we look ahead to the context \([\ell, r + 1] \supseteq [\ell, r]\):
    \begin{align*}
    &(1_{Y_{[\ell - 1, r - 1]}} \otimes \discard_{Y_{r}}) \circ (\shftZ(\vb f) \circ \delta_{\vb{X}})_{[\ell, r + 1]} 
    \subseteq ((\codiscard_{Y_{\ell - 1}} \circ \discard_{Y_{\ell - 1}}) \otimes 1_{Y_{[\ell, r - 1]}} \otimes \discard_{Y_r}) \circ f_{[\ell-1, r]} \circ (\delta_{\vb{X}})_{[\ell, r + 1]} \\
    &\quad \subseteq (\codiscard_{Y_{\ell - 1}} \otimes 1_{Y_{[\ell, r - 1]}} \otimes \discard_{Y_r}) \circ f_{[\ell, r]} \circ (\discard_{X_{\ell - 1}} \otimes 1_{X_{[\ell, r]}}) \circ (\delta_{\vb{X}})_{[\ell, r + 1]} \\
    &\quad \subseteq  (\codiscard_{Y_{\ell - 1}} \otimes 1_{Y_{[\ell, r - 1]}} \otimes \discard_{Y_r}) \circ f_{[\ell, r]} \circ (1_{X_{[\ell, r]}} \otimes \discard_{X_{r + 1}})
    = (\delta_{\vb{Y}} \circ \vb f)_{[\ell, r]} \otimes \discard_{X_{r + 1}}\, ,
  \end{align*}
  where the first inequality uses  \(1_{Y_{\ell - 1}} \subseteq \codiscard_{Y_{\ell - 1}} \circ \discard_{Y_{\ell - 1}}\), the second is left monotonicity of \(f\), and the third follows by definition  of\(\delta_{\vb{X}}\) and \(\discard_{X_{r + 1}} \circ \codiscard_{X_{r + 1}} \subseteq 1_I\).
  Similarly, for the other direction we look ahead to the context \([\ell - 1, r] \supseteq [\ell, r]\):
  \begin{align*}
  (\discard_{Y_{\ell - 2}} \otimes 1_{Y_{[\ell - 1, r - 1]}})(\delta_{\vb{Y}} \circ \vb f)_{[\ell - 1, r]}
  & \subseteq (1_{Y_{[\ell - 1, r - 1]}} \otimes  \discard_{Y_{r}}) \circ f_{[\ell - 1, r]}
   \subseteq f_{[\ell - 1, r - 1]} \otimes \discard_{X_{r}}\\
  & \subseteq  f_{[\ell - 1, r - 1]} \circ ((\codiscard_{X_{\ell - 1}} \circ \discard_{X_{\ell - 1}}) \otimes 1_{X_{[\ell, r - 1]}} \otimes \discard_{X_{r}})\\
  & = \discard_{X_{\ell - 1}} \otimes (f_{[\ell - 1, r - 1]} \circ (\delta_{\vb{X}})_{[\ell, r]})
  = \discard_{X_{\ell - 1}} \otimes [\shftZ(\vb f) \circ \delta_{\vb{X}}]_{[\ell, r]}\, .
  \end{align*}

  Therefore we have shown that the two sides of naturality are observationally equal.

  Finally, coherence of the delay natural transformation with respect to the unit and the monoidal product follows from the coherence axioms of \(\discard\). Since \(\Obs_\Z(\CC)\) is compact closed, \Cref{prop:state:feedback-compact-closed} ensures that \(\delta\) is moreover a natural isomorphism.
\end{proofE}
Using the universal property of the state construction and lax naturality, we obtain a functor witnessing that uninitialised behaviours are a feedback-preserving categorical semantics for Mealy machines in compact-closed categories:
\begin{corollaryE}\label{cor:obsfbk:always-unroll}
  Given a compact-closed discard bicategory  \(\CC\),
  there are feedback functors 
  \(\St(\CC) \to \St([\Z, \CC], \shftZ) \to \Obs_\Z(\CC)\).
\end{corollaryE}
\begin{proofE}
  The first functor is given by \Cref{lem:state:z-stateful}.

  By \Cref{thm:obsfbk:obs-guarded-feedback}, \(\Obs_\Z(\CC)\) is a compact-closed feedback category, with delay endofunctor
  \[\shftZ\colon\Obs_\Z(\CC)\to \Obs_\Z(\CC).\]

  Define a symmetric monoidal functor \(J\colon[\Z, \CC] \to \Obs_\Z(\CC)\) acting as the identity on objects and on morphisms by:
  \[
    J(f)\coloneqq
    \Big[\big\{\textstyle\bigotimes_{k=\ell}^{r} f_k \big\}_{[\ell,r]\in\J}\Big].
  \]
  Fix a sequence \(f_k\colon X_k\to Y_k\) and a context \([\ell,r]\in\J\).
  Using monoidality of \(\discard\) and \(\codiscard\), the right monotonicity inequality for \(J(f)\) reduces to
  \(
    \discard_Y\circ f_k \subseteq \discard_X,
  \) 
  and the left monotonicity inequality reduces to
  \(
    f_k\circ \codiscard_X \subseteq \codiscard_Y,
  \) 
  both of which hold by lax naturality.

  Moreover, it is immediate that \(J\) preserves the \(\shftZ\) endofunctor, i.e. \(\shftZ\circ J = J \circ \shftZ\).

  Invoking \Cref{prop:state:state-free} with \(\Gamma=1_{[\Z, \CC]}\) and \(\Delta=\shftZ\),
  there is a unique feedback functor
  \(
  F\colon\St([\Z, \CC])\to \Obs_\Z(\CC)
  \)
  such that \(F \circ \eta = J\) as desired.

  Finally, \(F\) is a feedback functor because both source and target are compact-closed feedback categories with trivial guard, therefore by \Cref{prop:state:feedback-compact-closed}, the feedback operator and the delay determine one another.  Therefore, the preservation of  feedback forces that the delay and induced compact-closed structure are preserved.
\end{proofE}

\begin{example}\label{ex:obsfbk:examples}
  Among our examples, these results apply to (i) non-deterministic processes in any Cartesian bicategory \(\rel(\CC)\) and (ii) unnormalised quantum processes in \(\CPM\) with the purification order. Indeed both are compact-closed discard bicategories. The resulting behaviours are very well-behaved in both cases, having the structure of a feedback compact closed discard bicategory. \(\Obs_\Z(\CPM)\) can be interpreted as a semantics for infinite dimensional quantum systems such as spin chains and quantum channels with memory.

  The \emph{structure theorem} of Kretschman and Werner \cite{kretschmann_quantum_2005} corresponds to the fullness of the functor \(U: \St(\CPM) \to \Obs_\Z(\CPM)\). Note that this functor is not full if \(\CPM\) contains only \emph{finite dimensional} systems. This is because, differently from the \(\N\)-indexed setting, behaviours could in principle retain memory from infinitely far into the past, and these cannot be expressed as time-evolutions of machines with finite-dimensional memory.

  The full structure theorem could possibly be recovered by allowing infinite-dimenional quantum systems in the process theory, as in \cite{selinger2004towards,cho_semantics_2014}. The advantage of a \emph{finite memory} assumption in this case, is that the resulting equational theory is more tractable and we gain access to truly infinite behaviour while retaining compact closure and a natural delay.

  We leave the analysis of stateful quantum processes to future work and, in the remainder of this paper, we investigate the compact-closed relational setting in more detail.
\end{example}

\section{Compactness theorem}
\label{sec:compactness}
\pratendSetLocal{category=compactness, end, restate}
\emph{When do local approximations completely and uniquely determine a global process?} Just as series need not converge, neither do behaviours.
However, when the category is sufficiently well-behaved, we may be able to glue the local approximations into a coherent infinite process.

In this section, we show that this is the case for  the category of closed relations between compact Hausdorff spaces, exhibiting a functor  \(\Obs(\rel(\Haus)) \to \rel(\Haus)\).

In first-order logic, the \emph{compactness theorem} states that a set of first-order sentences has a model if and only if every finite subset of it has a model (see G\"odel~\cite[Satz~X]{godel} and later Tarski~\cite[Thm.~13]{tarski}).
The construction of this functor can be interpreted as a categorified compactness theorem for {relational} models.

\begin{definition}\label{def:compactness:haus}
  Let \(\Haus\) denote the category of \textbf{compact Hausdorff topological spaces} and \textbf{continuous functions}.
\end{definition}

\begin{lemma}[{\cite[prop~15.]{khausrel}}]\label{lem:compactness:haus-regular}
  \(\Haus\) is a regular category, such that the morphisms  \(X \to Y\) in \(\rel(\Haus)\) are exactly the \emph{closed subspaces} of \(X \times Y\) with respect to the product topology.
\end{lemma}
Hausdorffness is needed for the identity to be closed; whereas compactness is needed for composition to preserve closedness.

We exhibit a functor \(\Obs(\rel(\Haus)) \to \rel(\Haus)\) which glues together observational equivalence classes of monotone nets on \(\{X_j\}_{j\in \JJ}\) into a closed subspace of \(\prod_{j \in \JJ} X_j \), subject to a covering condition on the set of contexts. We moreover show that this restricts to a \emph{faithful} functor \(\Obs(\rel(\FSet)) \to \rel(\Haus)\).

For an index set \(\I\) and a set of finite contexts \(j,j'\in\JJ\) such that \(j\subseteq j'\), we use the following notations for projections:
\[\pi_{j}^{\infty}\colon \prod_{k\in\I} X_k \to \prod_{k\in j} X_k,\quad\text{and}\quad \pi_{j}^{j'}\colon \prod_{k\in j'} X_k \to \prod_{k\in j} X_k.\]

Using the inverse image of these projections and intersections, we glue together monotone nets using the following construction:
\begin{lemmaE}\label{lem:compactness:lim}
  Given an \(I\)-indexed set of compact Hausdorff spaces \(\vb{X}=\{X_k \in \Ob(\Haus)\}_{k\in\I}\), and a monotone net of closed subspaces \(\vb{C} \coloneqq \{C_j \subseteq \prod_{k \in j} X_k\}_{j\in \JJ}\), then there is a compact, closed subspace
  \[ \lim(\vb{C}) \coloneqq \bigcap_{j \in \JJ} (\pi_{j}^{\infty})^{-1} (C_j)\subseteq \prod_{k\in\I} X_k.\]
\end{lemmaE}
\begin{proofE}
  The product \(\prod_{k\in\I} X_k\) is compact with the product topology by Tychonoff's theorem and any intersection of closed subsets of a compact space is closed and compact.
\end{proofE}

The construction of \Cref{lem:compactness:lim} lifts to a functor when the set of contexts possesses the following property:
\begin{definition}\label{def:compactness:cover}
  \(\JJ\subseteq \Pfin(\I)\) is a \textbf{cover}, in case for any \( i \in \I\) there exists a subset \(j \in \JJ\) containing \(i \in j\).
  \(\JJ\) is \textbf{cofinal}, in case for each finite subset \(s \in \mathcal{P}_{\mathit{f}} (\I) \) there is some \(j \supseteq s\) in \(\JJ\).
  Note that  \(\JJ\) is cofinal if and only if it is an upward-directed cover.
\end{definition}
\begin{theoremE}[Compactness]\label{thm:compactness:compactness}
  Given a set  \(\I\) with cofinal \(\JJ\subseteq \Pfin(\I)\), then \(\lim\) induces a poset-enriched discard functor 
  \(\Lim\colon \Obs_{\I, \JJ}(\rel(\Haus)) \rightarrow \rel(\Haus)\).
\end{theoremE}

\begin{proofE}
  Define \(\Lim\) on objects by
  \[\vb{X}=\{X_k \in \Ob(\Haus)\}_{k\in\I} \mapsto \prod_{k\in\I} X_k,\]
  and on morphisms \([\vb{S}]\colon\vb{X}\to \vb{Y}\) by
  \[
    \Lim([\vb S])=\shuffle_{\vb{X},\vb{Y}}\big(\lim(\vb S)\big)\subseteq
    \Big(\prod_{i\in\Z}X_i\Big)\times\Big(\prod_{i\in\Z}Y_i\Big),
  \]
  recalling the canonical shuffling isomorphism
  \[\shuffle_{\vb{X},\vb{Y}}\colon \prod_{k\in\I} X_k \otimes Y_k \cong \prod_{k\in\I} X_k \times \prod_{k\in\I} Y_k.\]

  We show that \(\Lim\) is a discard functor, i.e.\ that it preserves:
  \begin{enumerate}
    \claimitem{compactness:lim:unit}{the monoidal unit;}
    \claimitem{compactness:lim:discard}{discards;}
    \claimitem{compactness:lim:id}{identities;}
    \claimitem{compactness:lim:poset}{the poset structure (and hence is well-defined on equivalence classes);}
    \claimitem{compactness:lim:tensor}{monoidal products;}
    \claimitem{compactness:lim:lax}{composition laxly (\(\Lim(\vb T)\circ\Lim(\vb S)\subseteq\Lim(\vb T\circ\vb S)\));}
    \claimitem{compactness:lim:oplax}{composition oplaxly (\(\Lim(\vb T\circ\vb S)\subseteq\Lim(\vb T)\circ\Lim(\vb S)\)).}
  \end{enumerate}

  \claimparagraph{compactness:lim:unit}{Preservation of monoidal unit.}
  \(\Lim\) preserves the monoidal unit since \(\prod_{k\in\I} \{\ast\} \cong \{\ast\}\).

  \claimparagraph{compactness:lim:discard}{Preservation of discards.}
\(\Lim\) preserves discards, since
   \[\lim(\discard_{\vb{X}}) = \bigcap_{j \in \JJ}(\pi_{j}^{\infty})^{-1}(X_j) = \bigcap_{j \in \JJ} \prod_{k\in\I} X_k = \prod_{k\in\I} X_k.\]

  \claimparagraph{compactness:lim:id}{Preservation of identities.} We want to show \(1_{\prod_{k\in\I} X_k} = \Lim(1_{\vb{X}})\).

  \paragraph{\(1_{\prod X}\subseteq\Lim(1_{\vb{X}})\).}
  First, for any \(j\) we have \(\pi_{j}^{\infty} 1_{\prod_{k\in\I} X_k} \subseteq 1_{X_j}\) and therefore
  \[1_{\prod_{k\in\I} X_k} \subseteq \bigcap_{j \in \JJ} (\pi_{j}^{\infty})^{-1} (1_{X_j}) = \Lim(1_{\vb{X}}).\]
  \paragraph{\(\Lim(1_{\vb{X}})\subseteq 1_{\prod X}\).}
  For the other direction, suppose \((x, y) \in \Lim(1_{\vb{X}})\). Therefore \(\pi_{j}^{\infty}(x, y) \in 1_{X_j}\), so that for all \(j\), \(\pi_{j}^{\infty} x = \pi_{j}^{\infty} y\).
  Using the fact that  \(\JJ\) is a \emph{cover}, for any \(k \in \I\) there exists some \(j \ni k\) and projecting out the \(k^{\text{th}}\) component of \(X_j\) we get \(x_k = y_k\).
  Therefore \(x = y\) and we have shown \(\Lim(1_{\vb{X}}) \subseteq 1_{\prod_{i\in\I} X_i}\).

  \claimparagraph{compactness:lim:poset}{Preservation of poset and well-definedness on equivalence classes.}
  Suppose \(\vb{A} \subseteq \vb{B}\), then there is a function \(f\) that to any context \(j\) associates a larger context \(j' = f(j) \supseteq j\) such that \(\pi^{j'}_{j} A_{j'} \subseteq B_j\). Then since \(\pi_j^\infty = \pi^{j'}_{j} \circ \pi_{j'}^\infty\) we have \((\pi_j^\infty)^{-1} = (\pi_{j'}^{\infty})^{-1} \circ (\pi_{j}^{j'})^{-1}\) and so \((\pi_{j'}^\infty)^{-1}(A_{j'}) \subseteq (\pi_{j}^\infty)^{-1}(B_j)\). Therefore
  \begin{align*}
    \lim(\vb{A}) &= \bigcap_{j \in \JJ} (\pi_{j}^{\infty})^{-1} (A_j)
    \subseteq  \bigcap_{j \in \JJ} (\pi_{f(j)}^{\infty})^{-1} (A_{f(j)})
    \subseteq \bigcap_{j \in \JJ} (\pi_{j}^{\infty})^{-1} (B_{j}) = \lim(\vb{B}),
  \end{align*}
  so that \(\Lim\) preserves the poset structure on monotone nets, which implies that \(\Lim\) is well-defined on observational equivalence classes.

  \claimparagraph{compactness:lim:tensor}{Preservation of monoidal product.}
  Let \(\vb{S}\) and \(\vb{T}\) be two observational nets. Without loss of generality assume that \(S_j \subseteq A_j\) and \(T_j \subseteq B_j\) for some nets \(\vb{A}, \vb{B} \in \Ob(\CC)^\I\).
  Then \(\pi_{j'}^\infty(a, b) = (\pi_{j'}^{\infty, A} a, \pi_{j'}^{\infty, B} b)\) for the projections \(\pi_{j'}^{\infty, X}\colon \prod_{k\in\I} X_k \to \prod_{k\in j} X_k\). Therefore:
  \begin{align*}
  \lim(\vb{S} \otimes \vb{T}) &= \bigcap_{j \in \JJ} \big( (\pi_{j'}^{\infty, A})^{-1}(S_j) \times (\pi_{j'}^{\infty, B})^{-1}(T_j) \big)
  = \bigcap_{j \in \JJ} \big( (\pi_{j'}^{\infty, A})^{-1}(S_j) \big) \times \bigcap_{j \in \JJ} \big( (\pi_{j'}^{\infty, B})^{-1}(T_j) \big) \\
  &\cong \lim(\vb{S}) \times \lim(\vb{T}),
  \end{align*}
  from which it follows that \(\Lim\) preserves monoidal products.

  \claimparagraph{compactness:lim:lax}{Preservation of composition: laxness.}
  Take composable nets \(\vb{S}\) and \(\vb{T}\). We start by proving that  \(\Lim\) laxly preserves composition, meaning that
  \(\Lim(\vb{T}) \circ \Lim(\vb{S})\subseteq \Lim(\vb{T} \circ \vb{S}).\)

  Suppose that \((x, z) \in \Lim(\vb{T}) \circ \Lim(\vb{S})\) then there exists some \(y\) such that \((x, y) \in \Lim(\vb{S})\) and \((y, z) \in \Lim(\vb{T})\).
  Note that for any \(j_0 \in \JJ\), \(\pi_{j_0}^{\infty}(\vb{C}) \subseteq C_{j_0}\) because
  \begin{align*}
  \pi_{j_0}^{\infty} (\lim(\vb C))= \pi_{j_0}^{\infty} (\bigcap_{j\in\J}(\pi_{j}^{\infty})^{-1}(C_j))
  \subseteq& \pi_{j_0}^\infty ((\pi_{j_0}^{\infty})^{-1}(C_{j_0})) = C_{j_0}\, .
  \end{align*}
  Therefore,  for any \(j\), \(\pi_{j}^{\infty}(x, y) = (\pi_{j}^{\infty} x, \pi_{j}^{\infty} y) \in S_j\)
  and \((\pi_{j}^{\infty} y, \pi_{j}^{\infty} z) \in T_j\) which implies that \((\pi_{j}^{\infty} x, \pi_{j}^{\infty} z) \in T_j \circ S_j = (T \circ S)_j\).
  Therefore, \((x, z) \in \Lim(\vb{T} \circ \vb{S})\).

  \claimparagraph{compactness:lim:oplax}{Preservation of composition: oplaxness.}
  Now we prove the harder direction, that  \(\Lim\) oplaxly preserves composition, meaning that 
  \(\Lim(\vb{T} \circ \vb{S}) \subseteq \Lim(\vb{T}) \circ \Lim(\vb{S}).\)

  Suppose \((x, z) \in \Lim(\vb{T} \circ \vb{S})\) then for all \(j\), \((\pi_{j}^{\infty}x, \pi_{j}^{\infty} z) \in (T \circ S)_j = T_j \circ S_j\) there exists some \(y_j\) such that \((\pi_{j}^{\infty}x, y) \in S_j\) and \((y_j, \pi_{j}^{\infty} z) \in T_j\).
  Therefore the following sets are non-empty:
  \[
  W_j=\left\{y_j\in \prod_{i \in \I} Y_i \ \middle|\
  (\pi^\infty_j x,y_j)\in S_j,\ (y_j,\pi^\infty_j z) \in T_j\right\}.
  \]
  They are moreover closed, since \(W_j\) is the intersection of the image of \(\{\pi^\infty_j x\}\) under \(S_j\) and the preimage of \(\{\pi^\infty_j z\}\) under \(T_j\) and both are closed relations.

  Now take any \(j' \supseteq j\). It follows that
  \begin{align*}
    \pi^{j'}_{j}(W_{j'})
    &=\left\{y_j\in \prod_{i \in j} Y_i\ \middle| \
    \exists y_{j'}\in \prod_{i \in j'} Y_i\colon
    \begin{matrix}
      \pi^{j'}_{j}(y_{j'})=y_j,\\
      (\pi^\infty_{j'} x,y_{j'})\in S_{j'},\\
      (y_{j'},\pi^\infty_{j'} z)\in T_{j'}
      \end{matrix}
  \right\}
  \\
  &\subseteq
  \left\{y_j\in \prod_{i\in j} Y_i\ \middle|\
  \begin{matrix}
    (\pi^\infty_{j}x, y_j)\in S_{j},\\\
    (y_j,\pi^\infty_{j}z)\in T_{j}
  \end{matrix}\right\}
= W_j.
  \end{align*}
  Therefore \((\pi^{j'}_{j})^{-1}(W_{j}) \supseteq W_{j'}\).

  Now for each \(j \in \JJ\), define:
  \[
  C_j \coloneqq \bigl(\pi^{\infty}_{j}\bigr)^{-1}(W_j)\ \subseteq\ \prod_{i\in\Z} Y_i .
  \]
  Whenever \(j \subseteq j'\), \(\pi_j^\infty = \pi^{j'}_{j} \circ \pi_{j'}^\infty\), we have that
  \((\pi_j^\infty)^{-1} = (\pi_{j'}^{\infty})^{-1} \circ (\pi_{j}^{j'})^{-1}\,,\) therefore
  \begin{align*}
  C_j = (\pi_j^\infty)^{-1}(W_j) = (\pi_{j'}^{\infty})^{-1} \circ (\pi_{j}^{j'})^{-1} (W_{j})
  \supseteq (\pi_{j'}^{\infty})^{-1}(W_{j'}) = C_{j'}\, .
  \end{align*}

  Since each \(W_j\) is closed and the projections \(\pi^{\infty}_{j}\) are continuous, \(C_j\) is also closed.
  Moreover, since each \(W_j\) is non-empty so is \(C_j\).

  We show that, since \(\JJ\) is \emph{directed}, the \(C_j\) satisfy the finite intersection property.
  Fix any \(n\) and take arbitrary contexts \(\{j_1,  \dots, j_n\}\). By directedness there exists some \(j^* \in \JJ\) such that \(j^* \supseteq j_i\) for all \(i\).
  Therefore \(C_{j^*} \subseteq \bigcap_{i = 1}^n C_{j_i}\) and since \(C_{j^*}\) is non-empty, the finite intersection property holds.
  The intersection of closed subsets of a compact space satisfying the finite intersection property is nonempty, so that:
  \[ \bigcap_{j \in \JJ} (\pi^{\infty}_{j})^{-1}(W_j) = \bigcap_{j \in \JJ} C_j \neq \emptyset. \]

  Finally, take \(y \in \bigcap_{j \in \JJ} C_j\), then for every \(j\), \(\pi_j^\infty(x, y) \in S_j\) and \(\pi_j^\infty(y, z) \in T_j\).
  Therefore \((x, y) \in \Lim(\vb{S})\) and \((y, z) \in \Lim(\vb{T})\) and it follows that \(\Lim(\vb{T} \circ \vb{S}) \subseteq \Lim(\vb{T}) \circ \Lim(\vb{S})\).
\end{proofE}

Note that the assumption that \(\JJ\) is a cover is only used for preservation of the identity, while the condition of directedness is only used for preservation of composition.

Moreover, when each finite context admits only finitely many possible observations, the global behaviour \emph{uniquely determines} which local observations are possible:
\begin{textAtEnd}
  \begin{definition}\label{def:compactness:artinian}
    A poset \((P,\leq)\) is \textbf{Artinian} in case for every sequence \((x_k)_{k\in\N}\) in \(P\) with
    \(x_{k+1}\leq x_k\) for all \(k\in\N\), there exists \(d\in\N\) such that \(x_k=x_d\) for all \(k\geq d\).
  \end{definition}

  Importantly, the example that we care about, \(\rel(\FSet)\), is enriched in Artinian posets since every strictly descending chain of subsets of a finite set terminates.

  \begin{lemma}\label{lem:compactness:stabilisation}
    Let \(\CC\) be an Artinian poset-enriched compact-closed discard category satisfying
    \(1_X \subseteq \codiscard_X\circ \discard_X\) and \(\discard_X\circ \codiscard_X \subseteq 1_I\).
    Let \(\JJ\subseteq \Pfin(\I)\) be upward-directed.
    Then every monotone net \(\vb{S}\) of morphisms in \(\CC\) is observationally equivalent to a unique
    causal net \(\fix(\vb{S})\).
  \end{lemma}
  \begin{proof}
    Fix \(j\in\JJ\). For each \(j'\in\JJ\) with \(j\subseteq j'\), set
    \[
    B_{j}^{j'}
    \coloneqq
    (1_{Y_j}\otimes \discard_{Y_{j'\setminus j}})\circ S_{j'}\circ
    (1_{X_j}\otimes \codiscard_{X_{j'\setminus j}})
    \;\colon\; X_j\to Y_j\, .
    \]
    By monotonicity of \(\vb S\), enlarging \(j'\) can only refine this morphism, so the family
    \(\{B_{j}^{j'}\mid j'\in\JJ,\ j\subseteq j'\}\) stabilises at some \(\rho(\vb S,j)\in\JJ\) with \(j\subseteq \rho(\vb S,j)\).
    Define
    \(
    \fix(\vb S)\coloneqq \{\,B_{j}^{\rho(\vb S,j)}\,\}_{j\in\JJ}.
    \)

    We need to show that \(\fix(\vb S)\) is:
    \begin{enumerate}
      \claimitem{compactness:stab:causal}{causal;}
      \claimitem{compactness:stab:obseq}{observationally equivalent to \(\vb S\);}
      \claimitem{compactness:stab:unique}{the unique such causal net.}
    \end{enumerate}

    \claimparagraph{compactness:stab:causal}{Causality.}
    Let \(j\subseteq j'\).
    Then
    \(
    (1_{Y_j}\otimes \discard_{Y_{j'\setminus j}})\circ \fix(\vb S)_{j'}
    =
    (1_{Y_j}\otimes \discard_{Y_{j'\setminus j}})\circ B_{j'}^{\rho(\vb S,j')}
    =
    B_{j}^{\rho(\vb S,j')}.
    \)
    Since \(\rho(\vb S,j')\) contains \(j'\), it also contains \(j\), and by stabilisation of the family for \(j\) we have
    \(B_{j}^{\rho(\vb S,j')}=B_{j}^{\rho(\vb S,j)}=\fix(\vb S)_j\).
    Thus \(\fix(\vb S)\) is causal.

    \claimparagraph{compactness:stab:obseq}{Observational equivalence.}
    For each \(j\), monotonicity at \(j\subseteq \rho(\vb S,j)\) yields
    \[
    (1_{Y_j}\otimes \discard_{Y_{\rho(\vb S,j)\setminus j}})\circ S_{\rho(\vb S,j)}
    \subseteq
    S_j\otimes \discard_{X_{\rho(\vb S,j)\setminus j}}\, .
    \]
    Precomposing \(1_{X_j}\otimes \codiscard_{X_{\rho(\vb S,j)\setminus j}}\) and using
    \(\discard\circ\codiscard\subseteq 1_I\) yields \(\fix(\vb S)_j\subseteq S_j\), and thus \(\fix(\vb S)\subseteq \vb S\).

    For the converse, using \(1\subseteq \codiscard\circ\discard\) on the extra input factors gives
    \[
    (1_{Y_j}\otimes \discard_{Y_{\rho(\vb S,j)\setminus j}})\circ S_{\rho(\vb S,j)}
    \subseteq
    \fix(\vb S)_j\otimes \discard_{X_{\rho(\vb S,j)\setminus j}}\, ,
    \]
    so that \(\vb S\subseteq \fix(\vb S)\).
    Therefore \(\vb S\) and \(\fix(\vb S)\) are observationally equivalent.

    \claimparagraph{compactness:stab:unique}{Uniqueness.}
    Let \(\vb T\) be a causal net observationally equivalent to \(\vb S\).
    Then \(\vb T\) is observationally equivalent to \(\fix(\vb S)\).
    Since \(\vb T\) and \(\fix(\vb S)\) are causal, approximation implies contextwise comparison, hence
    \(\vb T=\fix(\vb S)\) contextwise.
  \end{proof}

  \begin{lemma}\label{lem:compactness:stabilized-projection}
    Given a set \(\I\) and cofinal \(\JJ\subseteq \Pfin(\I)\),
    and any causal monotone net \(\vb{S}=\{S_j\}_{j\in\JJ}\) in \(\rel(\FSet)\), then for all \(j\in\JJ\), we have that
    \(
    \pi_{j}^{\infty}\bigl(\lim(\vb{S})\bigr)= S_j\, .
    \)
  \end{lemma}
  \begin{proof}
    The inclusion \(\pi_{j}^{\infty}(\lim(\vb S))\subseteq S_j\) is immediate since
    \(\lim(\vb S)\subseteq(\pi_{j}^{\infty})^{-1}(S_j)\).

    For the reverse inclusion, fix \(j_0\in\JJ\).
    If \(S_{j_0}=\emptyset\) there is nothing to show.
    Otherwise, choose \(x_{j_0}\in S_{j_0}\).
    For each \(k\in\JJ\) with \(k\supseteq j_0\), set
    \[
    C_k\coloneqq(\pi_{k}^{\infty})^{-1}(S_k)\ \cap\ (\pi_{j_0}^{\infty})^{-1}(\{x_{j_0}\})
    \subseteq \prod_{i\in\I} X_i.
    \]
    Since \(\vb S\) is causal, for every \(k\supseteq j_0\) we have \(\pi_{j_0}^{k}(S_k)=S_{j_0}\), so each \(C_k\) is nonempty.

    Fix \(n\in\N\) and a family \((k_m)_{1\le m\le n}\) in \(\JJ\) with \(k_m\supseteq j_0\) for all \(m\).
    Choose \(k^*\in\JJ\) with \(k^*\supseteq \bigcup_{1\le m\le n} k_m\).
    Then \(C_{k^*}\subseteq \bigcap_{1\le m\le n} C_{k_m}\), so \(\{C_k\}_{k\supseteq j_0}\) has the finite intersection property.
    Since \(\prod_{i\in\I}X_i\) is compact Hausdorff, \(\bigcap_{k\supseteq j_0}C_k\neq\emptyset\).
    Choose \(y\) in this intersection.
    Then \(y\in\lim(\vb S)\) and \(\pi_{j_0}^{\infty}(y)=x_{j_0}\),
    so \(x_{j_0}\in \pi_{j_0}^{\infty}(\lim(\vb S))\).
  \end{proof}
  From these two lemmas it follows immediately that:
\end{textAtEnd}
\begin{corollary}[Faithfulness] \label{cor:compactness:faithfulness}
  Given a set  \(\I\) with cofinal  \(\JJ\subseteq \Pfin(\I)\), \(\Lim\) restricts to a faithful poset-enriched discard functor 
  \(\Lim\colon \Obs_{\I, \JJ}(\rel(\FSet)) \rightarrow \rel(\Haus)\).
\end{corollary}
This follows because all descending chains of subsets of finite sets terminate (see   \Cref{lem:compactness:stabilisation} and \Cref{lem:compactness:stabilized-projection} for a detailed argument). However, because every infinite compact Hausdorff space has an infinite strictly descending chain of subsets, \(\FSet\) is the largest \emph{full} subcategory of \(\Haus\) which admits this faithfulness property.

\section{Linear time-invariant systems}
\label{sec:signalflow}
\pratendSetLocal{category=signalflow, end, restate}

Linear time-invariant (LTI) systems are central objects of study in signal processing \cite{oppenheim1989}, control theory \cite{Willems}, and electrical circuits \cite[\S~3.1.3]{lti}. They have received increased attention in the context of categorical semantics, thanks to the work of Bonchi, Zanasi et al. \cite{Bonchi2015,Bonchi2017,Bonchi2021,Zanasi2015InteractingHopfAlgebras} as well as Baez et al. \cite{baez, BaezFong, networks}.

While their categorical semantics is well-understood from the perspective of the frequency domain (which we will call the \(z\)-domain), particularly with the use of affine relations over fields of polynomials \cite{Bonchi2021, baez}, their time-domain interpretation has long eluded categorical treatment. Different approaches exist to capture the ``executions'' of LTI systems, from the behaviours of Willems \cite{Willems} and Oberst et al. \cite[\S~3.1.1]{lti} to the trajectories of Bonchi et al. \cite{Bonchi2021}, but none is compositional.

In this section, we show that \(\Obs_{\Z}\) and \(\Obs_{\N}\) give a categorical semantics for uninitialised and initialised LTI systems in the time-domain. We relate the \(z\)-domain interpretation of \cite{Bonchi2021, baez} to our time-domain semantics via an oplax feedback functor capturing the inverse \(z\)-transform. We moreover show that our notion of uninitialised behaviour, instantiated in this setting, precisely captures the behaviours of Willems \cite{Willems}. We end by discussing the syntactic implications for the theory of signal flow graphs.

\subsection{The z-domain}\label{sec:signalflow:z-domain}
Fix a field \(\K\). Let \(\FAff_{\K}\) denote the SMC of affine transformations between finite-dimensional vector spaces, and let  \(\FAffRel{\K}\) denote the full subcategory of \(\rel(\FAff_{\K})\) whose objects are finite-dimensional vector spaces.
The category of finite-dimensional affine relations,  \(\FAffRel{\K}\),  is a categorical semantics for time-independent open \(\K\)-linear systems \cite{baez, gaa}. The linear relations are interpreted as global constraints;  whereas  elements of \(\K^n\), regarded as affine displacements, are interpreted as  signals propagating through the system according to these constraints.

Let \(\fps\) denote the ring of formal power series, and \(\fls\) its field of fractions, the formal Laurent series (see \Cref{app:polynomials} for a review of the theory of polynomial-type rings).
Following Oberst et al.~\cite[Eq.~(8.2)\ \& Rem.~8.1.1]{lti}, we identify initialised discrete-time signals in the \emph{\(z\)-domain}\footnote{Oberst et al. refer to this as the algebraic domain \cite[P.~3]{lti}.}  with formal power series \(a(z)=\sum_{t=0}^\infty a_t z^t\), so that the signal at time \(t\) is given by the coefficient \(a_t\), where multiplication by \(z\) corresponds to a discrete time-delay.

By extending scalars along \(\fls \supset \fps\), the delay becomes invertible, so that by forgetting initial conditions,
LTI systems in the z-domain may be conveniently represented in terms of  \(\fls\)-linear subspaces of the vector space
\(\fls^{\In}\times\fls^{\Out}\) \cite[Cor.\ \& Def.~8.1.4]{lti}.

Just as linear time-independent systems can be composed in  \(\FAffRel {\K}\), following Baez and Erbele \cite{baez} as well as Bonchi et al.\ \cite[\S~6.2]{Bonchi2017}, uninitialised LTI systems in the \(z\)-domain can be composed in \(\FLinRel{\fls}\). In order to also account for the signals, we moreover need to include affine displacements, which altogether yields the following feedback category:
\begin{lemma}\label{lem:signalflow:z-feedback}
  There is faithful symmetric monoidal functor  \(\FLinRel{\fls} \rightarrowtail \FAffRel{\fls}\)
  and a monoidal natural transformation  \(\partial\colon 1_{\FLinRel{\fls}} \Longrightarrow 1_{\FLinRel{\fls}}\); \(\partial_V \coloneqq \{(\vb v, z\cdot \vb v) \mid \vb v\in V\}\),
  making \(\FAffRel{\fls}\)   into a feedback category.
\end{lemma}
Note that in order to compose LTI systems in the \(z\)-domain using relational composition we are forced to use formal Laurent series, rather than formal power series, since  \( \FAffRel{\fps}\simeq\FAffRel{\fls}\) \cite[Thm.~4.14]{Zanasi2015InteractingHopfAlgebras}. Therefore, this relational categorical semantics for LTI systems in the \(z\)-domain cannot accommodate for initial conditions.

\subsection{The time-domain}\label{sec:signalflow:time-domain}
In the \emph{time-domain}, initialised discrete-time signals are represented by \(\N\)-indexed sequences in \(\K\) \cite[\S~5.2.2]{lti}.
Accordingly, the  \emph{behaviours} or \emph{trajectories} of initialised LTI systems can be represented  by systems of linear constant coefficient \emph{difference equations}, the discrete analogue of differential equations \cite[\S~3.1.1]{lti}.  The solutions to these difference equations are certain \(\poly\)-linear subspaces of
\( (\K^\N)^{\In}\times(\K^\N)^{\Out}\).

The \textbf{inverse \(\mathbf z\)-transform} is the \(\poly\)-linear transformation \(\ZZ^{\circ}_n \colon\fps^n\to (\K^\N)^n\) which extracts coefficients, transforming signals in the \(z\)-domain into the time-domain \cite[Rem.~8.1.1]{lti}.

In order to remain compatible with the categorical semantics for the \(z\)-domain, we abandon initial conditions, and define the delay \(\Delta\colon\K^\Z\to \K^\Z\)  by \((\Delta x)_{k+1}\coloneqq x_{k}\). In this setting we can now define inverse \(z\)-transform as the \(\poly[\K][z,1/z]\)-linear transformation  \(\ZZ^{\circ}_n \colon\fls^n\to (\K^\Z)^n\) which extracts coefficients.

Note that the inverse \(z\)-transform preserves the time-delay so that 
\(\ZZ^{\circ}_n(z\cdot f(z))=\Delta(\ZZ^{\circ}_n(f(z)))\).

For each context \({[\ell,r]} \in \J\) postcomposing  this \(\ZZ^{\circ}_n\) with the projection \mbox{\((\K^\Z)^n\to \prod_{k=\ell}^r \K^n\)}, yields the \(\K\)-linear map
\(\zeta_{[\ell,r]}^n \colon \fls^n \to  \prod_{k=\ell}^r  \K^n\).
This allows us to categorify the inverse \(z\)-transform, making  \(\Obs_\Z(\FAffRel \K)\) a categorical semantics for uninitialised linear time-dependent  systems in the time-domain.
\begin{theoremE}\label{thm:signalflow:z}
  There is an oplax normal feedback functor
  \(\ZZ^*\colon\FAffRel{\fls}\to \Obs_\Z(\FAffRel \K)\)
  given on objects by \(n\mapsto n^\Z\) and on morphisms
  \(S\colon n\to m\) by
  \[
    \ZZ^*(S)\;\coloneqq\;
    \Big[\Big\{
      \bigl\{\bigl(\zeta^n_{j}(\vb x),\zeta^m_{j}(\vb y)\bigr)\ \big|\ (\vb x,\vb y)\in S\bigr\}
      \Big\}_{j\in\J}\Big];
  \]
  where oplax means that
    \(\ZZ^*(T\circ S)\subseteq \ZZ^*(T)\circ \ZZ^*(S)\)
    and normal means that
    \(\ZZ^*(1_n) = 1_{n^\Z}\).
\end{theoremE}
\begin{proofE}
  We need to show that \(\ZZ^*\):
  \begin{enumerate}
    \claimitem{signalflow:z:mono}{produces a monotone net of affine relations;}
    \claimitem{signalflow:z:oplax}{is oplax;}
    \claimitem{signalflow:z:fbk}{is normal and preserves the relevant structure.}
  \end{enumerate}

  \claimparagraph{signalflow:z:mono}{Well-definedness.}
  We show that \(\ZZ^*\) produces a monotone net of affine relations.
  First note that for all \(j\in \J\),  \((\ZZ^*(S))_{j}\) is an \(\K\)-affine subspace; moreover for all  \(n,m\in \N\), \(\zeta^n_{j}\) and \(\zeta^m_{j}\) are \(\K\)-linear maps.
  Furthermore, each \((\ZZ^*(S))_{j}\) is defined via the projection of the same infinite affine subspace, so \(\ZZ^*(S)\) produces a causal, and thus monotone, net of affine relations.

  \claimparagraph{signalflow:z:oplax}{Oplaxness.}
  Now we prove that \(\ZZ^*\) is an oplax functor.
  Fix \(j\in\J\).
  If \(T\circ S=\varnothing\) then \((\ZZ^*(T\circ S))_{j}=\varnothing\) by definition, so the inclusion is vacuous.

  Otherwise, take some \((u,w)\in(\ZZ^*(T\circ S))_{j}\). By definition of \((\ZZ^*(T\circ S))_{j}\), there exists some \((\vb x,\vb z)\in T\circ S\) such that
  \(
    u=\zeta^n_{j}(\vb x)
  \) and \(
    w=\zeta^p_{j}(\vb z).
  \)
  Since \((\vb x,\vb z)\in T\circ S\), by definition of relational composition in \(\FAffRel{\fls}\) there exists some \(\vb y\in\fls^m\) such that
  \(
    (\vb x,\vb y)\in S
  \) and \(
    (\vb y,\vb z)\in T.
  \)
  Since \((\vb x,\vb y)\in S\), by definition of \((\ZZ^*(S))_{j}\) we have
  \[
    \bigl(u,\zeta^m_{j}(\vb y)\bigr)
    =
    \bigl(\zeta^n_{j}(\vb x),\zeta^m_{j}(\vb y)\bigr)
    \in (\ZZ^*(S))_{j}\, .
  \]
  Similarly
  \[
    \bigl(\zeta^m_{j}(\vb y),w\bigr)
    =
    \bigl(\zeta^m_{j}(\vb y),\zeta^p_{j}(\vb z)\bigr)
    \in (\ZZ^*(T))_{j}\, .
  \]
  Thus, by definition of relational composition in \(\FAffRel{\K}\),
  \[
    (u,w)\in (\ZZ^*(T))_{j}\circ (\ZZ^*(S))_{j}\, .
  \]
  Therefore, \(\ZZ^*\) is oplax.

  \claimparagraph{signalflow:z:fbk}{Symmetric monoidal and feedback structure.}
  \(\ZZ^*\) preserves the identity, as well as all of the other symmetric monoidal and feedback structure, since they are defined contextwise.
\end{proofE}

By replacing \(\Z\) with \(\N\) we can therefore interpret  \(\Obs_\N(\FAffRel \K)\) as a categorical semantics for \emph{initialised} linear time-dependent systems in the time-domain.

\paragraph{Compactness.}
Given a  field \(\K\), let \(\Aff_{\K}\) denote the SMC of affine transformations between (possibly infinite) vector spaces and the empty set; moreover, let \(\AffRel {\K}\) denote the full subcategory of \(\rel(\Aff_{\K})\) whose objects are vector spaces.
Using \Cref{thm:compactness:compactness}, and observing that only the \emph{finite} fields are compact, it follows that:
\begin{propositionE}\label{prop:signalflow:lin-compact}
  Given a finite field \(\F_q\), and a cofinal set  \(\JJ\subseteq \Pfin(\I)\) there is a faithful poset-enriched discard functor
  \(\Lim\colon\Obs_{\I, \JJ}(\FAffRel {\F_q}) \to  \AffRel {\F_q}\).
\end{propositionE}
\begin{proofE}
  By \Cref{thm:compactness:compactness} there is a poset-enriched discard functor \(\Lim\colon \Obs_{\I, \JJ}(\rel(\Haus)) \longrightarrow \rel(\Haus)\).
  Since \(\F_q\) is a finite field, any vector space over \(\F_q\) is compact and Hausdorff with respect to the product topology of discrete topologies.
  Moreover since the projections are linear maps, inverse images of linear maps are linear maps, and intersections of linear subspaces are linear subspaces. Therefore, it follows that there is an induced poset-enriched discard functor \( \Obs_{\I,\JJ}(\AffRel {\F_q}) \to \AffRel {\F_q}\).
  Moreover, because the objects of \(\FAffRel {\F_q}\) are compact, by \Cref{cor:compactness:faithfulness}, it follows that this functor is moreover faithful.
\end{proofE}

Therefore, observing that \(\J \subseteq \Pfin(\Z)\) and \(\{[0,t] \ | \ t\in \N \} \subseteq \Pfin(\N)\) are both cofinal it follows that for finite fields \(\F_q\), there are faithful poset-enriched discard functors
\[\Obs_\N(\FAffRel {\F_q}) \to  \AffRel {\F_q}\qquad\text{and}\qquad \Obs_\Z(\FAffRel {\F_q}) \to  \AffRel {\F_q},\]
witnessing a ``global'' behavioural semantics both for initialised and uninitialised linear time-dependent systems.
Moreover, the induced oplax functor \(\Lim\circ \ZZ^*\) is given on morphisms by the \(\F_q\)-linear extension of
\(\ZZ^{\circ}_n\colon\fls[\F_q]^n\to (\F_q^\Z)^n\).

Following Willems, the behaviour of an LTI system is a \(\K\)-subspace of trajectories cut out by a finite system of constant-coefficient difference equations~\cite{Willems}.  Note that the work of Willems was one of the original motivations of Zanasi for developing \(\FAffRel{\rat[\F_q]}\) as a categorical semantics for LTI systems \cite[P.~12]{Zanasi2015InteractingHopfAlgebras}.

The name ``behaviour'' is no coincidence. Over a finite field, which is the setting where the compactness theorem holds,  the induced oplax functor
\(\FAffRel{\rat[\F_q]}\to \AffRel{\F_q}\)
can be seen to categorify this notion in our setting.

We take the formal definition of behaviours from the book of Oberst et al.~\cite[\S~3.1.1]{lti}.
\begin{definition}[Affine behaviour]\label{def:signalflow:affine-behaviour}
  Let \(S\subseteq \rat^\ell\) be a finite-dimensional \(\rat\)-linear subspace.
  Choose \(r\in\N\) and a \(\rat\)-linear map \(A\colon\rat^\ell\to \rat^r\) with \(S=\ker A\).
  Pick some nonzero \(d(z)\in\poly[\K][z,1/z]\) clearing denominators and set
  \(B(z)\coloneqq d(z) A\in \poly[\K][z,1/z]^{ r\times \ell}\).
  Evaluate entrywise to obtain the \(\poly[\K][z,1/z]\)-linear map
  \(
    B^{\rho^\ell}\colon ((\K^\Z)^\ell,\rho^\ell)\longrightarrow ((\K^\Z)^r,\rho^r)
  \)
  given by
  \(
    B^{\rho^\ell}(u)\coloneqq B(\Delta) u
  \).
  Moreover, define \(\beh(S) \coloneqq \ker (B^{\rho^\ell}) \subseteq (\K^\Z)^\ell.\)

  Let \(\gamma\colon\rat\to \fls\) denote the geometric series expansion and define \(\ZZ^{\gamma}_\ell(-) \coloneqq \ZZ^{\circ}_\ell(\gamma(-))\).
  We extend \(\beh\) to \(\rat\)-affine subspaces \(S\subseteq \rat^\ell\) by:
  for \(S=\varnothing\) define \(\beh(S)\coloneqq \varnothing\); otherwise choosing some \(\vb a\in S\), with \(R\coloneqq S-\vb a\):
  \(\beh(S)\coloneqq \beh(R)+\ZZ^{\gamma}_\ell(\vb a)\subseteq (\K^\Z)^\ell\).
\end{definition}
Note that Oberst et al.~\cite[\S~3.1.1]{lti} define this only over linear subspaces, but the extension to affine relations as above is immediately obvious.
Indeed the oplax functor \(\Lim\circ \ZZ^*\) can be seen to categorify this notion of behaviour, albeit only after restricting to the trajectories which are rational.
\begin{theoremE}\label{thm:signalflow:lti-compactness}
  Let \(\F_q\) be a finite field and \(S\subseteq \rat[\F_q]^\ell\) be a finite-dimensional \(\rat[\F_q]\)-affine subspace. Then
  \(
    (\Lim\circ \ZZ^*)(S)
    =
    \beh(S)\cap \ZZ^{\gamma}_\ell(\rat[\F_q]^\ell)
  \).
\end{theoremE}
\begin{proofE}
  We prove the claim for affine subspaces \(S\subseteq \rat[\F_q]^\ell\).  The analogous statement for affine relations
  \(R\subseteq \rat[\F_q]^n\times \rat[\F_q]^m\) follows immediately by viewing \(R\) as an affine subspace of
  \(\rat[\F_q]^{n+m}\) and transporting along the canonical shuffling isomorphism
  \[
    \shuffle_{n,m}\colon  (\F_q^n)^\Z \times (\F_q^m)^\Z \cong  (\F_q^{n+m})^\Z.
  \]

  First, take \(S=\varnothing\), then \((\Lim\circ \ZZ^*)(S)=\varnothing=\beh(S)\), so the claim follows immediately.
  
  Assume otherwise, taking some \(\vb a\in S\). Define \(R\coloneqq S-\vb a\). Therefore,  \(R\subseteq \rat[\F_q]^\ell\) is an \(\rat[\F_q]\)-linear subspace and \(S=R+\vb a\).
  By definition of \(\beh\) on affine subspaces,
  \[
    \beh(S)=\beh(R)+\ZZ^{\gamma}_\ell(\vb a).
  \]
  Since \(\Lim\circ \ZZ^*\) is given on morphisms by \(\ZZ^{\gamma}\), and \(\ZZ^{\gamma}\) is \(\poly[\F_q][z,1/z]\)-linear, we also have
  \[
    (\Lim\circ \ZZ^*)(S)=(\Lim\circ \ZZ^*)(R)+\ZZ^{\gamma}_\ell(\vb a),
  \]
  and \(\ZZ^{\gamma}_\ell(\vb a)\in \ZZ^{\gamma}_\ell(\rat[\F_q]^\ell)\).
  Thus it suffices to prove the following for the linear subspace \(R\):
  \begin{enumerate}
    \claimitem{signalflow:lti:zlinear}{for all \(u\in \fls[\F_q]^\ell\), \(\ZZ^{\gamma}_r(B(z)u)=B(\Delta)\ZZ^{\gamma}_\ell(u)\);}
    \claimitem{signalflow:lti:fwd}{\((\Lim\circ \ZZ^*)(R)\subseteq \beh(R)\);}
    \claimitem{signalflow:lti:bwd}{\(\beh(R)\cap \ZZ^{\gamma}_\ell(\rat[\F_q]^\ell)\subseteq (\Lim\circ \ZZ^*)(R)\).}
  \end{enumerate}

  Choose  some \(r\in\N\) and an \(\rat[\F_q]\)-linear map \(A\colon\rat[\F_q]^\ell\to \rat[\F_q]^r\) with \(R=\ker A\).
  Pick \(d(z)\in \poly[\F_q][z,1/z]\setminus\{0\}\) clearing denominators and set
  \(B(z)\coloneqq d(z)A\in \poly[\F_q][z,1/z]^{r\times \ell}\).
  Then \(\beh(R)=\ker(B(\Delta))\subseteq (\F_q^\Z)^\ell\).

  \claimparagraph{signalflow:lti:zlinear}{Linearity.}
  Since \(\ZZ^{\gamma}\) is \(\poly[\F_q][z,1/z]\)-linear, for each entry \(b(z)\in \poly[\F_q][z,1/z]\) and \(x\in \fls[\F_q]\) we have
  \(
  \ZZ^{\gamma}_1(b(z)x)=b(\Delta)\ZZ^{\gamma}_1(x),
  \)
  and applying this entrywise to \(B(z)\) and componentwise to vectors yields \(\ZZ^{\gamma}_r(B(z)u)=B(\Delta)\ZZ^{\gamma}_\ell(u)\).

  \claimparagraph{signalflow:lti:fwd}{Forward inclusion.}
  Take \(u\in (\Lim\circ \ZZ^*)(R)\). Then there  exists some \(x\in R\) such that \(u=\ZZ^{\gamma}_\ell(x)\).
  Since \(x\in R=\ker(B(z))\), we have that \(B(z)x=0\); thus,
  \[
  0=\ZZ^{\gamma}_r(B(z)x)=B(\Delta)\ZZ^{\gamma}_\ell(x)=B(\Delta)u,
  \]
  therefore \(u\in \ker(B(\Delta))=\beh(R)\).

  \claimparagraph{signalflow:lti:bwd}{Backward inclusion.}
  Take \(u\in \beh(R)\cap \ZZ^{\gamma}_\ell(\rat[\F_q]^\ell)\).
  Choose \(x\in \rat[\F_q]^\ell\) with \(u=\ZZ^{\gamma}_\ell(x)\).
  Since \(u\in \beh(R)=\ker(B(\Delta))\), we have that \(B(\Delta)u=0\), hence
  \(
  0=B(\Delta)\ZZ^{\gamma}_\ell(x)=\ZZ^{\gamma}_r(B(z)x).
  \)
  By injectivity of \(\ZZ^{\gamma}_r\), \(B(z)x=0\), so \(x\in \ker(B(z))=R\).
  Thus \(u=\ZZ^{\gamma}_\ell(x)\in (\Lim\circ \ZZ^*)(R)\).

  Therefore \((\Lim\circ \ZZ^*)(R)=\beh(R)\cap \ZZ^{\gamma}_\ell(\rat[\F_q]^\ell)\).
\end{proofE}

\subsection{Signal flow graphs}\label{sec:signalflow:sfg}
Linear \emph{signal flow graphs} are a syntax for linear time-invariant systems; see for example, the seminal work Shannon~\cite{shannon}, and later of Mason~\cite{Mason1953, Mason1955}.
In this section, we discuss how the \(z\)-domain and time-domain semantics for uninitialised LTI systems produce different equational theories for signal flow graphs; followed briefly by a discussion concerning the syntactic  implications of our initialised time-domain semantics.

First we recall the following equational theory for linear \emph{time-independent} systems:
\begin{theorem}[{Bonchi et al.~\cite{gaa}}]\label{thm:signalflow:gaa}
  Given a field \(\K\), the SMC \(\FAffRel{\K}\) is presented by the following generators, for all \(a \in \K\):
  \begin{align*}
    \Big\llbracket\ \begin{tikzpicture}
	\begin{pgfonlayer}{nodelayer}
		\node [style=none] (19) at (-0.25, 0.5) {};
		\node [style=none] (20) at (1.25, 0.5) {};
		\node [style=none] (21) at (1.25, -0.5) {};
		\node [style=none] (22) at (-0.25, -0.5) {};
		\node [style=bn] (23) at (0.5, 0) {};
		\node [style=none] (24) at (1.25, 0.375) {};
		\node [style=none] (25) at (1.25, -0.375) {};
		\node [style=none] (26) at (-0.25, 0.375) {};
		\node [style=none] (27) at (-0.25, -0.375) {};
		\node [style=none] (28) at (-0.025, 0) {\(\vdotss\)};
		\node [style=none] (37) at (1, 0) {\(\vdotss\)};
	\end{pgfonlayer}
	\begin{pgfonlayer}{edgelayer}
		\draw [style=background] (21.center)
			 to (20.center)
			 to (19.center)
			 to (22.center)
			 to cycle;
		\draw [in=60, out=180] (24.center) to (23);
		\draw [in=180, out=-60] (23) to (25.center);
		\draw [in=360, out=120] (23) to (26.center);
		\draw [in=-120, out=0] (27.center) to (23);
	\end{pgfonlayer}
\end{tikzpicture}\ \Big\rrbracket
    &\coloneqq \{ (\vb v, \vb w) \ | \ \forall j, k\colon v_j = w_k\} ,\\
    \Big\llbracket\ \begin{tikzpicture}
	\begin{pgfonlayer}{nodelayer}
		\node [style=none] (19) at (-0.25, 0.5) {};
		\node [style=none] (20) at (1.25, 0.5) {};
		\node [style=none] (21) at (1.25, -0.5) {};
		\node [style=none] (22) at (-0.25, -0.5) {};
		\node [style=wn] (23) at (0.5, 0) {};
		\node [style=none] (24) at (1.25, 0.375) {};
		\node [style=none] (25) at (1.25, -0.375) {};
		\node [style=none] (26) at (-0.25, 0.375) {};
		\node [style=none] (27) at (-0.25, -0.375) {};
		\node [style=none] (28) at (-0.025, 0) {\(\vdotss\)};
		\node [style=none] (37) at (1, 0) {\(\vdotss\)};
		\node [style=wphase] (38) at (0.5, 0.75) {\(a\)};
	\end{pgfonlayer}
	\begin{pgfonlayer}{edgelayer}
		\draw [style=background] (21.center)
			 to (20.center)
			 to (19.center)
			 to (22.center)
			 to cycle;
		\draw [in=60, out=180] (24.center) to (23);
		\draw [in=180, out=-60] (23) to (25.center);
		\draw [in=360, out=120] (23) to (26.center);
		\draw [in=-120, out=0] (27.center) to (23);
		\draw [style=dwphase] (38) to (23);
	\end{pgfonlayer}
\end{tikzpicture}\ \Big\rrbracket
    &\coloneqq \{(\vb v, \vb w) \ |\ a+ \Sigma\ v_j = \Sigma\ w_j  \},\\
    \Big\llbracket \ \begin{tikzpicture}
	\begin{pgfonlayer}{nodelayer}
		\node [style=none] (0) at (-0.75, 0.375) {};
		\node [style=none] (1) at (0.75, 0.375) {};
		\node [style=none] (2) at (0.75, -0.375) {};
		\node [style=none] (3) at (-0.75, -0.375) {};
		\node [style=rmat] (4) at (0, 0) {\(a\)};
		\node [style=none] (5) at (-0.75, 0) {};
		\node [style=none] (6) at (0.75, 0) {};
	\end{pgfonlayer}
	\begin{pgfonlayer}{edgelayer}
		\draw [style=background] (2.center)
			 to (1.center)
			 to (0.center)
			 to (3.center)
			 to cycle;
		\draw (5.center) to (6.center);
	\end{pgfonlayer}
\end{tikzpicture} \ \Big\rrbracket
    &\coloneqq \{ (b, ab) \ | \ \forall b \in \K \}.
  \end{align*}
  modulo the equations given in \Cref{app:gaa}.
\end{theorem}
These generators impose linear constraints on systems; where the label \(a\) on the white node is interpreted as a signal propagating through this linear system.

The question of finding an equational theory for LTI systems, and axiomatizing signal flow graphs, can therefore be reduced to the question of how to add a notion of   \emph{time} to the equational theory of \(\FAffRel{\K}\).

The simplest way to answer this question is to freely add a discrete-time delay generator which commutes only with the  constraints, which by \Cref{prop:state:feedback-compact-closed} is captured by the feedback category \(\St_{\FLinRel \K}(\FAffRel \K)\).
Concretely, this is presented by adding the generator \(\begin{tikzpicture}
	\begin{pgfonlayer}{nodelayer}
		\node [style=none] (34) at (5.25, 0.375) {};
		\node [style=none] (35) at (6.75, 0.375) {};
		\node [style=none] (36) at (6.75, -0.375) {};
		\node [style=none] (37) at (5.25, -0.375) {};
		\node [style=none] (187) at (5.25, 0) {};
		\node [style=rmat] (197) at (6, 0) {\(\partial\)};
		\node [style=none] (198) at (6.75, 0) {};
	\end{pgfonlayer}
	\begin{pgfonlayer}{edgelayer}
		\draw [style=background] (36.center)
			 to (35.center)
			 to (34.center)
			 to (37.center)
			 to cycle;
		\draw (187.center) to (198.center);
	\end{pgfonlayer}
\end{tikzpicture}\) to the equational theory of \(\FAffRel \K\) modulo the equations:
\[\begin{tikzpicture}
	\begin{pgfonlayer}{nodelayer}
		\node [style=none] (75) at (8.25, -1) {};
		\node [style=none] (76) at (10.5, -1) {};
		\node [style=none] (77) at (10.5, 1.25) {};
		\node [style=none] (78) at (8.25, 1.25) {};
		\node [style=bn] (79) at (9.75, 0) {};
		\node [style=none] (82) at (8.25, 0.5) {};
		\node [style=none] (83) at (8.25, -0.5) {};
		\node [style=none] (84) at (10.5, 0.5) {};
		\node [style=none] (85) at (10.5, -0.5) {};
		\node [style=none] (86) at (9, 0.5) {};
		\node [style=none] (87) at (9, -0.5) {};
		\node [style=none] (88) at (10.25, 0) {\(\vdotss\)};
		\node [style=none] (89) at (8.25, 0.5) {};
		\node [style=none] (90) at (8.5, 0) {\(\vdotss\)};
		\node [style=rmat] (92) at (9, 0.5) {\(\partial\)};
		\node [style=rmat] (93) at (9, -0.5) {\(\partial\)};
		\node [style=none] (94) at (11.25, -1) {};
		\node [style=none] (95) at (13.5, -1) {};
		\node [style=none] (96) at (13.5, 1.25) {};
		\node [style=none] (97) at (11.25, 1.25) {};
		\node [style=bn] (98) at (12, 0) {};
		\node [style=none] (99) at (13.5, 0.5) {};
		\node [style=none] (100) at (13.5, -0.5) {};
		\node [style=none] (103) at (12.75, 0.5) {};
		\node [style=none] (104) at (12.75, -0.5) {};
		\node [style=none] (105) at (11.25, 0.5) {};
		\node [style=none] (106) at (11.25, -0.5) {};
		\node [style=none] (107) at (13.25, 0) {\(\vdotss\)};
		\node [style=none] (109) at (11.5, 0) {\(\vdotss\)};
		\node [style=rmat] (111) at (12.75, 0.5) {\(\partial\)};
		\node [style=rmat] (112) at (12.75, -0.5) {\(\partial\)};
		\node [style=none] (113) at (10.875, 0) {\(=\)};
		\node [style=none] (114) at (14.25, -1) {};
		\node [style=none] (115) at (16.5, -1) {};
		\node [style=none] (116) at (16.5, 1.25) {};
		\node [style=none] (117) at (14.25, 1.25) {};
		\node [style=wn] (118) at (15.75, 0) {};
		\node [style=none] (121) at (14.25, 0.5) {};
		\node [style=none] (122) at (14.25, -0.5) {};
		\node [style=none] (123) at (16.5, 0.5) {};
		\node [style=none] (124) at (16.5, -0.5) {};
		\node [style=none] (125) at (15, 0.5) {};
		\node [style=none] (126) at (15, -0.5) {};
		\node [style=none] (127) at (16.25, 0) {\(\vdotss\)};
		\node [style=none] (128) at (14.25, 0.5) {};
		\node [style=none] (129) at (14.5, 0) {\(\vdotss\)};
		\node [style=rmat] (131) at (15, 0.5) {\(\partial\)};
		\node [style=rmat] (132) at (15, -0.5) {\(\partial\)};
		\node [style=none] (133) at (17.25, -1) {};
		\node [style=none] (134) at (19.5, -1) {};
		\node [style=none] (135) at (19.5, 1.25) {};
		\node [style=none] (136) at (17.25, 1.25) {};
		\node [style=wn] (137) at (18, 0) {};
		\node [style=none] (138) at (19.5, 0.5) {};
		\node [style=none] (139) at (19.5, -0.5) {};
		\node [style=none] (142) at (18.75, 0.5) {};
		\node [style=none] (143) at (18.75, -0.5) {};
		\node [style=none] (144) at (17.25, 0.5) {};
		\node [style=none] (145) at (17.25, -0.5) {};
		\node [style=none] (146) at (19.25, 0) {\(\vdotss\)};
		\node [style=none] (148) at (17.5, 0) {\(\vdotss\)};
		\node [style=rmat] (150) at (18.75, 0.5) {\(\partial\)};
		\node [style=rmat] (151) at (18.75, -0.5) {\(\partial\)};
		\node [style=none] (152) at (16.875, 0) {\(=\)};
		\node [style=none] (153) at (13.75, 0) {,};
		\node [style=none] (154) at (19.75, 0) {,};
		\node [style=none] (155) at (20.25, -1) {};
		\node [style=none] (156) at (22, -1) {};
		\node [style=none] (157) at (22, 1.25) {};
		\node [style=none] (158) at (20.25, 1.25) {};
		\node [style=none] (173) at (22.75, -1) {};
		\node [style=none] (174) at (24.5, -1) {};
		\node [style=none] (175) at (24.5, 1.25) {};
		\node [style=none] (176) at (22.75, 1.25) {};
		\node [style=none] (177) at (20.25, 0) {};
		\node [style=none] (178) at (22, 0) {};
		\node [style=none] (179) at (22.75, 0) {};
		\node [style=none] (180) at (24.5, 0) {};
		\node [style=rmat] (181) at (20.75, 0) {\(\partial\)};
		\node [style=rmat] (182) at (21.5, 0) {\(a\)};
		\node [style=rmat] (183) at (23.25, 0) {\(a\)};
		\node [style=rmat] (184) at (24, 0) {\(\partial\)};
		\node [style=none] (185) at (22.375, 0) {\(=\)};
	\end{pgfonlayer}
	\begin{pgfonlayer}{edgelayer}
		\draw [style=background] (77.center)
			 to (76.center)
			 to (75.center)
			 to (78.center)
			 to cycle;
		\draw (79)
			 to [in=360, out=135] (86.center)
			 to (82.center);
		\draw (79)
			 to [in=0, out=-135] (87.center)
			 to (83.center);
		\draw [in=180, out=45] (79) to (84.center);
		\draw [in=180, out=-45] (79) to (85.center);
		\draw [style=background] (96.center)
			 to (95.center)
			 to (94.center)
			 to (97.center)
			 to cycle;
		\draw [in=360, out=135] (98) to (105.center);
		\draw [in=0, out=-135] (98) to (106.center);
		\draw (98)
			 to [in=180, out=45] (103.center)
			 to [in=180, out=0] (99.center);
		\draw (98)
			 to [in=180, out=-45] (104.center)
			 to (100.center);
		\draw [style=background] (116.center)
			 to (115.center)
			 to (114.center)
			 to (117.center)
			 to cycle;
		\draw (118)
			 to [in=360, out=135] (125.center)
			 to (121.center);
		\draw (118)
			 to [in=0, out=-135] (126.center)
			 to (122.center);
		\draw [in=180, out=45] (118) to (123.center);
		\draw [in=180, out=-45] (118) to (124.center);
		\draw [style=background] (135.center)
			 to (134.center)
			 to (133.center)
			 to (136.center)
			 to cycle;
		\draw [in=360, out=135] (137) to (144.center);
		\draw [in=0, out=-135] (137) to (145.center);
		\draw (137)
			 to [in=180, out=45] (142.center)
			 to [in=180, out=0] (138.center);
		\draw (137)
			 to [in=180, out=-45] (143.center)
			 to (139.center);
		\draw [style=background] (157.center)
			 to (156.center)
			 to (155.center)
			 to (158.center)
			 to cycle;
		\draw [style=background] (175.center)
			 to (174.center)
			 to (173.center)
			 to (176.center)
			 to cycle;
		\draw (177.center) to (178.center);
		\draw (179.center) to (180.center);
	\end{pgfonlayer}
\end{tikzpicture}\,.\]
Where the white spider with no label denotes a white spider with label \(0\).

Let \(\rat\) denote the field of fractions of the polynomial ring \(\poly\) (i.e.~the rational functions over \(\K\)).
Just as \(\FAffRel{\fls}\) is a feedback category, so is \(\FAffRel{\rat}\). Moreover,
\begin{lemmaE}\label{lem:signalflow:ext}
  There is a full feedback functor
  \(\St_{\FLinRel{\K}}(\FAffRel \K) \twoheadrightarrow \FAffRel{\rat}\),
  given by identifying the delays and extending scalars along  \(\rat/\K\), and extending scalars along the geometric series embedding \(\fls/\rat\) induces a feedback functor \(\Ext\colon\St_{\FLinRel{\K}}(\FAffRel \K) \to \FAffRel{\fls}\).
\end{lemmaE}
\begin{proofE}
  By \Cref{prop:state:feedback-compact-closed},  \(\St_{\FLinRel{\K}}(\FAffRel{\K})\) is presented by freely adding a ``discrete-time delay'' automorphism to \(\FAffRel{\K}\) which is natural on the subcategory \(\FLinRel{\K}\).
\end{proofE}

Identifying the respective delay maps,  so that \(\begin{tikzpicture}
	\begin{pgfonlayer}{nodelayer}
		\node [style=none] (34) at (5.25, 0.375) {};
		\node [style=none] (35) at (6.75, 0.375) {};
		\node [style=none] (36) at (6.75, -0.375) {};
		\node [style=none] (37) at (5.25, -0.375) {};
		\node [style=none] (187) at (5.25, 0) {};
		\node [style=rmat] (197) at (6, 0) {\(\partial^n\)};
		\node [style=none] (198) at (6.75, 0) {};
	\end{pgfonlayer}
	\begin{pgfonlayer}{edgelayer}
		\draw [style=background] (36.center)
			 to (35.center)
			 to (34.center)
			 to (37.center)
			 to cycle;
		\draw (187.center) to (198.center);
	\end{pgfonlayer}
\end{tikzpicture}=\begin{tikzpicture}
	\begin{pgfonlayer}{nodelayer}
		\node [style=none] (34) at (5.25, 0.375) {};
		\node [style=none] (35) at (6.75, 0.375) {};
		\node [style=none] (36) at (6.75, -0.375) {};
		\node [style=none] (37) at (5.25, -0.375) {};
		\node [style=none] (187) at (5.25, 0) {};
		\node [style=rmat] (197) at (6, 0) {\(z^n\)};
		\node [style=none] (198) at (6.75, 0) {};
	\end{pgfonlayer}
	\begin{pgfonlayer}{edgelayer}
		\draw [style=background] (36.center)
			 to (35.center)
			 to (34.center)
			 to (37.center)
			 to cycle;
		\draw (187.center) to (198.center);
	\end{pgfonlayer}
\end{tikzpicture}\), the equational theory of  \(\FAffRel{\rat}\) is given by quotienting the equational theory of \(\St_{\FLinRel{\K}}(\FAffRel \K)\) by the following equations for all nonzero \(\vb{a} \in \poly\):
\begin{equation}\label{eq:signalflow:right-inverse}
  \begin{tikzpicture}
	\begin{pgfonlayer}{nodelayer}
		\node [style=none] (34) at (5.25, 1.125) {};
		\node [style=none] (37) at (5.25, -1) {};
		\node [style=none] (187) at (5.25, -0.125) {};
		\node [style=none] (188) at (5.875, -0.125) {};
		\node [style=none] (189) at (9.125, -0.125) {};
		\node [style=rmat] (191) at (6.875, 0.75) {\(a_0\)};
		\node [style=rmat] (192) at (6.875, -0.625) {\(a_n\)};
		\node [style=none] (193) at (6.625, 0.25) {};
		\node [style=none] (194) at (6.625, -0.625) {};
		\node [style=none] (195) at (8.375, 0.25) {};
		\node [style=none] (196) at (8.375, -0.625) {};
		\node [style=rmat] (197) at (8.125, -0.625) {\(\partial^n\)};
		\node [style=bn] (198) at (5.875, -0.125) {};
		\node [style=wn] (199) at (9.125, -0.125) {};
		\node [style=none] (200) at (5.875, -0.125) {};
		\node [style=none] (201) at (9.125, -0.125) {};
		\node [style=none] (202) at (6.625, 0.75) {};
		\node [style=none] (203) at (8.375, 0.75) {};
		\node [style=rmat] (204) at (6.875, 0.25) {\(a_1\)};
		\node [style=rmat] (205) at (8.125, 0.25) {\(\partial\)};
		\node [style=none] (206) at (7.5, -0.125) {\(\vdotss\)};
		\node [style=none] (207) at (13.75, 1.125) {};
		\node [style=none] (210) at (13.75, -1) {};
		\node [style=none] (211) at (13.75, -0.125) {};
		\node [style=none] (212) at (13.125, -0.125) {};
		\node [style=none] (213) at (9.875, -0.125) {};
		\node [style=lmat] (215) at (12.125, 0.75) {\(a_0\)};
		\node [style=lmat] (216) at (12.125, -0.625) {\(a_n\)};
		\node [style=none] (217) at (12.375, 0.25) {};
		\node [style=none] (218) at (12.375, -0.625) {};
		\node [style=none] (219) at (10.625, 0.25) {};
		\node [style=none] (220) at (10.625, -0.625) {};
		\node [style=lmat] (221) at (10.875, -0.625) {\(\partial^n\)};
		\node [style=bn] (222) at (13.125, -0.125) {};
		\node [style=wn] (223) at (9.875, -0.125) {};
		\node [style=none] (224) at (13.125, -0.125) {};
		\node [style=none] (225) at (9.875, -0.125) {};
		\node [style=none] (226) at (12.375, 0.75) {};
		\node [style=none] (227) at (10.625, 0.75) {};
		\node [style=lmat] (228) at (12.125, 0.25) {\(a_1\)};
		\node [style=lmat] (229) at (10.875, 0.25) {\(\partial\)};
		\node [style=none] (230) at (11.5, -0.125) {\(\vdotss\)};
		\node [style=none] (231) at (14.25, 0) {\(=\)};
		\node [style=none] (232) at (14.75, 1.125) {};
		\node [style=none] (233) at (14.75, -1) {};
		\node [style=none] (234) at (14.75, -0.125) {};
		\node [style=none] (253) at (16.25, 1.125) {};
		\node [style=none] (254) at (16.25, -1) {};
		\node [style=none] (255) at (16.25, -0.125) {};
	\end{pgfonlayer}
	\begin{pgfonlayer}{edgelayer}
		\draw [style=background] (34.center)
			 to (207.center)
			 to (210.center)
			 to [in=360, out=180] (37.center)
			 to cycle;
		\draw [style=background] (34.center) to (37.center);
		\draw (188.center) to (187.center);
		\draw (188.center)
			 to [in=180, out=60] (193.center)
			 to (195.center)
			 to [in=120, out=0] (189.center);
		\draw (189.center)
			 to [in=0, out=-120] (196.center)
			 to (194.center)
			 to [in=-60, out=180] (188.center);
		\draw (200.center)
			 to [in=180, out=75] (202.center)
			 to (203.center)
			 to [in=105, out=0] (201.center);
		\draw [style=background] (207.center) to (210.center);
		\draw (212.center) to (211.center);
		\draw (212.center)
			 to [in=0, out=120] (217.center)
			 to (219.center)
			 to [in=60, out=180] (213.center);
		\draw (213.center)
			 to [in=180, out=-60] (220.center)
			 to (218.center)
			 to [in=-120, out=0] (212.center);
		\draw (224.center)
			 to [in=0, out=105] (226.center)
			 to (227.center)
			 to [in=75, out=180] (225.center);
		\draw (201.center) to (225.center);
		\draw [style=background, in=360, out=180] (210.center) to (37.center);
		\draw [style=background] (34.center) to (207.center);
		\draw [style=background] (232.center)
			 to (253.center)
			 to (254.center)
			 to [in=360, out=180] (233.center)
			 to cycle;
		\draw [style=background] (232.center) to (233.center);
		\draw [style=background] (253.center) to (254.center);
		\draw [style=background, in=360, out=180] (254.center) to (233.center);
		\draw [style=background] (232.center) to (253.center);
		\draw (234.center) to (255.center);
	\end{pgfonlayer}
\end{tikzpicture},
\end{equation}
\begin{equation}\label{eq:signalflow:left-inverse}
  \begin{tikzpicture}
	\begin{pgfonlayer}{nodelayer}
		\node [style=none] (34) at (5.25, 1.125) {};
		\node [style=none] (37) at (5.25, -1) {};
		\node [style=none] (187) at (5.25, -0.125) {};
		\node [style=none] (188) at (5.875, -0.125) {};
		\node [style=none] (189) at (9.125, -0.125) {};
		\node [style=lmat] (191) at (6.875, 0.75) {\(a_0\)};
		\node [style=lmat] (192) at (6.875, -0.625) {\(a_n\)};
		\node [style=none] (193) at (6.625, 0.25) {};
		\node [style=none] (194) at (6.625, -0.625) {};
		\node [style=none] (195) at (8.375, 0.25) {};
		\node [style=none] (196) at (8.375, -0.625) {};
		\node [style=lmat] (197) at (8.125, -0.625) {\(\partial^n\)};
		\node [style=wn] (198) at (5.875, -0.125) {};
		\node [style=bn] (199) at (9.125, -0.125) {};
		\node [style=none] (200) at (5.875, -0.125) {};
		\node [style=none] (201) at (9.125, -0.125) {};
		\node [style=none] (202) at (6.625, 0.75) {};
		\node [style=none] (203) at (8.375, 0.75) {};
		\node [style=lmat] (204) at (6.875, 0.25) {\(a_1\)};
		\node [style=lmat] (205) at (8.125, 0.25) {\(\partial\)};
		\node [style=none] (206) at (7.5, -0.125) {\(\vdotss\)};
		\node [style=none] (207) at (13.75, 1.125) {};
		\node [style=none] (210) at (13.75, -1) {};
		\node [style=none] (211) at (13.75, -0.125) {};
		\node [style=none] (212) at (13.125, -0.125) {};
		\node [style=none] (213) at (9.875, -0.125) {};
		\node [style=rmat] (215) at (12.125, 0.75) {\(a_0\)};
		\node [style=rmat] (216) at (12.125, -0.625) {\(a_n\)};
		\node [style=none] (217) at (12.375, 0.25) {};
		\node [style=none] (218) at (12.375, -0.625) {};
		\node [style=none] (219) at (10.625, 0.25) {};
		\node [style=none] (220) at (10.625, -0.625) {};
		\node [style=rmat] (221) at (10.875, -0.625) {\(\partial^n\)};
		\node [style=wn] (222) at (13.125, -0.125) {};
		\node [style=bn] (223) at (9.875, -0.125) {};
		\node [style=none] (224) at (13.125, -0.125) {};
		\node [style=none] (225) at (9.875, -0.125) {};
		\node [style=none] (226) at (12.375, 0.75) {};
		\node [style=none] (227) at (10.625, 0.75) {};
		\node [style=rmat] (228) at (12.125, 0.25) {\(a_1\)};
		\node [style=rmat] (229) at (10.875, 0.25) {\(\partial\)};
		\node [style=none] (230) at (11.5, -0.125) {\(\vdotss\)};
		\node [style=none] (231) at (14.25, 0) {\(=\)};
		\node [style=none] (232) at (14.75, 1.125) {};
		\node [style=none] (233) at (14.75, -1) {};
		\node [style=none] (234) at (14.75, -0.125) {};
		\node [style=none] (253) at (16.25, 1.125) {};
		\node [style=none] (254) at (16.25, -1) {};
		\node [style=none] (255) at (16.25, -0.125) {};
	\end{pgfonlayer}
	\begin{pgfonlayer}{edgelayer}
		\draw [style=background] (34.center)
			 to (207.center)
			 to (210.center)
			 to [in=0, out=180] (37.center)
			 to cycle;
		\draw [style=background] (34.center) to (37.center);
		\draw (188.center) to (187.center);
		\draw (188.center)
			 to [in=180, out=60] (193.center)
			 to (195.center)
			 to [in=120, out=0] (189.center);
		\draw (189.center)
			 to [in=0, out=-120] (196.center)
			 to (194.center)
			 to [in=-60, out=180] (188.center);
		\draw (200.center)
			 to [in=180, out=75] (202.center)
			 to (203.center)
			 to [in=105, out=0] (201.center);
		\draw [style=background] (207.center) to (210.center);
		\draw (212.center) to (211.center);
		\draw (212.center)
			 to [in=0, out=120] (217.center)
			 to (219.center)
			 to [in=60, out=180] (213.center);
		\draw (213.center)
			 to [in=180, out=-60] (220.center)
			 to (218.center)
			 to [in=-120, out=0] (212.center);
		\draw (224.center)
			 to [in=0, out=105] (226.center)
			 to (227.center)
			 to [in=75, out=180] (225.center);
		\draw (201.center) to (225.center);
		\draw [style=background, in=0, out=180] (210.center) to (37.center);
		\draw [style=background] (34.center) to (207.center);
		\draw [style=background] (232.center)
			 to (253.center)
			 to (254.center)
			 to [in=360, out=180] (233.center)
			 to cycle;
		\draw [style=background] (232.center) to (233.center);
		\draw [style=background] (253.center) to (254.center);
		\draw [style=background, in=360, out=180] (254.center) to (233.center);
		\draw [style=background] (232.center) to (253.center);
		\draw (234.center) to (255.center);
	\end{pgfonlayer}
\end{tikzpicture}.
\end{equation}
Where the left facing arrows denote the transposes of the right facing arrows. These equations mean that the relation which multiplies its inputs by a nonzero polynomial is inverse to its converse relation.

We can also interpret the ``free'' syntax for LTI systems in our categorical semantics for the time-domain so that the affine translations are interpreted as impulses at time step 0 and the linear subspaces are invariant with respect to time:
\begin{lemmaE}\label{lem:signalflow:unroll}
  Let \(\eta\colon\FAffRel \K \to \St_{\FLinRel \K}(\FAffRel \K)\) be the canonical functor.
  Define \(J\colon\FAffRel \K \to \Obs_\Z(\FAffRel \K)\) by \(J(V)_t\coloneqq V\) for all \(t\in\Z\), and on morphisms by \(J(\emptyset) \coloneqq [\{\emptyset\}_{j \in \J}]\) and
  \[
    J(A+\vb a)
    \coloneqq
    \Big[ \Big\{
      \smbigoplus_{t=\ell}^{r}
      \begin{matrix}
        A+\vb a & \text{if } t=0,\\
        A & \text{if } t\neq 0
  \end{matrix}
    \Big\}_{{[\ell,r]} \in \J}  \Big].
  \]
  There is a unique feedback functor
  \(
    \Unroll\colon\St_{\FLinRel \K}(\FAffRel \K)\to \Obs_\Z(\FAffRel \K)
  \)
  satisfying
  \(\Unroll\!\circ \eta = J\).
\end{lemmaE}
\begin{proofE}
  We proceed by using the universal property of the state construction (see \Cref{prop:state:state-free}).

  By inspection, \(J\) produces monotone nets and preserves identities.
  The only nontrivial part is to show that composition is preserved.
  Take \(f\in \FAffRel \K(\K^n,\K^m)\) and \(g\in \FAffRel \K(\K^m,\K^p)\).
  If \(f\) or \(g\) is empty then
  \(J(g\circ f) = [\{\emptyset\}_{j\in \J}]  =  J(g) \circ J(f)\).
  Similarly if \(g\circ f\) is empty, then \(J(g\circ f)\) is empty at every context \(j\in \J\), hence by \Cref{prop:obs:zero},
  \(J(g\circ f) = [\{\emptyset\}_{j\in \J}]  =  J(g) \circ J(f)\).
  Otherwise, suppose that we have nonempty affine relations
  \(A+\vb{a}\in \FAffRel \K(\K^n,\K^m)\) and \(B+\vb{b}\in \FAffRel \K(\K^m,\K^p)\),
  with nonempty composite \(C+\vb{c}\in \FAffRel \K(\K^n,\K^p)\). Then:
  {\allowdisplaybreaks
  \begin{align*}
    &J(C+\vb{c})\\
    =&
    \left[
      \left\{
        \bigoplus_{k=\ell}^r
        \left(
          \begin{matrix}
            C+\vb{c} & \text{ if } k=0\\
            C        & \text{ if } k\neq 0
          \end{matrix}
        \right)
      \right\}_{{[\ell,r]}\in \J}
    \right]
    \\
    =&
    \left[
      \left\{
        \bigoplus_{k=\ell}^r
        \left(
          \begin{matrix}
            C+\vb{c} & \text{ if } k=0\\
            B\circ A & \text{ if } k\neq 0
          \end{matrix}
        \right)
        \right\}_{{[\ell,r]}\in \J}
    \right]
    \\
    =&
    \left[
      \left\{
        \bigoplus_{k=\ell}^r
        \left(
          \begin{matrix}
            B+\vb{b} & \text{ if } k=0\\
            B        & \text{ if } k\neq 0
          \end{matrix}
        \right)
        \circ
        \left(
          \begin{matrix}
            A+\vb{a} & \text{ if } k=0\\
            A        & \text{ if } k\neq 0
          \end{matrix}
        \right)
      \right\}_{{[\ell,r]}\in\J}
    \right]
    \\
    =&
    \left[
      \left\{
        \bigoplus_{k=\ell}^r
        \left(
          \begin{matrix}
            B+\vb{b} & \text{ if } k=0\\
            B        & \text{ if } k\neq 0
          \end{matrix}
        \right)
      \right\}_{{[\ell,r]} \in \J}
    \right]\\
    &\;\;\circ
    \left[
      \left\{
        \bigoplus_{k=\ell}^r
        \left(
          \begin{matrix}
            A+\vb{a} & \text{ if } k=0\\
            A        & \text{ if } k\neq 0
          \end{matrix}
        \right)
      \right\}_{{[\ell,r]} \in \J}
    \right]
    \\
    =&
    J(B+\vb{b}) \circ J(A+\vb{a}).
  \end{align*}}

  Let \(J|_{\FLinRel\K}\) denote the restriction of \(J\) along the inclusion \(\FLinRel\K \rightarrowtail \FAffRel\K\).
  For shift-invariance, take \(L\in \FLinRel\K(V,W)\) and \([\ell,r]\in\J\); then
  \begin{align*}
  \bigl(\shftZ(J|_{\FLinRel\K}(L))\bigr)_{[\ell,r]}
  =(J(L))_{[\ell-1,r-1]}
  =\smbigoplus_{t=\ell-1}^{r-1}L
  =\smbigoplus_{t=\ell}^{r}L
  =(J(L))_{[\ell,r]}.
  \end{align*}

  Now apply \Cref{prop:state:state-free} with
  \begin{gather*}
  \Gamma\colon\FLinRel\K \rightarrowtail \FAffRel\K,\qquad
  \Delta = 1_{\FLinRel\K},\qquad
  \Gamma' = 1_{\Obs_\Z(\FAffRel\K)},\\
  \Delta'=\shftZ,\qquad
  F=J,\qquad
  F|_{\FLinRel\K}=J|_{\FLinRel\K}.
  \end{gather*}
  The condition \(\Gamma'\circ F|_{\FLinRel\K}=F\circ \Gamma\) is the statement that \(J|_{\FLinRel\K}\) is the restriction of \(J\),
  and the condition \(\Delta'\circ F|_{\FLinRel\K}=F|_{\FLinRel\K}\circ \Delta\) is exactly shift-invariance of \(J|_{\FLinRel\K}\).
  Hence there exists a unique feedback functor
  \[
  \Unroll\colon\St_{\FLinRel \K}(\FAffRel \K)\to \Obs_\Z(\FAffRel \K)
  \]
  such that \(\Unroll\!\circ \eta = J\). By construction, \(\Unroll\) agrees with \(J\) as defined.
\end{proofE}
It is natural to ask if invertibility equations \Cref{eq:signalflow:right-inverse,eq:signalflow:left-inverse} also hold in \(\Unroll(\St_{\FLinRel \K}(\FAffRel \K))\); this is only partly true:

\begin{propositionE}\label{prop:signalflow:lax-monomial}
  Given any polynomial \(z+a \in \poly\) for \(a\neq 0\): 
  \[
  \Unroll\Bigl(\begin{tikzpicture}
	\begin{pgfonlayer}{nodelayer}
		\node [style=none] (34) at (5.25, 1.125) {};
		\node [style=none] (37) at (5.25, -1) {};
		\node [style=none] (187) at (5.25, -0.125) {};
		\node [style=none] (188) at (5.875, -0.125) {};
		\node [style=none] (189) at (9.125, -0.125) {};
		\node [style=lmat] (191) at (6.875, 0.75) {\(a_0\)};
		\node [style=lmat] (192) at (6.875, -0.625) {\(a_n\)};
		\node [style=none] (193) at (6.625, 0.25) {};
		\node [style=none] (194) at (6.625, -0.625) {};
		\node [style=none] (195) at (8.375, 0.25) {};
		\node [style=none] (196) at (8.375, -0.625) {};
		\node [style=lmat] (197) at (8.125, -0.625) {\(\partial^n\)};
		\node [style=wn] (198) at (5.875, -0.125) {};
		\node [style=bn] (199) at (9.125, -0.125) {};
		\node [style=none] (200) at (5.875, -0.125) {};
		\node [style=none] (201) at (9.125, -0.125) {};
		\node [style=none] (202) at (6.625, 0.75) {};
		\node [style=none] (203) at (8.375, 0.75) {};
		\node [style=lmat] (204) at (6.875, 0.25) {\(a_1\)};
		\node [style=lmat] (205) at (8.125, 0.25) {\(\partial\)};
		\node [style=none] (206) at (7.5, -0.125) {\(\vdotss\)};
		\node [style=none] (207) at (13.75, 1.125) {};
		\node [style=none] (210) at (13.75, -1) {};
		\node [style=none] (211) at (13.75, -0.125) {};
		\node [style=none] (212) at (13.125, -0.125) {};
		\node [style=none] (213) at (9.875, -0.125) {};
		\node [style=rmat] (215) at (12.125, 0.75) {\(a_0\)};
		\node [style=rmat] (216) at (12.125, -0.625) {\(a_n\)};
		\node [style=none] (217) at (12.375, 0.25) {};
		\node [style=none] (218) at (12.375, -0.625) {};
		\node [style=none] (219) at (10.625, 0.25) {};
		\node [style=none] (220) at (10.625, -0.625) {};
		\node [style=rmat] (221) at (10.875, -0.625) {\(\partial^n\)};
		\node [style=wn] (222) at (13.125, -0.125) {};
		\node [style=bn] (223) at (9.875, -0.125) {};
		\node [style=none] (224) at (13.125, -0.125) {};
		\node [style=none] (225) at (9.875, -0.125) {};
		\node [style=none] (226) at (12.375, 0.75) {};
		\node [style=none] (227) at (10.625, 0.75) {};
		\node [style=rmat] (228) at (12.125, 0.25) {\(a_1\)};
		\node [style=rmat] (229) at (10.875, 0.25) {\(\partial\)};
		\node [style=none] (230) at (11.5, -0.125) {\(\vdotss\)};
	\end{pgfonlayer}
	\begin{pgfonlayer}{edgelayer}
		\draw [style=background] (34.center)
			 to (207.center)
			 to (210.center)
			 to [in=0, out=180] (37.center)
			 to cycle;
		\draw [style=background] (34.center) to (37.center);
		\draw (188.center) to (187.center);
		\draw (188.center)
			 to [in=180, out=60] (193.center)
			 to (195.center)
			 to [in=120, out=0] (189.center);
		\draw (189.center)
			 to [in=0, out=-120] (196.center)
			 to (194.center)
			 to [in=-60, out=180] (188.center);
		\draw (200.center)
			 to [in=180, out=75] (202.center)
			 to (203.center)
			 to [in=105, out=0] (201.center);
		\draw [style=background] (207.center) to (210.center);
		\draw (212.center) to (211.center);
		\draw (212.center)
			 to [in=0, out=120] (217.center)
			 to (219.center)
			 to [in=60, out=180] (213.center);
		\draw (213.center)
			 to [in=180, out=-60] (220.center)
			 to (218.center)
			 to [in=-120, out=0] (212.center);
		\draw (224.center)
			 to [in=0, out=105] (226.center)
			 to (227.center)
			 to [in=75, out=180] (225.center);
		\draw (201.center) to (225.center);
		\draw [style=background, in=0, out=180] (210.center) to (37.center);
		\draw [style=background] (34.center) to (207.center);
	\end{pgfonlayer}
\end{tikzpicture}\Bigr)
  =
  \Unroll\Bigl(\begin{tikzpicture}
	\begin{pgfonlayer}{nodelayer}
		\node [style=none] (232) at (14.75, 1.125) {};
		\node [style=none] (233) at (14.75, -1) {};
		\node [style=none] (234) at (14.75, -0.125) {};
		\node [style=none] (253) at (16.25, 1.125) {};
		\node [style=none] (254) at (16.25, -1) {};
		\node [style=none] (255) at (16.25, -0.125) {};
	\end{pgfonlayer}
	\begin{pgfonlayer}{edgelayer}
		\draw [style=background] (232.center)
			 to (253.center)
			 to (254.center)
			 to [in=360, out=180] (233.center)
			 to cycle;
		\draw [style=background] (232.center) to (233.center);
		\draw [style=background] (253.center) to (254.center);
		\draw [style=background, in=360, out=180] (254.center) to (233.center);
		\draw [style=background] (232.center) to (253.center);
		\draw (234.center) to (255.center);
	\end{pgfonlayer}
\end{tikzpicture}\Bigr)
  \subsetneq
  \Unroll\Bigl(\begin{tikzpicture}
	\begin{pgfonlayer}{nodelayer}
		\node [style=none] (34) at (5.25, 1.125) {};
		\node [style=none] (37) at (5.25, -1) {};
		\node [style=none] (187) at (5.25, -0.125) {};
		\node [style=none] (188) at (5.875, -0.125) {};
		\node [style=none] (189) at (9.125, -0.125) {};
		\node [style=rmat] (191) at (6.875, 0.75) {\(a_0\)};
		\node [style=rmat] (192) at (6.875, -0.625) {\(a_n\)};
		\node [style=none] (193) at (6.625, 0.25) {};
		\node [style=none] (194) at (6.625, -0.625) {};
		\node [style=none] (195) at (8.375, 0.25) {};
		\node [style=none] (196) at (8.375, -0.625) {};
		\node [style=rmat] (197) at (8.125, -0.625) {\(\partial^n\)};
		\node [style=bn] (198) at (5.875, -0.125) {};
		\node [style=wn] (199) at (9.125, -0.125) {};
		\node [style=none] (200) at (5.875, -0.125) {};
		\node [style=none] (201) at (9.125, -0.125) {};
		\node [style=none] (202) at (6.625, 0.75) {};
		\node [style=none] (203) at (8.375, 0.75) {};
		\node [style=rmat] (204) at (6.875, 0.25) {\(a_1\)};
		\node [style=rmat] (205) at (8.125, 0.25) {\(\partial\)};
		\node [style=none] (206) at (7.5, -0.125) {\(\vdotss\)};
		\node [style=none] (207) at (13.75, 1.125) {};
		\node [style=none] (210) at (13.75, -1) {};
		\node [style=none] (211) at (13.75, -0.125) {};
		\node [style=none] (212) at (13.125, -0.125) {};
		\node [style=none] (213) at (9.875, -0.125) {};
		\node [style=lmat] (215) at (12.125, 0.75) {\(a_0\)};
		\node [style=lmat] (216) at (12.125, -0.625) {\(a_n\)};
		\node [style=none] (217) at (12.375, 0.25) {};
		\node [style=none] (218) at (12.375, -0.625) {};
		\node [style=none] (219) at (10.625, 0.25) {};
		\node [style=none] (220) at (10.625, -0.625) {};
		\node [style=lmat] (221) at (10.875, -0.625) {\(\partial^n\)};
		\node [style=bn] (222) at (13.125, -0.125) {};
		\node [style=wn] (223) at (9.875, -0.125) {};
		\node [style=none] (224) at (13.125, -0.125) {};
		\node [style=none] (225) at (9.875, -0.125) {};
		\node [style=none] (226) at (12.375, 0.75) {};
		\node [style=none] (227) at (10.625, 0.75) {};
		\node [style=lmat] (228) at (12.125, 0.25) {\(a_1\)};
		\node [style=lmat] (229) at (10.875, 0.25) {\(\partial\)};
		\node [style=none] (230) at (11.5, -0.125) {\(\vdotss\)};
	\end{pgfonlayer}
	\begin{pgfonlayer}{edgelayer}
		\draw [style=background] (34.center)
			 to (207.center)
			 to (210.center)
			 to [in=360, out=180] (37.center)
			 to cycle;
		\draw [style=background] (34.center) to (37.center);
		\draw (188.center) to (187.center);
		\draw (188.center)
			 to [in=180, out=60] (193.center)
			 to (195.center)
			 to [in=120, out=0] (189.center);
		\draw (189.center)
			 to [in=0, out=-120] (196.center)
			 to (194.center)
			 to [in=-60, out=180] (188.center);
		\draw (200.center)
			 to [in=180, out=75] (202.center)
			 to (203.center)
			 to [in=105, out=0] (201.center);
		\draw [style=background] (207.center) to (210.center);
		\draw (212.center) to (211.center);
		\draw (212.center)
			 to [in=0, out=120] (217.center)
			 to (219.center)
			 to [in=60, out=180] (213.center);
		\draw (213.center)
			 to [in=180, out=-60] (220.center)
			 to (218.center)
			 to [in=-120, out=0] (212.center);
		\draw (224.center)
			 to [in=0, out=105] (226.center)
			 to (227.center)
			 to [in=75, out=180] (225.center);
		\draw (201.center) to (225.center);
		\draw [style=background, in=360, out=180] (210.center) to (37.center);
		\draw [style=background] (34.center) to (207.center);
	\end{pgfonlayer}
\end{tikzpicture}\Bigr)
  \]
\end{propositionE}
\begin{proofE}
  Given a  polynomial \(a(z)\in \poly\), define the following morphism in \(\St_{\FLinRel \K}(\FAffRel \K)\):
  \[\mult_{a(z)} \coloneqq \begin{tikzpicture}
	\begin{pgfonlayer}{nodelayer}
		\node [style=none] (34) at (5.25, 1.125) {};
		\node [style=none] (35) at (9.875, 1.125) {};
		\node [style=none] (36) at (9.875, -1) {};
		\node [style=none] (37) at (5.25, -1) {};
		\node [style=none] (187) at (5.25, -0.125) {};
		\node [style=none] (188) at (5.875, -0.125) {};
		\node [style=none] (189) at (9.125, -0.125) {};
		\node [style=none] (190) at (9.875, -0.125) {};
		\node [style=rmat] (191) at (6.875, 0.75) {\(a_0\)};
		\node [style=rmat] (192) at (6.875, -0.625) {\(a_n\)};
		\node [style=none] (193) at (6.625, 0.25) {};
		\node [style=none] (194) at (6.625, -0.625) {};
		\node [style=none] (195) at (8.375, 0.25) {};
		\node [style=none] (196) at (8.375, -0.625) {};
		\node [style=rmat] (197) at (8.125, -0.625) {\(\partial^n\)};
		\node [style=bn] (198) at (5.875, -0.125) {};
		\node [style=wn] (199) at (9.125, -0.125) {};
		\node [style=none] (200) at (5.875, -0.125) {};
		\node [style=none] (201) at (9.125, -0.125) {};
		\node [style=none] (202) at (6.625, 0.75) {};
		\node [style=none] (203) at (8.375, 0.75) {};
		\node [style=rmat] (204) at (6.875, 0.25) {\(a_1\)};
		\node [style=rmat] (205) at (8.125, 0.25) {\(\partial\)};
		\node [style=none] (206) at (7.5, -0.125) {\(\vdotss\)};
	\end{pgfonlayer}
	\begin{pgfonlayer}{edgelayer}
		\draw [style=background] (36.center)
			 to (35.center)
			 to (34.center)
			 to (37.center)
			 to cycle;
		\draw (188.center) to (187.center);
		\draw (189.center) to (190.center);
		\draw (188.center)
			 to [in=180, out=60] (193.center)
			 to (195.center)
			 to [in=120, out=0] (189.center);
		\draw (189.center)
			 to [in=0, out=-120] (196.center)
			 to (194.center)
			 to [in=-60, out=180] (188.center);
		\draw (200.center)
			 to [in=180, out=75] (202.center)
			 to (203.center)
			 to [in=105, out=0] (201.center);
	\end{pgfonlayer}
\end{tikzpicture}\, .\]

  Take some nonzero \(a\in \K\). We need to show that \(\ZZ^*(\FAffRel{\fls})\):
  \begin{enumerate}
    \claimitem{signalflow:lax-monomial:left}{satisfies \Cref{eq:signalflow:left-inverse};}
    \claimitem{signalflow:lax-monomial:right}{only satisfies \Cref{eq:signalflow:right-inverse} oplaxly.}
  \end{enumerate}

  \claimparagraph{signalflow:lax-monomial:left}{Left inverse.}
  First, observe that
  \[\begin{tikzpicture}
	\begin{pgfonlayer}{nodelayer}
		\node [style=none] (0) at (0, 2.625) {};
		\node [style=none] (1) at (4.75, 2.625) {};
		\node [style=none] (2) at (0, 1.625) {};
		\node [style=none] (3) at (4.75, 1.625) {};
		\node [style=bn] (4) at (0.75, 2.625) {};
		\node [style=wn] (5) at (2.25, 1.625) {};
		\node [style=rmat] (6) at (1.5, 2.125) {\(a\)};
		\node [style=bn] (7) at (0.75, 1.625) {};
		\node [style=wn] (8) at (2.25, 0.625) {};
		\node [style=rmat] (9) at (1.5, 1.125) {\(a\)};
		\node [style=wn] (10) at (2.25, 2.625) {};
		\node [style=bn] (11) at (1.75, 3.125) {};
		\node [style=none] (12) at (0, 0.625) {};
		\node [style=none] (13) at (4.75, 0.625) {};
		\node [style=none] (14) at (0, -1.625) {};
		\node [style=none] (15) at (4.75, -1.625) {};
		\node [style=bn] (16) at (1, -1.625) {};
		\node [style=wn] (17) at (2.5, -2.625) {};
		\node [style=rmat] (18) at (1.75, -2.125) {\(a\)};
		\node [style=none] (19) at (0, -2.625) {};
		\node [style=none] (20) at (4.75, -2.625) {};
		\node [style=none] (23) at (1.125, -0.625) {};
		\node [style=wn] (24) at (2.5, -1.625) {};
		\node [style=rmat] (25) at (1.75, -1.125) {\(a\)};
		\node [style=bn] (26) at (0.75, 0.625) {};
		\node [style=none] (27) at (1.125, 0.125) {};
		\node [style=bn] (28) at (3.25, -2.625) {};
		\node [style=bn] (29) at (4, -3.125) {};
		\node [style=none] (30) at (0, 3.625) {};
		\node [style=none] (31) at (4.675, 3.625) {};
		\node [style=none] (32) at (4.75, -3.625) {};
		\node [style=none] (33) at (0, -3.625) {};
		\node [style=none] (34) at (1.125, -0.325) {\(\vdotss\)};
		\node [style=wirelable] (35) at (1.75, -0.325) {\(r-\ell\)};
		\node [style=none] (37) at (8.75, 2.625) {};
		\node [style=none] (38) at (5.75, 1.625) {};
		\node [style=none] (39) at (8.75, 1.625) {};
		\node [style=none] (40) at (5.75, 2.625) {};
		\node [style=wn] (41) at (8, 1.625) {};
		\node [style=rmat] (42) at (7.25, 2.125) {\(a\)};
		\node [style=bn] (43) at (6.5, 1.625) {};
		\node [style=wn] (44) at (8, 0.625) {};
		\node [style=rmat] (45) at (7.25, 1.125) {\(a\)};
		\node [style=bn] (46) at (8, 2.625) {};
		\node [style=none] (48) at (5.75, 0.625) {};
		\node [style=none] (49) at (8.75, 0.625) {};
		\node [style=none] (50) at (5.75, -1.625) {};
		\node [style=none] (51) at (8.75, -1.625) {};
		\node [style=bn] (52) at (6.75, -1.625) {};
		\node [style=wn] (53) at (8.25, -2.625) {};
		\node [style=rmat] (54) at (7.5, -2.125) {\(a\)};
		\node [style=none] (55) at (5.75, -2.625) {};
		\node [style=none] (56) at (8.75, -2.625) {};
		\node [style=none] (57) at (6.875, -0.625) {};
		\node [style=wn] (58) at (8.25, -1.625) {};
		\node [style=rmat] (59) at (7.5, -1.125) {\(a\)};
		\node [style=bn] (60) at (6.5, 0.625) {};
		\node [style=none] (61) at (6.875, 0.125) {};
		\node [style=none] (64) at (5.75, 3.625) {};
		\node [style=none] (65) at (8.675, 3.625) {};
		\node [style=none] (66) at (8.75, -3.625) {};
		\node [style=none] (67) at (5.75, -3.625) {};
		\node [style=none] (68) at (6.875, -0.325) {\(\vdotss\)};
		\node [style=wirelable] (69) at (7.5, -0.325) {\(r-\ell\)};
		\node [style=none] (70) at (5.25, 0) {\(=\)};
		\node [style=none] (71) at (-4, 0) {\(\Unroll(\mult_{z+a})_{[\ell,r]}\)};
		\node [style=none] (72) at (-0.75, 0) {\(=\)};
	\end{pgfonlayer}
	\begin{pgfonlayer}{edgelayer}
		\draw [style=background] (32.center)
			 to (33.center)
			 to (30.center)
			 to (31.center)
			 to cycle;
		\draw (0.center) to (1.center);
		\draw (2.center) to (3.center);
		\draw [in=-180, out=-60, looseness=1.50] (4) to (6);
		\draw [in=135, out=0, looseness=1.50] (6) to (5);
		\draw [in=-180, out=-60, looseness=1.50] (7) to (9);
		\draw [in=135, out=0, looseness=1.50] (9) to (8);
		\draw (11) to (10);
		\draw (12.center) to (13.center);
		\draw (14.center) to (15.center);
		\draw [in=-180, out=-60, looseness=1.50] (16) to (18);
		\draw [in=135, out=0, looseness=1.50] (18) to (17);
		\draw (19.center) to (20.center);
		\draw [in=-180, out=-90] (23.center) to (25);
		\draw [in=135, out=0] (25) to (24);
		\draw [in=90, out=-75] (26) to (27.center);
		\draw [in=180, out=-60] (28) to (29);
		\draw [style=background] (66.center)
			 to (67.center)
			 to (64.center)
			 to (65.center)
			 to cycle;
		\draw (38.center) to (39.center);
		\draw [in=-180, out=0, looseness=1.75] (40.center) to (42);
		\draw [in=135, out=0, looseness=1.50] (42) to (41);
		\draw [in=-180, out=-60, looseness=1.50] (43) to (45);
		\draw [in=135, out=0, looseness=1.50] (45) to (44);
		\draw (48.center) to (49.center);
		\draw (50.center) to (51.center);
		\draw [in=-180, out=-60, looseness=1.50] (52) to (54);
		\draw [in=135, out=0, looseness=1.50] (54) to (53);
		\draw (55.center) to (56.center);
		\draw [in=-180, out=-90] (57.center) to (59);
		\draw [in=135, out=0] (59) to (58);
		\draw [in=90, out=-75] (60) to (61.center);
		\draw (46) to (37.center);
	\end{pgfonlayer}
\end{tikzpicture}\,.\]
  Therefore
  {\allowdisplaybreaks
  \begin{align*}
    &(\Unroll(\mult_{z+a}) \circ \Unroll(\mult_{z+a})^*)_{[\ell,r]}\\
    &=\input{gaa/ladder_coladder_0.tikz}\, .
  \end{align*}
  }
  From which it follows that
  \[ (\Unroll(\mult_{z+a}) \circ \Unroll(\mult_{z+a})^*)_{[\ell,r]} \circ (\top_1\oplus 1_{r-\ell}) = \top_1 \oplus 1_{r-\ell}\,,\]
  so that
  \[\Unroll(\mult_{z+a}) \circ \Unroll(\mult_{z+a})^*= 1_{1^\Z}.\]

  \claimparagraph{signalflow:lax-monomial:right}{Right inverse.}
  On the other hand, fixing \(n \coloneqq r-\ell\),
  {\allowdisplaybreaks
  \begin{align*}
  &(\Unroll(\mult_{z+a})^* \circ \Unroll(\mult_{z+a}))_{[\ell,r]}\\
  &=\input{gaa/coladder_ladder_0.tikz}\\
  &={\overset{n}{\cdots}}=\input{gaa/coladder_ladder_1.tikz}\, .
  \end{align*}
  }

  Notice that the final inclusion is strict, which can be checked semantically as the bottom diagram has half the dimension of the one on top.

  Therefore
  \(\Unroll(\mult_{z+a})^* \circ \Unroll(\mult_{z+a})\supsetneq 1_{1^\Z}\).
\end{proofE}

In particular, this means, at least that for algebraically closed \(\K\),  all nonzero polynomials in \(\poly\) have left inverse in \(\Obs_\Z(\FAffRel {\K})\). However, only polynomials of the form \(az^n\) for nonzero \(a \in \K\) and \(n \in \N\) have right inverses, where for all other nonzero polynomials the right inverse equation holds oplaxly.
Nevertheless, in a \emph{relaxed} sense, both semantics are compatible:
\begin{propositionE}\label{prop:signalflow:lax-diagram}
  For any map \(f\) in \(\St_{\FLinRel \K}(\FAffRel \K)\),  we have that 
   \(\ZZ^*(\Ext(f)) \subseteq \Unroll(f)\).
\end{propositionE}
\begin{proofE}
  A simple diagram chase shows that the diagram commutes strictly on affine displacements, linear relations, and the delay.  Therefore, by oplaxness and normality, the desired inclusion holds.
\end{proofE}

The obstruction to invertibility is a consequence of the fact that the solutions to difference equations are determined only up to initial conditions. This can be resolved by working in \(\Obs_\N(\FAffRel{\K})\) and initialising the relevant monotone sequences to begin with the zero subspace.

\section{Conclusion}
\label{sec:conclusion}

Given any discard bicategory \(\CC\) and an upward-directed set of finite contexts \(\JJ \subseteq \Pfin(\I)\), we introduced the discard bicategory of observational behaviours \(\Obs_{\I,\JJ}(\CC)\), in which two processes have the same behaviour when every finite observation of one can be verified by an observation of the other (\Cref{def:obs:obs} and \Cref{thm:obs:obs}). We showed that our construction gives a functorial interpretation of stateful morphism sequences (\Cref{prop:obs:functorial}), propagates inconsistency globally (\Cref{thm:obsfbk:trace-semantics}), and carries a natural time delay (\Cref{prop:obsfbk:delay-natural-initialised} and \Cref{thm:obsfbk:obs-guarded-feedback}). Moreover, in the relational setting, we proved a compactness theorem (\Cref{thm:compactness:compactness}) gluing each compatible family of finite observations into a single infinite relation, yielding a compositional time-domain semantics for signal flow graphs that recovers Willems' behaviours over finite fields (\Cref{thm:signalflow:lti-compactness}).

Several questions remain open. A natural follow up would be generalising our compactness theorem from relations to probabilistic and quantum systems.
In the total probabilistic setting, this should follow from the existence of Kolmogorov products \cite[Def.~4.1]{fritz_infinite_2020}. For partial probabilistic systems, this would yield a generalisation of the Kolmogorov extension theorem to substochastic processes. 
In the quantum setting, Kretschmann and Werner interpret causal quantum channels with memory as infinite channels between quasi-local algebras \cite{kretschmann_quantum_2005}. It would be interesting to recover this result and generalise it to post-selected quantum processes via a quantised version of our compactness theorem.

A second direction for future work is the study of graphical languages for behaviours. We have shown that several intuitive axioms hold, but completeness is still an open problem. 
For uninitialised linear time-invariant behaviours, we expect that some of the axioms of Bonchi et al.~\cite{Bonchi2021} need to be relaxed to inequations.
In the quantum setting, our framework gives a semantics for the ZX calculus with a \emph{natural} delay. A complete axiomatisation in this setting would complement previous approaches to quantum dataflow programming \cite{delayedtracememory, felice_fusion_2024}. 

\section*{Acknowledgements}
This work has been partially funded by the European Union through the MSCA SE
project QCOMICAL and within the framework of \enquote{Plan France 2030}, under
the research projects EPIQ ANR-22-PETQ-0007 and HQI-R\&D ANR-22-PNCQ-0002.

We thank Vladimir Zamdzhiev for their useful feedback on an earlier draft of this
paper, and anonymous reviewers for their helpful comments and suggestions.

\bibliographystyle{plainurl}
\nocite{KatisSabadiniWalters2000_Published}
\bibliography{style/bibliography}

\appendix
\crefalias{section}{appendix}
\crefalias{subsection}{appendix}

\ifappendixproofs
  \section{Graphical affine algebra}\label{app:gaa}
  The following result is due to Bonchi et al.~\cite{gaa}; however, for convenience, we use the  ``spider'' notation, which is given in \cite{gsa}:
\begin{repthm}{thm:signalflow:gaa}
  Given a field \(\K\), the SMC \(\FAffRel{\K}\) is presented by the generators, for all \(a \in \K\):
  \begin{align*}
    \Big\llbracket\ \begin{tikzpicture}
	\begin{pgfonlayer}{nodelayer}
		\node [style=none] (19) at (-0.25, 0.5) {};
		\node [style=none] (20) at (1.25, 0.5) {};
		\node [style=none] (21) at (1.25, -0.5) {};
		\node [style=none] (22) at (-0.25, -0.5) {};
		\node [style=bn] (23) at (0.5, 0) {};
		\node [style=none] (24) at (1.25, 0.375) {};
		\node [style=none] (25) at (1.25, -0.375) {};
		\node [style=none] (26) at (-0.25, 0.375) {};
		\node [style=none] (27) at (-0.25, -0.375) {};
		\node [style=none] (28) at (-0.025, 0) {\(\vdotss\)};
		\node [style=none] (37) at (1, 0) {\(\vdotss\)};
	\end{pgfonlayer}
	\begin{pgfonlayer}{edgelayer}
		\draw [style=background] (21.center)
			 to (20.center)
			 to (19.center)
			 to (22.center)
			 to cycle;
		\draw [in=60, out=180] (24.center) to (23);
		\draw [in=180, out=-60] (23) to (25.center);
		\draw [in=360, out=120] (23) to (26.center);
		\draw [in=-120, out=0] (27.center) to (23);
	\end{pgfonlayer}
\end{tikzpicture}\ \Big\rrbracket
    &\coloneqq \{ (\vb v, \vb w) \ | \ \forall j, k: v_j = w_k\} ,\\
    \Big\llbracket\ \begin{tikzpicture}
	\begin{pgfonlayer}{nodelayer}
		\node [style=none] (19) at (-0.25, 0.5) {};
		\node [style=none] (20) at (1.25, 0.5) {};
		\node [style=none] (21) at (1.25, -0.5) {};
		\node [style=none] (22) at (-0.25, -0.5) {};
		\node [style=wn] (23) at (0.5, 0) {};
		\node [style=none] (24) at (1.25, 0.375) {};
		\node [style=none] (25) at (1.25, -0.375) {};
		\node [style=none] (26) at (-0.25, 0.375) {};
		\node [style=none] (27) at (-0.25, -0.375) {};
		\node [style=none] (28) at (-0.025, 0) {\(\vdotss\)};
		\node [style=none] (37) at (1, 0) {\(\vdotss\)};
		\node [style=wphase] (38) at (0.5, 0.75) {\(a\)};
	\end{pgfonlayer}
	\begin{pgfonlayer}{edgelayer}
		\draw [style=background] (21.center)
			 to (20.center)
			 to (19.center)
			 to (22.center)
			 to cycle;
		\draw [in=60, out=180] (24.center) to (23);
		\draw [in=180, out=-60] (23) to (25.center);
		\draw [in=360, out=120] (23) to (26.center);
		\draw [in=-120, out=0] (27.center) to (23);
		\draw [style=dwphase] (38) to (23);
	\end{pgfonlayer}
\end{tikzpicture}\ \Big\rrbracket
    &\coloneqq \{(\vb v, \vb w) \ |\ a+ \Sigma\ v_j = \Sigma\ w_j  \},\\
    \Big\llbracket \ \begin{tikzpicture}
	\begin{pgfonlayer}{nodelayer}
		\node [style=none] (0) at (-0.75, 0.375) {};
		\node [style=none] (1) at (0.75, 0.375) {};
		\node [style=none] (2) at (0.75, -0.375) {};
		\node [style=none] (3) at (-0.75, -0.375) {};
		\node [style=rmat] (4) at (0, 0) {\(a\)};
		\node [style=none] (5) at (-0.75, 0) {};
		\node [style=none] (6) at (0.75, 0) {};
	\end{pgfonlayer}
	\begin{pgfonlayer}{edgelayer}
		\draw [style=background] (2.center)
			 to (1.center)
			 to (0.center)
			 to (3.center)
			 to cycle;
		\draw (5.center) to (6.center);
	\end{pgfonlayer}
\end{tikzpicture} \ \Big\rrbracket
    &\coloneqq \{ (b, ab) \ | \ \forall b \in \K \},
  \end{align*}
  modulo the following equations, for all \(a,b\in \K\) and nonzero \(c \in \K\):
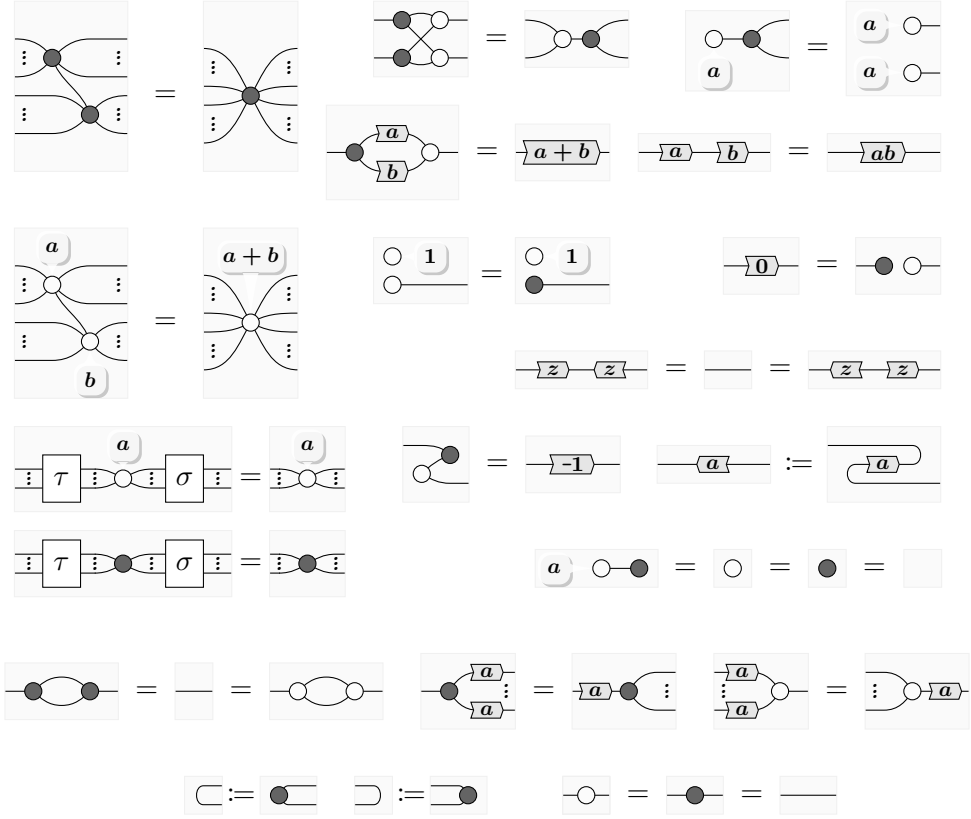
  \[ \input{gaa/axioms.tikz}\]
\end{repthm}

By reflecting the bicategory structure of \(\FAffRel{\K}\) into its diagrammatic presentation, we find that the following equation holds, which we shall use:
\begin{lemma}
\[\begin{tikzpicture}
	\begin{pgfonlayer}{nodelayer}
		\node [style=none] (47) at (6.25, 0.75) {};
		\node [style=none] (48) at (8.5, 0.75) {};
		\node [style=none] (49) at (8.5, -0.75) {};
		\node [style=none] (50) at (6.25, -0.75) {};
		\node [style=wn] (52) at (7.75, 0) {};
		\node [style=none] (53) at (8.5, 0.5) {};
		\node [style=none] (56) at (8.5, -0.5) {};
		\node [style=wn] (57) at (7, 0) {};
		\node [style=none] (58) at (6.25, 0.5) {};
		\node [style=none] (59) at (6.25, -0.5) {};
		\node [style=none] (60) at (9.5, 0.75) {};
		\node [style=none] (61) at (11.75, 0.75) {};
		\node [style=none] (62) at (11.75, -0.75) {};
		\node [style=none] (63) at (9.5, -0.75) {};
		\node [style=none] (65) at (11.75, 0.5) {};
		\node [style=none] (66) at (11.75, -0.5) {};
		\node [style=none] (68) at (9.5, 0.5) {};
		\node [style=none] (69) at (9.5, -0.5) {};
		\node [style=none] (70) at (9, 0) {\(\subset\)};
	\end{pgfonlayer}
	\begin{pgfonlayer}{edgelayer}
		\draw [style=background] (49.center)
			 to (50.center)
			 to (47.center)
			 to (48.center)
			 to cycle;
		\draw [in=-180, out=60] (52) to (53.center);
		\draw [in=-60, out=180] (56.center) to (52);
		\draw [in=0, out=120] (57) to (58.center);
		\draw [in=-120, out=0] (59.center) to (57);
		\draw (57) to (52);
		\draw [style=background] (62.center)
			 to (63.center)
			 to (60.center)
			 to (61.center)
			 to cycle;
		\draw (68.center) to (65.center);
		\draw (66.center) to (69.center);
	\end{pgfonlayer}
\end{tikzpicture}\]
\end{lemma}

This lemma can be checked semantically; moreover, it follows from the bicategorical axiomatisation of \(\rel(\FVect)\) of Paix\~ao et al.~\cite{Paixo2022}.

  \section{Rings of polynomials}\label{app:polynomials}
  In this Appendix we recall some basic algebraic notions involving rings of polynomials. In order to establish notation, we start with the following definition:

\begin{definition}
  Given a field \(\K\), the \textbf{polynomial ring} \(\poly\) is a principal ideal domain (PID).
\end{definition}

This can be extended to a larger PID by localising \(z\)

\begin{definition}
  The PID \(\poly[\K][z,1/z]\) given by adding an inverse to the polynomial indeterminate \(z\) in \(\poly\) is called the PID of \textbf{Laurent polynomials}  over \(\K\) in \(z\).
\end{definition}
\begin{lemma}
  There is an embedding of PIDs \[\poly \rightarrowtail \poly[\K][z,1/z].\]
\end{lemma}

The polynomial ring can also be extended to a larger PID:

\begin{definition}
  The PID of \textbf{formal power series} \(\fps\) has elements given by infinite sequences in \(\K\), denoted
  \[f(z)\coloneqq \sum_{j=0}^{\infty} f_j \cdot z^j \]
  The additive, Abelian group structure is given pointwise as for polynomials:
  \[
    {f}(z)+{g}(z) \coloneqq \sum_{j=0}^{\infty} (f_j+g_j) \cdot z^j,
  \]
  \[- {f}(z) \coloneqq \sum_{j=0}^{\infty} (-f_j)\cdot z^j  ,\quad\text{and}\quad
  0\coloneqq 0\cdot z^0.\]
  The multiplicative monoid structure is given by:
  \[
    {f}(z) {g}(z) \coloneqq \sum_{j=0}^{\infty}  \sum_{k=0}^{\infty} (f_j \cdot g_k) \cdot z^{j+k},
    \ \text{and}\
    1\coloneqq 1\cdot z^0
  \]
\end{definition}
\begin{lemma}
  There is an embedding of PIDs \(\poly \rightarrowtail \fps\).
\end{lemma}

The polynomial ring can be extended to a field:

\begin{definition}
  The field of fractions of polynomials or Laurent polynomials in \(z\),  \(\rat\) is called the field of \textbf{rational functions}  over \(\K\) in \(z\).
\end{definition}
\begin{lemma}
  There is an embedding of PIDs \(\poly[\K][z,1/z] \rightarrowtail \rat\).
\end{lemma}
By taking the \(z\)-adic completion, this can be extended to an even bigger field, which is a bit more involved to define:

\begin{definition}
  The field of \textbf{formal Laurent series} \(\fls\)  is the \(z\) adic completion of \(\poly\), \(\poly[\K][z,1/z]\), \(\fps\) or \(\rat\),  which is more complicated to describe concretely.  \(\fls\) has elements which are biinfinite sequences  \({f}(z)\in \K^\Z\) of elements in \(\K\), bounded below by some \(m\in \Z\), denoted:

  \[f(z)\coloneqq \sum_{j=m}^{\infty} f_j \cdot z^j \]

  The largest \(j\) such that for all \(k\in \N\) \(f_{j-k-1}=0\), if it exists,  is called the \textbf{order} of \( {f}(z)\), denoted \(\ord( {f}(z))\); otherwise, set \(\ord( {f}(z))\coloneqq 0\).

  The additive group structure and the multiplicative monoid structure are essentially the same as for formal power series.

  The multiplicative inverse is a bit more complicated. Take some nonzero \(f(z)\).
  Then we can factor out the lowest order term:
  \[
    {f}(z) = f_{\ord({f}(z))} z^{\ord({f}(z))} \left(1 + \sum_{j=1}^{\infty} \frac{f_{j+\ord({f}(z))}}{f_{\ord({f}(z))}} z^j \right)
  \]

  The inverse of \(1 + \sum_{j=1}^{\infty} \frac{f_{j+\ord({f}(z))}}{f_{\ord({f}(z))}} z^j \) is therefore a geometric series \({g}(z)\coloneqq \sum_{j=0}^{\infty} g_j\cdot z^j\) where the coefficients are given by the recursive formula:
  \[
    g_0 = 1, \quad\text{and for all }n\ge 0,\quad g_{n+1} = -\sum_{k=1}^{n+1} \frac{f_{k+\ord(f)}}{f_{\ord(f)}} g_{n+1-k}
  \]

  This gives us the inverse \({f}^{-1}(z) \coloneqq \frac{z^{-\ord({f}(z))}}{f_{\ord({f}(z))}} {g}(z)\).
\end{definition}

\begin{lemma}
  There is a field extension \(\gamma:\rat\rightarrowtail \fls\), called the \textbf{geometric series embedding}, sending \(f(z)/g(z) \in \rat \) to \(f(z) \cdot  g^{-1}(z) \in \fls\).
\end{lemma}

Putting this all together we have the lattice of PIDs:
\[\xymatrix@C=1em@R=1em{
  \K  \ar@{ >->}[r]  & \poly \ar@{ >->}[r] \ar@{ >->}[d] &  \poly[\K][z,1/z]  \ar@{ >->}[r]  &  \rat  \ar@{ >->}[d]\\
  & \fps\ar@{ >->}[rr] & & \fls
  }\]

  \section[Proofs of Section 2]{Proofs of \Cref{sec:approx}}
  \IfFileExists{\jobname-pratendapprox.tex}{\printProofs[approx]}{}

  \section[Proofs of Section 3]{Proofs of \Cref{sec:state}}
  \IfFileExists{\jobname-pratendstate.tex}{\printProofs[state]}{}

  \section[Proofs of Section 4]{Proofs of \Cref{sec:obs}}
  \IfFileExists{\jobname-pratendobs.tex}{\printProofs[obs]}{}

  \section[Proofs of Section 5]{Proofs of \Cref{sec:compactness}}
  \IfFileExists{\jobname-pratendcompactness.tex}{\printProofs[compactness]}{}

  \section[Proofs of Section 6]{Proofs of \Cref{sec:signalflow}}
  \IfFileExists{\jobname-pratendsignalflow.tex}{\printProofs[signalflow]}{}
\fi

\end{document}